%% file: main.tex
\PassOptionsToPackage{svgnames,dvipsnames}{xcolor}
\documentclass{lmcs}
\pdfoutput=1
\usepackage[utf8]{inputenc}

\usepackage{lastpage}
\lmcsdoi{19}{3}{9}
\lmcsheading{}{\pageref{LastPage}}{}{}%
{Oct.~18,~2022}{Aug.~02,~2023}{}

\newif\iftr\trtrue    
\input{styles}
\input{macros_util}
\input{macroIv}
\input{ggmacros}
\newcommand{\px}{{\ptp[x]}}

\DeclareMathOperator{\cat}{\!\cdot\!}
\mkfun{\ba}{ba}{\clang}

\newcommand{\intaut}[2]{\ifempty{#1}{\mathbf{A}_{#2}}{[\mathbf{A}_{#2}]^{#1}}}


\begin{document}
\title{A Theory of Formal Choreographic Languages}
\titlecomment{{\lsuper*}This paper is a revised and extended version of~\cite{blt22}.}
\thanks{
  \tnxbehapi. \tnxitmatters. The first and second authors have
	 also been partially supported by INdAM as members of GNCS
	 (Gruppo Nazionale per il Calcolo Scientifico).  The first author
	 has also been partially supported by Project "National Center for "HPC, Big Data e Quantum Computing",  Programma M4C2 – dalla ricerca all’impresa – Investimento 1.3: Creazione di “Partenariati estesi alle università, ai centri di ricerca, alle aziende per il finanziamento di progetti di ricerca di base”   – Next Generation EU. The authors thank the anonymous reviewers for their
	 helpful comments, in particular one reviewer of a previous
	 submission for suggesting the relation with Galois
	 connections.
	 The last author acknowledges the support of the PRO3 MUR project
	 Software Quality, and PNRR MUR project VITALITY (ECS00000041),
	 Spoke 2 ASTRA - Advanced Space Technologies and Research Alliance.
	 \\
	 We are grateful to Rolf Hennicker for flagging an error in an example (also spotted by a reviewer).
	 The authors also thank Mariangiola Dezani-Ciancaglini for her
	 support.  }

\author[F.~Barbanera]{Franco Barbanera\lmcsorcid{0000-0002-8039-1085}}[a]
\author[I.~Lanese]{Ivan Lanese\lmcsorcid{0000-0003-2527-9995}}[b]
\author[E.~Tuosto]{Emilio Tuosto\lmcsorcid{0000-0002-7032-3281}}[c]

\address{
  Dept. of Mathematics and Computer Science, University of Catania (Italy)
}
\email{barba@dmi.unict.it}
\address{
  Focus Team, University of Bologna/INRIA (Italy)
}
\email{ivan.lanese@gmail.com}
\address{
  Gran Sasso Science Institute (Italy)
}
\email{emilio.tuosto@gssi.it}

\begin{abstract}
  We introduce a meta-model based on formal languages, dubbed
  \emph{formal choreographic languages}, to study message-passing
  systems.
  Our framework allows us to generalise standard constructions from
  the literature and to compare them.
  In particular, we consider notions such as global view, local view,
  and projections from the former to the latter.
  The correctness of local views projected from global views is
  characterised in terms of a closure property.
  We consider a number of communication properties --such as
  (dead)lock-freedom-- and give conditions on formal choreographic
  languages to guarantee them.
  Finally, we show how formal choreographic languages can capture
  existing formalisms; specifically we consider communicating
  finite-state machines, choreography automata, and multiparty session
  types.
  Notably, formal choreographic languages, differently from most
  approaches in the literature, can naturally model systems exhibiting
  non-regular behaviour.
\end{abstract}

\maketitle

\section{Introduction}\label{sec:intro}
\input{intro}

\section{Formal Choreographic Languages}\label{sec:c-lang}
\input{c_lang}

\section{Correctness and Completeness}\label{sec:ui}
\input{ui}

\section{Communication Properties}\label{sec:prop}
\input{properties}

\section{Communication Properties by Construction}\label{sec:proj-prop}
\input{proj-properties}

\section{CFSMs and Choreography Languages}\label{sec:cfsmfcl}
\input{cfsm}

\section{Choreography Automata and Choreography Languages}\label{sec:chor-automata}
\input{choreography-automata}

\subsection{Deciding CUI and Branch-Awareness}\label{sec:decidability}
\input{decidability}

\section{Global Types as Choreographic Languages}\label{sec:mpst}
\input{DS-globalTypes}

\section{Related Work}\label{sec:related}
\input{related}

\section{Concluding Remarks}\label{sec:conc}
\input{conc}


\bibliographystyle{alphaurl}
\bibliography{bib}

\end{document}


%% file: styles.tex
\usepackage{amsfonts}
\usepackage{xspace}
\usepackage{centernot}
\usepackage{multicol}
\usepackage{prooftree}

\usepackage{amssymb}
\iftr \usepackage[appendix=inline]{apxproof}
\else \usepackage[appendix=strip]{apxproof}
\fi

\usepackage{graphicx}                 

\usepackage{multirow}                 
\usepackage{thmtools, thm-restate, thm-autoref}  

\usepackage{xifthen}                
\usepackage{appendix}               
\usepackage{xcolor} 

\usepackage{stmaryrd}

\usepackage{tikz}                   
\usetikzlibrary{shadows,arrows,shapes,automata,positioning,decorations.markings}
\usepgflibrary{arrows.spaced}

\usepackage{mathtools} 
\usepackage{amsbsy} 

\theoremstyle{plain}
\newtheorem{theorem}[thm]{Theorem} 
\newtheorem{lemma}[thm]{Lemma}
\newtheorem{proposition}[thm]{Proposition}
\newtheorem{corollary}[thm]{Corollary}
\theoremstyle{definition}
\newtheorem{example}[thm]{Example}
\newtheorem{definition}[thm]{Definition}
\newtheorem{remark}[thm]{Remark}
\newtheorem{myfact}[thm]{Fact}

\DeclareMathSizes{10}{10}{6}{6}
\usepackage{hyperref}
\hypersetup{breaklinks=true}
\usepackage[capitalise]{cleveref}
\crefname{myfact}{Fact}{fact}
\crefname{lstlisting}{Algorithm}{Algorithm}
\crefname{definition}{Definition}{Definitions}
\crefname{theorem}{Theorem}{Theorems}
\crefname{corollary}{Corollary}{Corollaires}
\crefname{proposition}{Proposition}{Propositions}
\crefname{figure}{Figure}{Figures}
\crefname{example}{Example}{Examples}
\crefname{table}{Table}{Tables}
\crefformat{enumi}{condition~#2#1#3}

\numberwithin{figure}{section}

\usepackage[textwidth=\marginparwidth]{todonotes}


%% file: macros_util.tex
\newcommand{\tnxbehapi}[1][partly]{
  Research {#1} supported by the EU H2020 RISE programme under the
  Marie Sk{\l}odowska-Curie grant agreement No 778233}

\newcommand{\tnxitmatters}[1][Work partially funded]{
   #1 by MIUR project PRIN 2017FTXR7S \emph{IT MATTERS}
  (Methods and Tools for Trustworthy Smart Systems)}

\usepackage{xargs}
\usepackage[normalem]{ulem} 

\usepackage{xifthen}        
\newcommand{\ifempty}[3]{%
  \ifthenelse{\isempty{#1}}{#2}{#3}%
}

\newcommand{\sclang}[1][]{sc-language{#1}\xspace}

\DeclareGraphicsExtensions{%
  .png,.PNG,%
  .pdf,.PDF,%
  .jpg,.mps,.jpeg,.jbig2,.jb2,.JPG,.JPEG,.JBIG2,.JB2
}

\newif\ifemi

\usepackage[most]{tcolorbox}
\newtcolorbox{markbox}{
  enhanced,
  breakable,
  size=minimal,
  parbox=false,
  after upper={},
  before upper={},
  colback=white,
  overlay = {
	 \draw[line width=2pt]
	 (frame.north east) -| ([xshift=3mm]frame.east) |-(frame.south east);
  },
  overlay first={\draw[line width=2pt] (frame.north east) -| ([xshift=3mm]frame.south east);},
  overlay middle={\draw[line width=2pt] ([xshift=3mm]frame.north east) -- ([xshift=3mm]frame.south east);},
  overlay last={\draw[line width=2pt] ([xshift=3mm]frame.north east) |- (frame.south east);},
}

\newcommand{\eMcomm}[2][check]{
  \ifthenelse{\equal{#1}{new}}{{\color{magenta}#2}}{
    \todo[color=orange!20]{\tiny eM: \color{NavyBlue}#1}
    {\color{OliveGreen}{#2}}
  }
}
\definecolor{ao(english)}{rgb}{0.0, 0.5, 0.0}

\newcommand{\hidden}[1]{}
\newcommand{\hide}[1]{}

\newcommand{\cf}[2]{
  \fontsize{#1}{#1}{\selectfont{#2}}
}

\ifemi
\usepackage{showlabels}

\newcommand{\emi}[2]{
  \marginpar{\fcolorbox{red}{shadecolor}{\cf{#1}{{#2}}}}
}
\newcommand{\emic}[2]{\par
  \fcolorbox{red}{shadecolor}{\parbox{\linewidth}{ 
      \color{gray}
      \begin{description}
      \item[{\color{blue} #2}]{\sf #1}
      \end{description}}}
}
\else
\newcommand{\emi}[2]{}
\newcommand{\emic}[2]{}{}
\fi

\newcommand{\Set}[1]{\{\,#1\,\}}

\def\colorFun{\color{orange}}
\newcommand{\mkfun}[4][\colorFun]{
  \newcommand{#2}[1][#4]{
    {#1\textsf{#3}}
    \ifempty{##1}{}{
      ({##1})}
  }
}
\newcommand{\proofcase}[1]{\ \newline\noindent\fbox{\scriptsize #1}}
\mkfun{\head}{hd}{}
\mkfun{\tail}{tl}{}

\newcommand{\mkuop}[4][\colorFun]{
  \newcommand{#2}[1][#4]{
    {#1\textsf{#3}}
    \ifempty{##1}{}{
      \, {##1}}
  }
}

\newcommand{\sst}{\;\big|\;}

\newcommand{\conf}[1]{\ensuremath{\langle {#1} \rangle}}

\newcommand{\emptyword}{\varepsilon}
\newcommand{\acword}{w}
\newcommand{\aaword}{v}
\newcommand{\aword}{z}

\newcommand{\qqand}[1][and]{\qquad\text{#1}\qquad}
\newcommand{\qand}[1][and]{\quad\text{#1}\quad}

\newcommand{\bnfdef}{\ ::=\ }

\newcommand{\bnfmid}{\;\ \big|\ \;}
\mkfun{\pref}{pref}{}

\newcommand{\glanguages}{\mathbb{G}}
\newcommand{\csystems}{\mathbb{S}}
\mkfun{\cl}{cl}{}

\newcommand{\aka}{\text{a.k.a.}\xspace}

\newcommand{\squo}[1]{\lq {#1}\rq}
\newcommand{\quo}[1]{\lq\lq {#1}\rq\rq}
\def\finex{{\unskip\nobreak\hfil
\penalty50\hskip1em\null\nobreak\hfil$\diamond$
\parfillskip=0pt\finalhyphendemerits=0\endgraf}}

\newcommand{\rmkend}{\finex}

\definecolor{shadecolor}{rgb}{1,0.99,0.9}
\definecolor{bg}{rgb}{0.95,0.95,0.95}

\ExplSyntaxOn
\NewDocumentCommand{\ucgreek}{m}
 {
  \str_case:nn { #1 }
   {
    {A}{\mathrm{A}}
    {B}{\mathrm{B}}
    {C}{\Sigma}
    {D}{\Delta}
    {E}{\mathrm{E}}
    {F}{\Phi}
    {G}{\Gamma}
    {H}{\mathrm{H}}
    {I}{\mathrm{I}}
    {J}{\Theta}
    {K}{\mathrm{K}}
    {L}{\Lambda}
    {M}{\mathrm{M}}
    {N}{\mathrm{N}}
    {O}{\mathrm{O}}
    {P}{\Pi}
    {Q}{\mathrm{X}}
    {R}{\mathrm{P}}
    {S}{\Sigma}
    {T}{\mathrm{T}}
    {U}{\Upsilon}
    {W}{\Omega}
    {X}{\Xi}
    {Y}{\Psi}
    {Z}{\mathrm{Z}}
   }
 }
\NewDocumentCommand{\lcgreek}{m}
 {
  \str_case:nn { #1 }
   {
    {a}{\alpha}
    {b}{\beta}
    {c}{\varsigma}
    {d}{\delta}
    {e}{\varepsilon}
    {f}{\varphi}
    {g}{\gamma}
    {h}{\eta}
    {i}{\iota}
    {j}{\vartheta}
    {k}{\kappa}
    {l}{\lambda}
    {m}{\mu}
    {n}{\nu}
    {o}{o}
    {p}{\pi}
    {q}{\chi}
    {r}{\rho}
    {s}{\sigma}
    {t}{\tau}
    {u}{\upsilon}
    {w}{\omega}
    {x}{\xi}
    {y}{\psi}
    {z}{\zeta}
   }
 }
\ExplSyntaxOff


%% file: macroIv.tex
\newcommand{\sset}{\mathcal{S}}
\newcommandx{\intpar}[2][1={\aint},2={\aint[b]},usedefault=@]{{#1}\parallel{#2}}

\def \bfr {\begin{color}{magenta}} 
\def \efr {\end{color}}

\def \beM {\begin{color}{orange}} 
\def \eeM {\end{color}}

\newcommand{\restrict}[2]{
  \protect\ifempty{#1}{\_}{
	 \protect\StrCount{#1,}{,}[\l]
	 \protect\ifthenelse{\l = 1}{
		#1
	 }{
		\left(\protect\dolist{\cat}{#1}\right)
	 }
  }{\left|_{#2}\!\right.}
}
\renewcommand{\epsilon}{\varepsilon}

\newcommand{\Nat}{\mathbb{N}}

\newcommand{\chora}[1][A]{\mathbb{C}\mathrm{#1}}

\newcommand{\proj}[2]{#1{\downarrow}_{{\ptp[#2]}}}

\newcommand{\arro}[1]{\xrightarrow{#1}}

\newcommand{\ssem}[1]{{\llbracket #1 \rrbracket}}

\newcommand{\comm}{\sim}



%% file: ggmacros.tex
\usepackage{listings}
\usepackage{mfirstuc}
\usepackage{xstring}

\newif\ifcp
\cpfalse
\newcommand{\gname}[1][i]{\ifcp{\colorNode{\scriptstyle\textsf{#1}}}\else{}\fi}

\newif\ifguard
\guardfalse
\newcommand{\aguard}{\ifguard{\colorGuard \phi}\else{}\fi}

\def\colorGuard{\color{cyan}}
\def\colorPtp{\color{blue}}
\def\colorFun{\color{NavyBlue}}

\def\colorOp{\color{OliveGreen}}
\def\colorNode{\color{LightCoral}}
\def\colorR{\color{OliveGreen}}
\def\colorE{\color{orange}}

\def\colorMsg{\color{BrickRed}}

\def\colorLang{\color{Green}}
\newcommand{\fillcolor}{orange!5}

\newcommand{\msg}[1][m]{\mathsf{\colorMsg{#1}}}

\newcommand{\lset}{\mathcal{L}}

\newcommand{\ptp}[1][A]{\ensuremath{\mathsf{\colorPtp{\capitalisewords{#1}}}}}

\newcommand{\p}{\ptp}
\newcommand{\q}{{\ptp[B]}}

\newcommand{\aint}[1][a]{\lcgreek{#1}}
\newcommand{\aact}[1][a]{\mathsf{#1}}

\newcommand{\sndint}[1][{\gint[]}]{\mathrm{sdr}(#1)}
\newcommand{\rcvint}[1][{\gint[]}]{\mathrm{rcv}(#1)}
\mkfun{\msgof}{msg}{{\gint[]}}
\mkfun{\ptpof}{ptp}{{\gint[]}}
\newcommandx{\ggcommon}[3][1=\ptp,2={\aH},3={\aH'},usedefault=@]{f_{#1}}
\newcommandx{\opair}[2][1={\ae},2={\ae'},usedefault=@]{\conf{{#1},{#2}}}
\newcommandx{\hopair}[2][1={\aE},2={\aE'},usedefault=@]{\llparenthesis\, {#1},{#2}\, \rrparenthesis}
\newcommandx{\wf}[2][1={\aG},2={\aG'},usedefault=@]{wf({#1}, {#2})}
\newcommandx{\wb}[2][1={\aG},2={\aG'},usedefault=@]{wb({#1}, {#2})}
\newcommandx{\ws}[2][1={\aG},2={\aG'},usedefault=@]{ws({#1}, {#2})}

\newcommandx{\widx}[2][1={\aW},2={i},usedefault=@]{{#1}[{#2}]}
\newcommandx{\outop}[2][1=\gname,2={}]{{\colorOp{!}}^{{#1}{#2}}}
\newcommandx{\inop}[2][1=\gname,2={}]{{\colorOp{?}}^{{#1}{#2}}}
\newcommandx{\aout}[5][1=a,2=b,3={},4=m,5={},usedefault=@]{
  \scalebox{.8}{$\achan[#1][#2] \outop[{#3}] {\msg[#4]} {#5}$}
}
\newcommandx{\ain}[5][1={\p},2={\q},3=\gname,4=m,5={},usedefault=@]{
  \scalebox{.8}{$\achan[#1][#2] \inop[{#3}] {\msg[{#4}]}{#5}$}
}
\newcommandx{\adep}[1][1={}]{
  \conf{ \aout[@][@][@][@][{#1}], \ain[@][@][@][@][{#1}]}
}

\newcommandx{\hproj}[2][1=\aH, 2=\ptp, usedefault=@]{
  \ifempty{#1}{}{{#1}}\ifempty{#2}{}{{^{\scriptscriptstyle @{#2}}}}
}
\newcommandx{\eproj}[2][1=\aE,2=A, usedefault=@]{
  {{#1}}\ifempty{#2}{}{{^{\scriptscriptstyle @{{\ptp[#2]}}}}}
}

\tikzset{
  component/.style={
    draw,
    fill = white,
    minimum width = 1.5cm,
    minimum height = .5cm,
    drop shadow
  }
}
\tikzset{
  file/.style={
    thin,
    fill = blue!5,
    font = \tt\scriptsize,
    text width = .8cm,
    minimum width = 1.0cm,
    minimum height = .5cm,
    drop shadow
  }
}

\tikzset{
  dataflow/.style={
    thick,
    draw, ->, >=latex,
    dashed,
    OliveGreen
  }
}

\tikzset{
  pipeline/.style={
    thick,
    draw, ->, >=latex,
    double,
    red
  }
}

\tikzset{
  CA/.style={
    transform shape,
	 node distance = 1.9cm,
    every state/.style = {cnode},
	 every edge/.style = {carrow}
  }
}

\tikzset{
  pomsetcloud/.style={
    cloud,
	 cloud puffs=10,
	 cloud ignores aspect,
	 minimum height=.1cm,
	 minimum width=2cm,
	 fill=blue!10,
	 opacity=.5,
	 draw
  }
}

\newcommand{\apom}{r}

\newcommand{\alfab}{\Sigma}
\newcommand{\alfint}{\alfab_{\text{int}}}
\newcommand{\alfact}{\alfab_{\text{act}}}

\newcommand{\aR}[1][R]{{\colorR{#1}}}

\newcommand{\aConf}{s}

\newcommandx{\detM}[1][1=\aCM,usedefault=@]{\Delta({#1})}

\newcommand{\cssem}[1]{(\hspace{-2.2pt}|#1|\hspace{-2.2pt})}

\renewcommand{\lset}{\mathbf{\Lambda}}
\newcommand{\mlang}{\mathbf{L}}

\tikzset{
  cnode/.style={
    shape=circle,
    minimum size = 0mm,
    inner sep = 1pt,
    font=\tiny,
    draw
  },
  carrow/.style={
    ->,
    shorten >=1pt,
    >=stealth',
    auto,
    font=\scriptsize,
    draw,
    sloped
  }
}

\newcommandx{\choranimation}[2][1=1,2=1,usedefault=@]{
  \begin{overlayarea}{#1\linewidth}{#2\textheight}
    \begin{tikzpicture}[
      node distance=1cm and 2cm,
      scale=0.35,
      every node/.style={transform shape},
      font=\large
      ]
      \node [choreo, align=center] (global){Choreography $\aG$ \\ global viewpoint};
      \node [below= of global] (fake) {};
      \node<2-> (synctxt) at (9,0)  {\textcolor{NavyBlue}{\bf Synchrony}};
      \node [choreo, align=center, local, below =of fake] (typei) {$\aCM_1$ \\ Local viewpoint$_1$};
      \node [choreo, align=center, local, left =of typei] (type1) {$\aCM_i$ \\ Local viewpoint$_1$};
      \node [choreo, align=center, local, right=of typei] (typen) {$\aCM_n$ \\ Local viewpoint$_n$};
      \node<2-> (asynctxt) at (9,-4) {\textcolor{NavyBlue}{\bf Asynchrony}};
      \node<2-> [below=of type1] (fake1) {};
      \node<2-> [below=of typei] (fakei) {};
      \node<2-> [below=of typen] (faken) {};
      \node<3-> [process, below=of fake1] (proc1) {Component$_1$};
      \node<3-> [process, below=of fakei] (proci) {Component$_1$};
      \node<3-> [process, below=of faken] (procn) {Component$_n$};
      \node<4-> [process, right=of procn,xshift=4cm] (evolve1) {Component'$_1$};
      \node<4-> [process, right=of evolve1] (evolvei) {Component'$_i$};
      \node<4-> [process, right=of evolvei] (evolven) {Component'$_n$};
      \node<5-> [choreo, align=center, local, above=of evolve1, yshift=1.5cm] (t11) {New $\aCM'_1$ \\ Local viewpoint$_1$};
      \node<5-> [choreo, align=center, local, above=of evolvei, yshift=1.5cm] (t1i) {New $\aCM'_i$ \\ Local viewpoint$_i$};
      \node<5-> [choreo, align=center, local, above=of evolven, yshift=1.5cm] (t1n) {New $\aCM'_n$ \\ Local viewpoint$_n$};
      \node<6> [above=of t1i, yshift=1.5cm] (qm) {\Huge \textcolor{red}{ ??? }};
      \node<6-> [above=of t1i] (dummy) {};
      \node<7-> [choreo,align=center,above=of dummy,scale=.85] (global') {New choreography $\aG'$ \\ global viewpoint};
      \path<3-> [bigar,->,dashed,gray] (type1) edge[sloped,above] node {Validate} (proc1);
      \path<3-> [bigar,->,dashed,gray] (typei) edge[sloped,above] node {Validate} (proci);
      \path<3-> [bigar,->,dashed,gray] (typen) edge[sloped,above] node {Validate} (procn);
      \path<2-> [bigar] (global) edge[sloped,above] node {\color{OliveGreen}Project} (type1);
      \path<2-> [bigar] (global) edge[sloped,above] node {\color{OliveGreen}Project} (typei);
      \path<2-> [bigar] (global) edge[sloped,above] node {\color{OliveGreen}Project} (typen);
      \path<2-> [elli] (type1) -- (typei);
      \path<2-> [elli] (typei) -- (typen);
      \path<3-> [elli] (proc1) -- (proci);
      \path<3-> [elli] (proci) -- (procn);
      \path<4-> [bigar,blue,dotted] (procn) edge node [above] {evolve/replace/compose} (evolve1);      
      \path<4-> [elli] (evolve1) -- (evolvei);
      \path<4-> [elli] (evolvei) -- (evolven);
      \path<5-> [elli] (t11) -- (t1i);
      \path<5-> [elli] (t1i) -- (t1n);
      \path<5-> [bigar,->,dashed,gray] (evolve1) edge[sloped,above] node {Extract} (t11);
      \path<5-> [bigar,->,dashed,gray] (evolvei) edge[sloped,above] node {Extract} (t1i);
      \path<5-> [bigar,->,dashed,gray] (evolven) edge[sloped,above] node {Extract} (t1n);
      \path<7-> [bigar,-] (t11) -- (dummy);
      \path<7-> [bigar,->] (t1i) edge node[right,xshift=1em] {\color{OliveGreen}Synthesise} (global');
      \path<7-> [bigar,-] (t1n) -- (dummy);
    \end{tikzpicture}
  \end{overlayarea}
}

\newcommandx{\cm}[2][1=\ptp, 2=\aM]{{#2}_{#1}}
\newcommandx{\achan}[2][1=A,2=B,usedefault=@]{{\ptp[#1]}{\,}{\ptp[#2]}}
\newcommand{\ptpset}{\mathcal{\colorPtp{P}}}

\newcommand{\oact}{\outop[]}
\newcommand{\iact}{\inop[]}
\newcommand{\tset}{\to}

\newcommand{\RS}[1][]{\mathsf{R}({#1})}

\newcommand{\trans}[2][{}]{\,\xrightarrow{#2}_{#1}\,}

\newcommandx{\acfsmout}[3][1=A,2=B,3=m,usedefault=@]{\achan[{#1}][{#2}] \oact {\msg[{#3}]}}
\newcommandx{\acfsmin}[3][1=A,2=B,3=m,usedefault=@]{\achan[{#1}][{#2}] \iact {\msg[{#3}]}}
\newcommandx{\fsaout}[2][1={\p},2={},usedefault=@]{
  \ptp[#1] \ \outop[]\ \msg[{#2}]
}
\newcommandx{\fsain}[2][1={\p},2={},usedefault=@]{
  \ptp[#1] \ \inop[]\ \msg[{#2}]
}

\makeatletter
\newcommand{\linenumfontsize}{\@setfontsize{\linenumfontsize}{3pt}{3pt}}
\makeatother
\lstdefinelanguage{sys}{
	commentstyle=\color{Gray},
	morecomment=[s]{[}{]},
	keywords=[0]{system,of,do,end},	keywordstyle=\color{orange}\bfseries,
}

\lstdefinelanguage{sgg}{
  commentstyle=\color{Gray},
  morecomment=[l]{..},
  morecomment=[s]{[}{]},
  keywords=[0]{repeat,branch,sel},
  keywordstyle=\color{orange}\bfseries,
  morekeywords=[1]{*,\+,|,->},
  literate={->}{$\colorOp \xrightarrow$}1 {|}{$\gparop$}1 {;}{$\gseqop$}1 {+}{$\gchoop$}1 {\{}{{\textcolor{NavyBlue}{\{}}}1 {\}}{{\textcolor{NavyBlue}{\}}}}1
}

\newcommand{\aG}{\mathsf{G}}

\newcommand{\gseqop}{{\colorOp ;}\,}
\newcommand{\gparop}{{\colorOp \ |\ }}
\newcommand{\gchoop}{{\colorOp \ +\ }}
\newcommand{\grecop}{{\colorOp *}}
\newcommand{\grecopp}{{\colorOp{@}}}
\newcommandx{\nmerge}[2][1={i},2={},usedefault=@]{
  \ifempty{#2}{
    \ifempty{#1}{\mu}{\gname[-{#1}]}
  }{-{#2}}
}

\mkfun{\cuui}{cui}{\clang}
\mkfun{\esbj}{sbj}{\ae}

\makeatletter%
\@ifclassloaded{exam-paper}%
  {}%
  {\makeatletter%
    \@ifclassloaded{test}%
    {}%
    \makeatother%
  }
\makeatother%

\newcommandx{\gnode}[2][1=i,2=\gint,usedefault=@]{
  \ifcp{
    \ifempty{#1}{#2}{\gname[#1].\big({#2}\big)}
  }
  \else
  {#2}
  \fi
}

\newcommandx{\gint}[4][1=i,2=A,3=m,4=B,usedefault=@]{
  \scalebox{.8}{$
	 \ptp[#2] {\colorOp \xrightarrow{\scriptstyle \gname[#1]}} \ptp[#4] \! {\colorOp \colon} \! \!{\msg[{#3}]}
  $}
}
\newcommandx{\gout}[4][1=\gname,2=\ptp,3=m,4={\ptp[C]},usedefault=@]{
  \achan[{#2}][{#4}] {\colorOp {\colorOp{!}}} {\msg[{#3}]}
}
\newcommandx{\gin}[4][1=\gname,2=\ptp,3=m,4={\ptp[C]},usedefault=@]{
  \achan[{#2}][{#4}] {\colorOp {\colorOp{?}}} {\msg[{#3}]}
}
\newcommandx{\gseq}[3][1=\gname,2={\aG},3={\aG'},usedefault=@]{
  \def\ggraph{{#2} \gseqop {#3}}
  \ggraph
}
\newcommand{\ginfix}[4]{
  \def\ggraph{{#2} {#4} {#3}}
  \gnode[#1][\ggraph]
}
\newcommandx{\gpar}[3][1=i,2={\aG},3={\aG'},usedefault=@]{
  \ginfix{#1}{#2}{#3}{\gparop}
}
\newcommandx{\gcho}[3][1=i,2={\aG},3={\aG'},usedefault=@]{
  \ginfix{#1}{#2}{#3}{\gchoop}
}
\newcommandx{\gchov}[3][1=\gname,2={\aG},3={\aG'},usedefault=@]{
  \def\ggraph{\left(
  \begin{array}l
    \ifempty{#1}{{#2} \\ \gchoop \\ {#3}}{\!\!{#2} \\ \gchoop \\ {#3}}
  \end{array}\right)}
  \ifcp\gnode[{$#1$}][\ggraph] \else \ggraph \fi
}
\newcommandx{\grec}[3][1=i,2={\aG},3={\p},usedefault=@]{
  \def\ggraph{\grecop {#2} \ifempty{#3}{}{\grecopp {#3}}}
  \ifempty{#1}{\ggraph}{\gname[{$#1$}][\ggraph]}
}	

\newcommand{\getcentroid}[2]{
    \coordinate (tmpgatecoord) at (0,0);
    \foreach \n [count=\i] in {#1}{
      \path (\n);
      \coordinate (tmpgatecoord) at ($(tmpgatecoord) + (\n)$);
      \coordinate (#2) at ($1/\i*(tmpgatecoord)$);
    }
}

\tikzset{
  hgsem/.style={
    draw,
    node distance=2cm and 1cm,
    transform shape,
    smooth,
    every node/.style = {font=\sffamily\bfseries}
  }
}

\tikzset{
  hgstyle/.style={
    src color={#1},
    tgt color={#1},
    centroid color={#1},
    centroid label={#1},
    centroid name={#1},
    centroid radius={#1},
    centroid ratio={#1},
    xoffset={#1},
    yoffset={#1},
    xsrcoffset={#1},
    ysrcoffset={#1},
    xtgtoffset={#1},
    ytgtoffset={#1},
    font={#1},
    centroid angle={#1},
    centroid tolerance={#1}
  },
  src color/.store in = \hgsrccol,
  tgt color/.store in = \hgtgtcol,
  centroid color/.store in =\hgfillcolor,
  centroid label/.store in =\hglabel,
  centroid name/.store in =\hgname,
  centroid radius/.store in = \hgradius,
  centroid ratio/.store in = \hgratio,
  xoffset/.store in =\hgxoffset,
  yoffset/.store in =\hgyoffset,
  xsrcoffset/.store in =\hgxsrcoffset,
  ysrcoffset/.store in =\hgysrcoffset,
  xtgtoffset/.store in =\hgxtgtoffset,
  ytgtoffset/.store in =\hgytgtoffset,
  centroid angle/.store in =\hgangle,
  centroid tolerance/.store in =\hgtolerance,
  src color = black,
  tgt color = black,
  centroid color = orange!40,
  centroid label={},
  centroid name={dummycentroid},
  centroid radius = .7pt,
  centroid ratio = .35,
  xoffset = 0,
  yoffset = 0,
  xsrcoffset = 0,
  ysrcoffset = 0,
  xtgtoffset = 0,
  ytgtoffset = 0,
  font=\sffamily\scriptsize,
  centroid angle=0,
  centroid tolerance=10pt
}

\newcommandx{\mkhg}[5][1={},4={},5={},usedefault=@]{
  \begingroup
  \tikzset{#1}
  \StrCount{#2,}{,}[\l] 
  \StrCount{#3,}{,}[\m] 
  \ifthenelse{\l = 1 \AND \m = 1}{
    \ifempty{#4}{
      \ifempty{#5}{
        \path[hgsem, ->, >=stealth', shorten >=1pt] (#2) -- (#3);
      }{
        \path[hgsem, ->, >=stealth', shorten >=1pt] (#2) #5 (#3);
      }
    }{
      \ifempty{#5}{
        \path[hgsem, ->, >=stealth', shorten >=1pt, #4] (#2) -- (#3);
      }{
        \path[hgsem, ->, >=stealth', shorten >=1pt, #4] (#2) #5 (#3);
      }
    }
  }{
    \coordinate (srcoffset) at (\hgxsrcoffset,\hgysrcoffset);
    \coordinate (tgtoffset) at (\hgxtgtoffset,\hgytgtoffset);
    \getcentroid{#2}{srccentroid};
    \getcentroid{#3}{tgtcentroid};
    \node[label={left:\hglabel}] (\hgname) at ($(srccentroid)!{1-\hgratio}!\hgangle:(tgtcentroid) + (\hgxoffset,\hgyoffset)$) {};
    \pgfgetlastxy \xc \yc;
    \pgfmathtruncatemacro{\xcontrol}{\xc};
    \pgfmathtruncatemacro{\ycontrol}{\yc};
    \foreach \n in {#2}{
      \path (\n);
      \pgfgetlastxy \xntmp \yntmp;
      \pgfmathtruncatemacro{\xn}{\xntmp};
      \pgfmathtruncatemacro{\yn}{\yntmp};
      \pgfmathsetmacro\xtmpdiff{abs(\xn - \xcontrol + \hgxsrcoffset)};
      \pgfmathsetmacro\ytmpdiff{abs(\yn - \ycontrol + \hgytgtoffset)};
      \ifdim \xtmpdiff pt > \hgtolerance
      \ifempty{#4}{
        \path[hgsem, \hgsrccol] (\n) .. controls ($(srccentroid.center) + (srcoffset)$) .. (\hgname.center);
      }{
        \path[hgsem, \hgsrccol] (\n) .. controls ($(srccentroid.center) + (srcoffset)$) .. (\hgname.center);
      }
      \else
      \ifempty{#4}{
        \path[hgsem, \hgsrccol] (\n) -- (\hgname.center);
      }{
        \path[hgsem, \hgsrccol, #4] (\n) -- (\hgname.center);
      }
      \fi
    }
    \foreach \n in {#3}{
      \path (\n);
      \pgfgetlastxy \xntmp \yntmp;
      \pgfmathtruncatemacro{\xn}{\xntmp};
      \pgfmathtruncatemacro{\yn}{\yntmp};
      \pgfmathsetmacro\xtmpdiff{abs(\xn - \xcontrol)};
      \pgfmathsetmacro\ytmpdiff{abs(\yn - \ycontrol)};
      \ifdim \xtmpdiff pt > \hgtolerance
      \ifempty{#4}{
        \path[hgsem, ->, >=stealth', shorten >=1pt, \hgtgtcol] (\hgname.center) .. controls (tgtcentroid.center) and ($(tgtcentroid.center) + (tgtoffset)$) .. (\n);
      }{
        \path[hgsem, ->, >=stealth', shorten >=1pt, \hgtgtcol,#4] (\hgname.center) .. controls (tgtcentroid.center) and ($(tgtcentroid.center) + (tgtoffset)$) .. (\n);
      }
      \else
      \ifempty{#4}{
        \path[hgsem, ->, >=stealth', shorten >=1pt, \hgtgtcol] (\hgname.center) --  (\n);
      }{
        \path[hgsem, ->, >=stealth', shorten >=1pt, \hgtgtcol] (\hgname.center) --  (\n);
      }
      \fi
    }
    \fill[\hgfillcolor] (\hgname) circle [radius=\hgradius];
  }
  \endgroup
}

\newcommandx{\hgordeq}[1][1={\aH},usedefault=@]{\sqsubseteq_{#1}}
\newcommandx{\gintsem}[4][4=.5]{
  \tikz[hgsem,scale=#4,every node/.style={font=\scriptsize}]{
    \node (out) {$\aout[{#1}][{#2}][][{#3}]$};
    \node[below = 20pt of out] (in) {$\ain[{#1}][{#2}][][{#3}]$};
    \mkhg{out}{in};
  }
}

\newcommandx{\gsem}[2][1={\aG},2={},usedefault=@]{\left\llbracket {#1} \right\rrbracket_{#2}}

\newcommandx{\rbot}{\text{undef}}
\newcommandx{\rtrs}[1][1={\aH},usedefault=@]{{#1}^{\star}}
\newcommandx{\gord}[1][1={\aG},usedefault=@]{\leq_{#1}}
\newcommandx{\gordeq}[1][1={\aG},usedefault=@]{\leq_{#1}}
\mkfun{\cause}{cs}{}
\mkfun{\effect}{ef}{}

\newcommandx{\aW}{w}
\newcommand{\clang}{\mathcal{\colorLang L}}
\newcommand{\alang}{\mathbb{\colorLang L}}
\newcommand{\gfun}[1]{\ensuremath{\mathsf{\colorFun #1}}}
\mkfun{\eact}{\gfun{act}}{}
\mkfun{\enode}{\gfun{cp}}{}

\mkuop{\rmax}{\gfun{max}}{\aH}
\mkuop{\rmin}{\gfun{min}}{\aH}
\mkuop{\rMAX}{\gfun{lst}}{\aH}
\mkuop{\rMIN}{\gfun{fst}}{\aH}

\newcommandx{\rseq}[2][1=\aG,2={\aG'},usedefault=@]{\gfun{seq}({#1},{#2})}
\newcommandx{\rpar}[2][1=\aG,2={\aG'},usedefault=@]{\gfun{par}({#1},{#2})}

\newcommandx{\gproj}[2][1=\aG,2=\ptp]{{#1}\downarrow_{#2}}
\newcommandx{\cinit}[1][1={\aQzero},usedefault=@]{{#1}}
\newcommandx{\cfinal}[1][1={q_e},usedefault=@]{{#1}}

\newcommandx{\geproj}[4][1=\aG,2=\ptp,3=\cinit,4=\cfinal,usedefault=@]{
  {#1}\downarrow_{#2}^{{#3},{#4}}
}

\newcommand*{\StrikeThruDistance}{0.15cm}%
\tikzset{strike thru arrow/.style={
    decoration={markings, mark=at position 0.5 with {
        \draw [blue, thick,-] 
            ++ (-\StrikeThruDistance,-\StrikeThruDistance) 
            -- ( \StrikeThruDistance, \StrikeThruDistance);}
    },
    postaction={decorate},
}}

\newcommandx{\ich}[1][1={\aG},usedefault=@]{{#1}^{\oplus}}
\newcommandx{\ichedges}[2][1={\aG},2={\gname},usedefault=@]{{#1}^{\oplus}({#2})}
\newcommandx{\parts}[1]{2^{#1}}
\newcommandx{\actch}{c}
\newcommandx{\soundactch}[2][1={\aG},2={\actch},usedefault=@]{{#1} \,\circledR\, {#2}}
\newcommandx{\rOnActch}[2][1={\aG},2={\actch},usedefault=@]{{#1} \setminus {#2}}
\newcommandx{\rOnActchClean}[2][1={\aG},2={\actch},usedefault=@]{{#1} \circledR {#2}}
\newcommandx{\rAllEvents}[1][1={\aG},usedefault=@]{\mathit{dom}(#1)}

\newcommand{\AV}{\mathcal{V}}
\newcommand{\aH}{H}

\newcommandx{\hgvertex}[2][1=\al,2=\gname,usedefault=@]{{#1}_{\textcolor{red}{[{#2}]}}}
\newcommand{\aE}{{\colorE E}}
\renewcommand{\ae}[1][e]{{\colorE{#1}}}

\newcommand{\al}[1][l]{{\colorE{#1}}}
\newcommandx{\hyedge}[1]{\{#1\}}

\newcommandx{\rdiv}[2][1=\gcho,2=\ptp,usedefault=@]{
  \gfun{div}_{#2}(#1)
}

\newcommandx{\rrdiv}[5][1={\aG},2={\aG'},3={\AV},4={,\AV'},5=\ptp,usedefault=@]{
  \gfun{div}^{#3#4}_{#5}(#1,#2)
}
\newcommandx{\pdiv}[3][1={\apom_1},2={\apom_2},3={\apom},usedefault=@]{
  \gfun{div}_{#3}(#1,#2)
}
\newcommandx{\pfork}[3][1={\apom_1},2={\apom_2},3={\apom},usedefault=@]{
  \gfun{fork}_{#3}(#1,#2)
}

\tikzset{
  pomset/.style={
    node distance = .6cm and .6cm,
    scale = .7,
    transform shape,
    smooth
  }
}

\newcommandx{\mkint}[6][3=i,4=\p,5=\msg,6=\q,usedefault=@]{
  \node[bblock, #1] (#2) {$\gint[#3][#4][#5][#6]$};
}

\newcommand{\mkseq}[2]{\path[line] (#1) -- (#2);}

\newcommand{\mknseq}[1]{
  \StrCount{#1}{,}[\l] 
  \StrBefore{#1}{,}[\myhead]
  \StrBehind{#1}{,}[\mytail]
  \StrBefore{\mytail}{,}[\sndel]
  \ifnum \l > 1 {
    \mkseq{\myhead}{\sndel};
    \mknseq{\mytail}
  }
  \else{\ifnum \l > 0{
      \mkseq{\myhead}{\mytail};
    }
    \else{}
    \fi
  }
  \fi
}

\newcommandx{\mkgateblock}[6][6=yellow!10,usedefault=@]{
  \path(#2);
  \pgfgetlastxy{\xgate}{\ygate};
  \pgfmathtruncatemacro{\xgateround}{\xgate};
  \StrCount{#3,}{,}[\l] 
  \ifnum \l < 2 {\errmessage{#3 argument should be a comma-separated list of lenght >= 2}}
  \else{
    \foreach \n in {#3}{
      \path (\n);
      \pgfgetlastxy{\xnode}{\ynode};
      \pgfmathtruncatemacro{\xnround}{\xnode};
      \pgfmathsetmacro\tmpdiff{abs(\xnround - \xgateround)}
      \ifdim \tmpdiff pt > 1 pt \path[line] (#2) -| (\n);
      \else
        \path[line] (#2) -- (\n);
      \fi
    }
  }
  \fi
  \StrCount{#4,}{,}[\l] 
  \ifnum \l < 2 {\errmessage{#4 argument should be a comma-separated list of lenght >= 2}}
  \else{
    \foreach \n in {#4}{
      \path (\n);
      \pgfgetlastxy{\xnode}{\ynode};
      \pgfmathtruncatemacro{\xnround}{\xnode};
      \pgfmathsetmacro\tmpdiff{abs(\xnround - \xgateround)}
      \ifdim \tmpdiff pt > 1 pt \path[line] (\n) |- (#5);
      \else
        \path[line] (\n) -- (#5);
      \fi
    }
  }
  \fi
  \node[#1] at (#2) {};
  \node[#1] at (#5) {};
  {
    \begin{pgfonlayer}{background}
      \path[fill=#6,rounded corners]
      (current bounding box.south west) rectangle
      (current bounding box.north east);
    \end{pgfonlayer}
  }
}

\newcommandx{\mkbranchblock}[5][5=@]{
  \mkgateblock{ogate}{#1}{#2}{#3}{#4}[#5]
}

\newcommandx{\mkforkblock}[5][5=@]{
  \mkgateblock{agate}{#1}{#2}{#3}{#4}[#5]
}

\newcommandx{\mkgraph}[3][1=.5cm, usedefault=@]{
  \node[source,above = #1 of {#2}] (src#2) {};
  \node[sink,below  = #1 of {#3}] (sink#3) {};
  \path[line] (src#2) -- (#2);
  \path[line] (#3) -- (sink#3);
}

\newcommandx{\mkloop}[5][1=.5, 2=1.5, 5=\aguard, usedefault=@]{
  \node[lgate,above = #1 of {#3}] (entry#3) {};
  \pgfgetlastxy \xentry \yentry;
  \pgfmathtruncatemacro{\xentryrounded}{\xentry};
  \node[below = #1 of {#4}, label = {above right:{$#5$}},yshift=-1em] (dummy) {};
  \node[lgate,below  = #1 of {#4}] (exit#4) {};
  \pgfgetlastxy \xexit \yexit;
  \pgfmathtruncatemacro{\xexitrounded}{\xexit};
  \path[line] (entry#3) -- (#3);
  \path[line] (#4) -- (exit#4);
  \pgfmathsetmacro\tmpdiff{abs(\xentryrounded - \xexitrounded)}
  \path[line, color=teal] (exit#4) -| ($(exit#4)+(\tmpdiff,0)+(#2,0)$) |- (entry#3);
}

\newcommandx{\mkfork}[4][2=gatenode,3=i,4=.6,usedefault=@]{
  \mkgatebegin{#1}[{\gname[{#3}]}][agate][#4]{#2}
}

\newcommandx{\mkbranch}[4][2=gatenode,3=i,4=.6,usedefault=@]{
  \mkgatebegin{#1}[{\gname[{#3}]}][ogate][#4]{#2}
}

\newcommandx{\mkgatebegin}[5][2={},3=ogate,4=.5,usedefault=@]{
  \coordinate (gatecord) at (0,0);
  \coordinate (xmax) at (0,0);
  \coordinate (xmin) at (0,0);
  \pgfgetlastxy \xmin \xmax;
  \foreach \n [count=\i] in {#1}{
    \pgfgetlastxy \xc \yc;
    \path (\n);
    \pgfgetlastxy \xn \yn;
    \ifnum \i = 1
      \coordinate (xmin) at (\xn,0);
      \coordinate (xmax) at (\xn,0);
      \coordinate (max) at (0,\yn);
    \else
      \ifdim \xn < \xmin
        \coordinate (xmin) at (\xn,0);
      \fi
      \ifdim \xn > \xmax
        \coordinate (xmax) at (\xn,0);
      \fi
      \ifdim \yn < \yc
        \coordinate (max) at (0,\yc);
      \else
        \coordinate (max) at (0,\yn);
      \fi
    \fi
  }
  \coordinate (gatecord) at ($(xmin)!.5!(xmax) + (max) + (0,#4) + (max)$);
  \node[#3,label={below:$#2$}] (#5) at (gatecord) {};
  \pgfgetlastxy{\xgate}{\ygate};
  \pgfmathtruncatemacro{\xgateround}{\xgate};
  \StrCount{#1,}{,}[\l] 
  \ifnum \l < 2 {\errmessage{#1 argument should be a comma-separated list of lenght >= 2}}
  \else{
    \foreach \n in {#1}{
      \path (\n);
      \pgfgetlastxy{\xnode}{\ynode};
      \pgfmathtruncatemacro{\xnround}{\xnode};
      \pgfmathsetmacro\tmpdiff{abs(\xnround - \xgateround)}
      \ifdim \tmpdiff pt > 1 pt \path[line] (#5) -| (\n);
      \else
        \path[line] (#5) -- (\n);
      \fi
    }
  }
  \fi
}

\newcommandx{\mkgatebeginold}[5][2={},3=ogate,4=.5,usedefault=@]{
  \coordinate (gatecord) at (0,0);
  \foreach \n [count=\i] in {#1}{
    \pgfgetlastxy \xc \yc;
    \path (\n);
    \pgfgetlastxy \xn \yn;
    \coordinate (gatecord) at ($(gatecord) + (\xn,0)$);
    \coordinate (gatecord) at ($1/\i*(gatecord)$);
    \ifdim \yn < \yc
    \node (max) at (0,\yc) {};
    \else
    \node (max) at (0,\yn) {};
    \fi
  }
  \coordinate (gatecord) at ($(gatecord) + (0,#4) + (max)$);
  \node[#3,label={below:$#2$}] (#5) at (gatecord) {};
  \pgfgetlastxy{\xgate}{\ygate};
  \pgfmathtruncatemacro{\xgateround}{\xgate};
  \StrCount{#1,}{,}[\l] 
  \ifnum \l < 2 {\errmessage{#1 argument should be a comma-separated list of lenght >= 2}}
  \else{
    \foreach \n in {#1}{
      \path (\n);
      \pgfgetlastxy{\xnode}{\ynode};
      \pgfmathtruncatemacro{\xnround}{\xnode};
      \pgfmathsetmacro\tmpdiff{abs(\xnround - \xgateround)}
      \ifdim \tmpdiff pt > 1 pt \path[line] (#5) -| (\n);
      \else
        \path[line] (#5) -- (\n);
      \fi
    }
  }
  \fi
}

\newcommandx{\mkmerge}[4][2=gatenode,3=i,4=.5,usedefault=@]{
  \mkgateend{#1}[{\ifempty{#3}{}{\nmerge[#3]}}][ogate][#4]{#2}
}

\newcommandx{\mkjoin}[4][2=gatenode,3=i,4=.5,usedefault=@]{\mkgateend{#1}[{\ifempty{#3}{}{\nmerge[#3]}}][agate][#4]{#2}}

\newcommandx{\mkgateend}[5][2={},3=ogate,4=.5,usedefault=@]{
  \coordinate (gatecord) at (0,0);
  \coordinate (xmax) at (0,0);
  \coordinate (xmin) at (0,0);
  \pgfgetlastxy \xmin \xmax;
  \foreach \n [count=\i] in {#1}{
    \pgfgetlastxy \xc \yc;
    \path (\n);
    \pgfgetlastxy \xn \yn;
    \ifnum \i = 1
      \coordinate (xmin) at (\xn,0);
      \coordinate (xmax) at (\xn,0);
      \coordinate (ymin) at (0,\yn);
    \else
      \ifdim \xn < \xmin
        \coordinate (xmin) at (\xn,0);
      \fi
      \ifdim \xn > \xmax
        \coordinate (xmax) at (\xn,0);
      \fi
      \ifdim \yn > \yc
        \coordinate (ymin) at (0,\yc);
      \else
        \coordinate (ymin) at (0,\yn);
      \fi
    \fi
  }
  \coordinate (gatecord) at ($(xmin)!.5!(xmax) + (ymin)$);
  \node[#3,label={above:$#2$}] (#5) at ($(gatecord) - (0,{#4})$) {};
  \pgfgetlastxy{\xgate}{\ygate};
  \pgfmathtruncatemacro{\xgateround}{\xgate};
  \StrCount{#1,}{,}[\l] 
  \ifnum \l < 2 {\errmessage{#1 argument should be a comma-separated list of lenght >= 2}}
  \else{
    \foreach \n in {#1}{
      \path (\n);
      \pgfgetlastxy{\xnode}{\ynode};
      \pgfmathtruncatemacro{\xnround}{\xnode};
      \pgfmathsetmacro\tmpdiff{abs(\xnround - \xgateround)}
      \ifdim \tmpdiff pt > 1 pt \path[line] (\n) |- (#5);
      \else
        \path[line] (\n) -- (#5);
      \fi
    }
  }
  \fi
}

\newcommandx{\mkgateendold}[5][2={},3=ogate,4=.5,usedefault=@]{
  \coordinate (gatecord) at (0,0);
  \coordinate (xmax) at (0,0);
  \coordinate (xmin) at (0,0);
  \pgfgetlastxy \xmin \xmax;
  \foreach \n [count=\i] in {#1}{
    \pgfgetlastxy \xc \yc;
    \path (\n);
    \pgfgetlastxy \xn \yn;
    \ifdim \xn < \xmin
    \coordinate (xmin) at (\xn,0);
    \fi
    \ifdim \xn > \xmax
    \coordinate (xmax) at (\xn,0);
    \fi
    \ifdim \yn > \yc
    \coordinate (ymin) at (0,\yc);
    \else
    \coordinate (ymin) at (0,\yn);
    \fi
    \coordinate (gatecord) at ($(xmin)!.5!(xmax) + (ymin)$);
  }
  \node[#3,label={above:$#2$}] (#5) at ($(gatecord) - (0,{#4})$) {};
  \pgfgetlastxy{\xgate}{\ygate};
  \pgfmathtruncatemacro{\xgateround}{\xgate};
  \StrCount{#1,}{,}[\l] 
  \ifnum \l < 2 {\errmessage{#1 argument should be a comma-separated list of lenght >= 2}}
  \else{
    \foreach \n in {#1}{
      \path (\n);
      \pgfgetlastxy{\xnode}{\ynode};
      \pgfmathtruncatemacro{\xnround}{\xnode};
      \pgfmathsetmacro\tmpdiff{abs(\xnround - \xgateround)}
      \ifdim \tmpdiff pt > 1 pt \path[line] (\n) |- (#5);
      \else
        \path[line] (\n) -- (#5);
      \fi
    }
  }
  \fi
}

\newcommand{\gatedistancein}{3pt}
\newcommand{\gatedistanceinand}{2pt}

\usetikzlibrary{
  arrows,
  backgrounds,
  chains,
  calc,
  decorations.markings,decorations.pathreplacing,
  fadings,
  fit,
  patterns,
  petri,
  positioning,
  shadows,
  shapes,automata,shapes.callouts
}

\tikzset{
  src/.style={draw,circle,fill=white,
    minimum size=2mm,
    inner sep=0pt
  },
  sink/.style={draw,circle,double,fill=white,
    minimum size=1.5mm,
    inner sep=0pt
  },
  node/.style={draw,circle,fill=black,
    minimum size=2mm,
    inner sep=0pt
  },
  source/.style={draw,circle,fill=white,
    minimum size=3mm,
    inner sep=0pt
  },
  sink/.style={draw,circle,double,fill=white,
    minimum size=3mm,
    inner sep=0pt
  },
  block/.style = {rectangle, draw=gray, align=center, fill=orange!25, rounded corners=0.1cm,
    minimum size=5mm, inner sep=2pt},
  prenode/.style = {minimum size=9pt,inner sep=2pt, font=\Large},
  bblock/.style = {rectangle, draw=blue!50, opacity=.7, line width=.5pt, align=center, fill=white, rounded corners=0.1cm,
    minimum size=4mm, inner sep=1pt},
  prenode/.style = {minimum size=9pt,inner sep=2pt, font=\Large},
  agate/.style={draw, rectangle,
    minimum size=3mm,
    inner sep=0pt,
    fill=orange!25,
    label={[red]center:$\mid$}
  },
  ogate/.style = {
    diamond, draw, fill=orange!25,
    minimum size=4mm,
    inner sep=0pt,
    label={[red]center:$+$}
  },
  lgate/.style = {
    diamond, draw, fill=orange!25,
    minimum size=4mm,
    inner sep=0pt,
    label={[red]center:$\circlearrowleft$}
    },
  altogate/.style = {
    diamond, draw,
    minimum size=4mm,
    inner sep=0pt,
    postaction={path picture={%
        \draw
        ([yshift=\gatedistancein]path picture bounding box.south) -- ([yshift=-\gatedistancein]path picture bounding box.north)
        ([xshift=-\gatedistancein]path picture bounding box.east) -- ([xshift=\gatedistancein]path picture bounding box.west)
        ;}}},
  altgate/.style={draw, rectangle,
    minimum size=3mm,
    inner sep=0pt,
    postaction={path picture={%
        \draw
        ([yshift=\gatedistanceinand]path picture bounding box.south) --
        ([yshift=-\gatedistanceinand]path picture bounding box.north) ;}}},
  anygate/.style = {circle, draw, fill=white,
    minimum size=4mm,
    inner sep=0pt,
    postaction={path picture={%
        \draw[black]
        ([xshift=-\gatedistancein,yshift=\gatedistancein]path picture bounding box.south east) --
        ([xshift=\gatedistancein,yshift=-\gatedistancein]path picture bounding box.north west)
        ([xshift=-\gatedistancein,yshift=-\gatedistancein]path picture bounding box.north east) --
        ([xshift=\gatedistancein,yshift=\gatedistancein]path picture bounding box.south west)
        ;}}
  },
  smallglobal/.style={
        node distance=1cm and 0.8cm, semithick, scale=0.8, every node/.style={transform shape}
  },
  elli/.style = {draw,densely dotted,-},
  line/.style = {draw,->, rounded corners=0.07cm,>=latex},
  nline/.style = {draw,semithick, ->},
  pline/.style = {draw,->,>=latex},
  node distance=1cm and 0.7cm,
  baseline=(current  bounding  box.center),
  local/.style={rectangle, draw, fill=\fillcolor, drop shadow,
    text centered, rounded corners, minimum height=5em
  },
  bigar/.style={
    draw,very thick, ->
  },
  process/.style={rectangle, draw=gray, fill=\fillcolor, drop shadow,
    text centered, minimum height=5em,text=gray
  },
  choreo/.style={rectangle, draw, fill=\fillcolor, drop shadow,
    text centered, rounded corners, minimum height=5em
  },
  mycfsm/.style={
        font=\footnotesize,
        initial where=above,
        ->,>=stealth,auto, node distance=1cm and 1cm,
        scale=1, every node/.style={transform shape},
        every state/.style=inner sep=2pt,
        baseline=(current  bounding  box.center),
        initial text={}
  },
  machinecloud/.style={
    cloud, cloud puffs=10, cloud ignores aspect, minimum height=.1cm, minimum width=2cm, draw
  },
  fitting node/.style={
    inner sep=0pt,
    fill=none,
    draw=none,
    reset transform,
    fit={(\pgf@pathminx,\pgf@pathminy) (\pgf@pathmaxx,\pgf@pathmaxy)}
  },
  mypetri/.style={
    font=\footnotesize,
    baseline=(current  bounding  box.center)
  },
  silentrans/.style = {rectangle, draw=black, align=center, fill=black,
    minimum height=1pt,
    minimum width=15pt,
    inner sep=1.5pt
  },
  reset transform/.code={\pgftransformreset},
  tmtape/.style={draw,minimum size=1.2cm}
}

\newcommand{\gunlessop}{\mbox{\colorOp\tiny\tt unless}}

\newcommandx{\gtry}[5][1=\gname,2={\aG_1 \gchoop \cdots \gchoop \aG_n},3=\gin,4=\gout,5={j},usedefault=@]{
  \def\foo{\gtryop\ {#2} \ \gcatchop\ {#3} {\colorOp \Rightarrow} {#4} {\colorOp \bullet} {\gname[{#5}]}}
  \gnode[{#1}][{\ifempty{#1} {\foo } {(\foo)}}]
}

\newcommandx{\gtrycatch}[4][1=\gname,2={\aG},3=\gin,4={\aG'},usedefault=@]{
  \def\foo{\gtryop\ {#2} \ \gcatchop\ {#3} \gdoop\ {#4}}
  \gnode[{#1}][{\ifempty{#1} {\foo} {(\foo)}}]
}

\newcommandx{\agG}[2][1={\aG},2=\aguard]{{#1} \ifempty{#2}{}{\ \gunlessop\ {#2}}}

\newcommandx{\grcho}[5][1=\gname,2={\agG},3={\agG[\aG'][\aguard']},4={\cdots},5=A,usedefault=@]{
  \def\foo{{#2} {\ \ifempty{#4}{\gchoop}{\gchoop \ifempty{#4}{}{\ {#4}\  \gchoop}}\ } {#3}}
  \ifempty{#1}{\ifempty{#5}{\foo}{\gselop\ \cpt[{#1}][{\ptp[#5]}]\big\{ \foo \big\}}}{\gselop\ \cpt[{#1}][{\ptp[#5]}]\big\{ \foo \big\}}
}

\newcommandx{\ggprefix}[3][1=\ptp,2={\aR},3={\aR'},usedefault=@]{f_{#1}} 
\newcommand{\aconfigfn}{\chi}
\newcommand{\aconfig}{\ell}

\newcommand{\lstates}{\statemap}
\newcommandx{\sysconfig}[3][1=\lstates,2=\aconfigfn,3={},usedefault=@]{
  \conf{ {#1},{#2} \ifempty{#3}{}{, #3} }
}
\newcommand{\sysctxfn}[1][]{\gamma_{#1}}
\newcommandx{\sysctx}[2][1=\aQ,2={},usedefault=@]{({#1},\sysctxfn[{#2}])}

\newcommandx{\alog}[4][1=\msg,2=q,3=\gname,4=t,usedefault=@]{({#1},{#2},{#3},{#4})}

\newcommand{\aCM}{M}\newcommand{\aM}{\aCM}
\newcommand{\aQ}{Q}
\newcommandx{\aQzero}[1][1=,usedefault=@]{
  {\ifempty{#1}{q_0}{q_{0#1}}}
}
\newcommand{\badbranches}[1][]{\beta\ifempty{#1}{}{({#1})}}
\newcommand{\aTrs}{\tset}
\newcommandx{\guardedaction}[2][1=\al,2=\aguard,usedefault=@]{
  {#1} \ifempty{#2}{}{/} {#2}
}
\newcommandx{\atrM}[4][1=q,2=\al,3={\hat q,\hat \al, \aguard},4=q',usedefault=@]{
  {#1} \xrightarrow[{#3}]{\guardedaction[{#2}][]} {{#4}}
}
\newcommandx{\atrS}[5][
  1={\sysconfig[@][@][\badbranches]},
  2=\al,
  3=\aguard,
  4={\sysconfig[\lstates'][\aconfigfn'][\badbranches]},
  5=\sysctx,usedefault=@
]{
  {#1} \xRightarrow{\qquad} {{#4}}
}
\newcommandx{\arevtrS}[2][
  1={\sysconfig[@][@][\badbranches]},
  2={\sysconfig[\lstates'][\aconfigfn'][\badbranches']},
  usedefault=@
]{
  {#1} \rightsquigarrow {#2}
}
\newcommand{\aCS}{S}

\newcommandx{\enables}[2][1=\aconfigfn,2=\aguard,usedefault=@]{{#1} \vdash {#2}}
\newcommandx{\gprojfn}[5][1=\aG,2=\ptp,3=\cinit,4=\cfinal,5={},usedefault=@]{
  \mathbf{proj}_{#2}({#1},{#3},{#4}\ifempty{#5}{}{,{#5}})
}

\newcommandx{\rbp}[3][1=\aG,2=\aconfigfn,3=\achan,usedefault=@]{\mathtt{RBP}_{{#1},{#2}}\ifempty{#3}{}{({#3})}}

\newcommand{\apseudoCFSM}{\mathtt{M}}
\newcommandx{\pseudoseq}[2][1=\apseudoCFSM,2=\apseudoCFSM',usedefault=@]{{#1}  ; {#2}}
\newcommandx{\pseudoCFSM}[4][1=\aQ,2=\aQzero,3=\cfinal,4=\aTrs,usedefault=@]{(#1 \ ; #2 \ ; #3 \ ; #4)}
\newcommandx{\markt}[3][1=\hat{\al},2=\hat{q},3=\aguard,usedefault=@]{\%\big({#1} , {#2}, {#3}\big)}

\newcommandx{\borderfn}[2][1=\aconfig,2=\aloop,usedefault=@]{
  \mathsf{border}_{{#2}}\ifempty{#1}{}{({#1})}
}

  \tikzset{
    mycallout/.style={
        fill=gray!10, opacity=.5, overlay, align=center,
        cloud callout, cloud puffs=10, aspect=1.9, cloud ignores aspect, cloud puff arc=100
      }
  }
\newcommandx{\ggvisually}[6][1=5pt,2=15pt,3=5pt,4=5pt,5=1.0cm,6=\scriptsize,usedefault=@]{
  \def\dist{\hspace{#5}}
  $\begin{array}{c@{\dist}c@{\dist}c@{\dist}c@{\dist}c@{\dist}c}
      \begin{tikzpicture}[node distance=0.9cm and 0.4cm, every node/.style={scale=.7,transform shape}]
        \node[source] (srcint) {};
        \node[sink,below=of srcint] (sinkint) {};
        \node[mycallout, above = .3cm of srcint, xshift=1cm, callout absolute pointer={(srcint.east)}] {source node};
        \node[mycallout, below = .3cm of sinkint, xshift=1cm, callout absolute pointer={(sinkint.west)}] {sink node};
        \path[line] (srcint) -- (sinkint);
      \end{tikzpicture}
       &
      \begin{tikzpicture}[node distance=0.9cm and 0.4cm, every node/.style={scale=.7,transform shape}]
        \mkint{}{int}[]
        \mkgraph{int}{int};
      \end{tikzpicture}
       &
      \begin{tikzpicture}[node distance=.9cm and 0.4cm, every node/.style={scale=.7,transform shape}]
        \node[bblock] at (0,0) (g) {$\aG$};
        \node[node, below=of g] (s1) {};
        \node[bblock, below=of s1] (gp) {$\aG'$};
        \path[line,dotted] (g) -- (s1);
        \path[line,dotted] (s1) -- (gp);
      \end{tikzpicture}
       &
      \begin{tikzpicture}[node distance=.4cm and 0.4cm, every node/.style={scale=.7,transform shape}]
        \node[bblock] at (-.7,0) (g) {$\aG$};
        \node[bblock] at (.7,0)  (gp) {$\aG'$};
        \node[node, above=of g] (f) {};
        \node[node, below=of g] (j) {};
        \node[node, above=of gp] (fp) {};
        \node[node, below=of gp] (jp) {};
        \path[line,dotted] (f) -- (g);
        \path[line,dotted] (g) -- (j);
        \path[line,dotted] (fp) -- (gp);
        \path[line,dotted] (gp) -- (jp);
        \mkfork{f,fp}[fork][][#1];
        \mkjoin{j,jp}[join][][#2];
        \mkgraph{fork}{join};
        \node[mycallout, above = .3cm of fork, xshift=-1cm, callout absolute pointer={(fork.west)}] {fork gate};
        \node[mycallout, above = -.9cm of join, xshift=-1cm, callout absolute pointer={(join.west)}] {join gate};
      \end{tikzpicture}
       &
      \begin{tikzpicture}[node distance=.4cm and 0.4cm, every node/.style={scale=.7,transform shape}]
        \node[bblock] at (-.7,0) (g) {$\aG$};
        \node[bblock] at (.7,0)  (gp) {$\aG'$};
        \node[node, above=of g] (f) {};
        \node[node, below=of g] (j) {};
        \node[node, above=of gp] (fp) {};
        \node[node, below=of gp] (jp) {};
        \path[line,dotted] (f) -- (g);
        \path[line,dotted] (g) -- (j);
        \path[line,dotted] (fp) -- (gp);
        \path[line,dotted] (gp) -- (jp);
        \mkbranch{f,fp}[fork][][#3];
        \mkmerge{j,jp}[join][][#4];
        \mkgraph{fork}{join};
        \node[mycallout, above = .3cm of fork, xshift=-1cm, callout absolute pointer={(fork.west)}] {branch gate};
        \node[mycallout, above = -.9cm of join, xshift=-1cm, callout absolute pointer={(join.west)}] {merge gate};
      \end{tikzpicture}
      \\
      \text{#6 empty}
       &
      \text{#6 interaction}
       &
      \text{#6 sequential}
       &
      \text{#6 parallel}
       &
      \text{#6 branching}
    \end{array}$
}

\newcommandx{\wwwcquote}[1][1=quo:w3c,usedefault=@]{
  \ifempty{#1}{}{\begin{quote}\label{#1}}
	 \lq\lq Using the Web Services Choreography specification, a
	 \textcolor{orange}{contract} containing a global definition of the
	 common \textcolor{orange}{ordering} conditions and constraints
	 under which \textcolor{orange}{messages} are exchanged, is
	 produced that describes, from a \textcolor{orange}{global
		viewpoint} [...]  observable behaviour of all the parties
	 involved.
    \textcolor{OliveGreen}{Each party} can then use the global definition to
    \textcolor{OliveGreen}{build and test solutions that conform to it}.
    The global specification is in turn \textcolor{OliveGreen}{realised by combination of} the
    resulting \textcolor{OliveGreen}{local systems} [...]\rq\rq
  \ifempty{#1}{}{\else\end{quote}}
}


%% file: intro.tex
Choreographic models of message-passing systems are becoming
increasingly popular both in
academia~\cite{BasuB11,BravettiZ07,CarboneHY12} and
industry~\cite{WS-CDL,BPMN,boner18}.
These models advocate two complementary views of communicating
systems.
The \emph{global} view can be thought of as an holistic description
of
the interactions carried out by a number of participants.
Dually, the \emph{local} views describe the expected communication
behaviour of each participant in isolation.

As discussed in \cref{sec:related}, the literature offers various
choreographic models.
Here, we introduce \emph{formal choreographic languages} (FCL) as a
meta-model to formalise message-passing systems; existing
choreographic models can be conceived as specifications of FCLs.
Essentially, this results in the definition of \emph{global} and
\emph{local} languages.
The words of a global language (g-language for short) consist of
(possibly infinite) sequences of \emph{interactions}; an interaction
has the form $\gint[]$ and represents the fact that participant \p\
sends message $\msg$ to participant \q, and participant \q\ receives
it.
The words of a local language (l-language for short) consist of
(possibly infinite) sequences of \emph{actions}; actions can take two
forms, $\ain$ and $\aout$, respectively representing that participant
\q\ receives message $\msg$ from \p\ and that participant \p\ sends
message $\msg$ to \q.

Such languages provide an abstract description of the possible runs of
a system in terms of sequences of interactions at the global level, 
executed through synchronous message-passing at the local level.
A word $\acword$ in a global language represents then a possible run
expected of a communicating system.
For instance, we can model a continuation of $\acword$ as another word
$\acword'$ such that $\acword \cat \acword'$ is also in the language.
Also, $\acword$ induces an expected \quo{local} behaviour on each
participant \p; in fact, the behaviour of \p\ is obtained by
\emph{projecting} $\acword$, that is by ignoring interactions on the
run $\acword$ not involving \p\ while retaining only the \emph{output}
and \emph{input} actions performed by \p.

Our language-theoretic treatment is motivated mainly by the need for a
general setting immune to syntactic restrictions.
This naturally leads us to consider e.g., choreographies represented
by context-free languages (cf. \cref{ex:parenthesis}).
In fact, we strive for generality; basically \emph{prefix-closure} is
the only requirement we impose on FCLs.
(We discuss some implications of relaxing prefix-closure in
\cref{sec:conc}.)
The gist is that, if a sequence of interactions or of communications
is an observable behaviour of a system, any prefix of the sequence
should be observable as well.
This allows us to consider partial executions as well as
\quo{complete} ones, which can be formalised in terms of maximal words
(namely words without continuations).
We admit infinite words to account for diverging computations,
ubiquitous in communication protocols.

Some g-languages cannot be faithfully executed by distributed
components.
This can be illustrated with a simple example; consider a
  g-language containing only the word $\gint \cat \gint[][c][n][d]$
  and its prefixes.  Such a language -- in our interpretation of words
  concatenation as a \quo{strict sequencing} operator (see
  e.g.~\cite{cm13,SeveriD19} for non-strict interpretations of
  sequencing) -- does specify that the interactions $\gint$ and
  $\gint[][c][n][d]$ must occur only in the given order.
  Clearly, this is not possible if the participants are distributed or
  act concurrently because \p[c] and \p[d] cannot be aware of when the
  interaction between \p\ and \q\ takes place.

\noindent\textbf{Contributions \& structure.}
We summarise below our main contributions.

\cref{sec:c-lang} introduces FCL (g-languages in \cref{def:chorlang},
l-languages in \cref{def:actlang}) and adapts standard constructions from
the literature.
In particular, we render communicating systems and choreographies as,
respectively, sets of l-languages (\cref{def:commSyst}) and
g-languages, while we take inspiration for projections from
choreographies and multiparty session types.
We consider synchronous interactions; as discussed in \cref{sec:conc},
the asynchronous case is scope for future work.

\cref{sec:ui} considers correctness of communicating systems with
respect to choreographies (the communicating system
\quo{executes} an interaction only if it is specified) and
completeness (the communicating system \quo{executes} at
least the specified interactions).
An immediate consequence of our constructions is the completeness of
systems projected from g-languages (\cref{th:completeness}).
Correctness is more tricky and requires to introduce \emph{closure
  under unknown information} (CUI, cf.  \cref{def:closedness}).
Intuitively, a g-language is CUI if it contains extensions of words
with a single interaction whose participants cannot distinguish the
extended word from other words of the language.
\cref{th:correctness} characterises correctness of projected systems
in terms of CUI.

\cref{sec:prop} shows how FCLs allow us to capture many relevant
communication properties in a fairly uniform way.

\cref{sec:proj-prop} proposes \emph{branch-awareness} (\cref{def:ba})
to ensure the communication properties defined in \cref{sec:prop}
(\cref{thm:baconseq}).
Intuitively, branch-awareness requires each participant to
\quo{distinguish} words where its behaviour differs.
Notably, we separate the condition for correctness (CUI) from the one
ensuring the communication properties (branch awareness). Most
approaches in the literature instead combine them into a single
condition, which takes names such as well-branchedness or
projectability~\cite{hlvlcmmprttz16}.  Thus, these single conditions
are stronger than each of CUI and branch-awareness.

We illustrate the generality of FCLs on three case studies
(cf.~\cref{sec:cfsmfcl,sec:chor-automata,sec:mpst}), respectively
taken from communicating finite-state machines~\cite{bz83},
choreography automata~\cite{blt20} and multiparty session
types~\cite{SeveriD19}.
We remark that FCLs can capture protocols that cannot be represented by
regular g-languages such as the \quo{task dispatching} protocol in
\cref{ex:parenthesis}.
To the best of our knowledge this kind of protocols cannot be
formalised in other approaches.

\cref{sec:related} compares with related work while
\cref{sec:conc} draws some conclusions and discusses
future work.

This paper is a revised and extended version of~\cite{blt22}; the main
differences are: ($i$) the results on decidability of CUI and
branch-awareness in \cref{sec:decidability} are new; ($ii$)
\cref{sec:cfsmfcl} yields a new case study; ($iii$)
\cref{sec:chor-automata} has been significantly extended; ($iv$) full
proofs have been added; ($v$) the text has been largely revised in
order to improve readability.


%% file: c_lang.tex
We briefly recall a few basic notions used throughout the paper and
fix some notation.
Let $\Sigma$ be an alphabet (i.e., a set of symbols).

We define the sets of finite and infinite words on $\Sigma$ as
\[
  \Sigma^\infty = \Sigma^\star \cup \Sigma^\omega
\]
where $\Sigma^\star$ is the usual Kleene-star closure on $\Sigma$ and
$\Sigma^\omega$ is the set of infinite words on $\Sigma$, that is maps
from natural numbers to $\Sigma$ (aka $\omega$-words~\cite{Staiger97}).

The concatenation operator is denoted as $\_\cat\_$ and $\epsilon$ is
its neutral element.
If $w$ is infinite then $w \cat w' = w$.
We write $a_0 \cat a_1 \cat a_2 \cat \ldots$ for the word mapping $i$
to $a_i \in \Sigma$ for all natural numbers $i$.

A language $L$ on $\Sigma$ is a subset of $\Sigma^\infty$, namely
$L \subseteq \Sigma^\infty$.
The prefix-closure of $L \subseteq \Sigma^\infty$ is defined as
\[
  \pref[L] = \Set{\aword \in \Sigma^\infty \sst \text{there is } \aword' \in L
	 \text{ such that } \aword \preceq \aword'}
\]
where $\preceq$ is the prefix relation; $L$ is \emph{prefix-closed} if
$L = \pref[L]$.

A word $\aword$ is \emph{maximal} \label{page:maximal} in a language
$L \subseteq \Sigma^\infty$ if $\aword \preceq \aword'$ for
$\aword' \in L$ implies $\aword' = \aword$.
As usual we shall write $\aword\prec\aword'$ whenever
$\aword\preceq\aword'$ and $\aword\neq\aword'$.

\medskip We shall deal with languages on particular
alphabets\footnote{These alphabets may be infinite; formal languages
  over infinite alphabets have been studied, e.g., in~\cite{ajb80}.},
namely the alphabets of \emph{interactions} and \emph{actions} whose
definition (see below) we borrowed from~\cite{blt20}.

\begin{definition}[Interactions and actions alphabets]
Let $\mathfrak{P}$ be a set of \emph{participants} (ranged over by $\p$, $\ptp[B]$, $\ptp[X],\ldots)$ and $\mathfrak{M}$ (ranged over by $\msg$, $\msg[x],\ldots$)
a  set of \emph{messages}, such that $\mathfrak{P}$ and $\mathfrak{M}$ are disjoint.
We define
\begin{align*}
  \alfint & = \Set{\gint \sst \p \neq \q \in \mathfrak{P}, \msg \in \mathfrak{M}}
  & & \text{ranged over by $\aint, \aint[b], \ldots$}
  \\
  \alfact & = \Set{\aout, \ain \sst \ptp[A] \neq \ptp[B] \in \mathfrak{P}, \msg \in \mathfrak{M}}
  & & \text{ranged over by $\aact, \aact[b], \ldots$}
\end{align*}
We call $\alfint$ and  $\alfact$, respectively, \emph{alphabet of interactions} and
 \emph{alphabet of actions} (over $\mathfrak{P}$ and $\mathfrak{M}$).
\end{definition}
Our results are independent of the sets $\mathfrak{P}$ and
$\mathfrak{M}$ we consider and hence we do not further specify them.
Words on $\alfint^\infty$ (ranged over by $\acword,\acword',...$) are
called \emph{interaction words} while those on $\alfact^\infty$
(ranged over by $\aaword,\aaword',...$) are called \emph{words of
  actions}.
Hereafter $\aword,\aword',...$ range over
$\alfint^\infty \cup \alfact^\infty$ and we use $\clang$ and $\alang$ to
range over subsets of, respectively, $\alfint^\infty$ and
$\alfact^\infty$.

Function $\ptpof[\cdot]$ yields the set of \emph{participants}
involved in an interaction or in an action and it is defined on
$\alfint \cup \alfact$ as follows:
\begin{align*}
  \ptpof[\gint] = \ptpof[\aout] = \ptpof[\ain] = \Set{\p,\q}
\end{align*}
This function extends naturally to (sets of) words.
The \emph{subject} of $\aout$ is the sender \p\ and the subject of
$\ain$ is the receiver \q.

We summarise in \cref{tab:nameconventions} all the name conventions
for variables ranging over the sets defined above.
\begin{table}[t]
  \begin{center}
	 \begin{tabular}{|l|@{\quad}l@{\quad}|@{\quad}l@{\quad}|}
		\hline
		{\bf Set} & {\bf Intuitive definition}& {\bf ranged over by} \\ [0.5ex]
		\hline
		$\mathfrak{P}$  & Participants & $\p$, $\ptp[B]$, $\ptp[X]\ldots$ \\\hline
		$\mathfrak{M}$ & Messages &$\msg$, $\msg[x],\ldots$ \\\hline
		finite subsets of  $\mathfrak{P}$ & Sets of participants & $\ptpset$\\\hline
                $\alfint$ & Interactions & $\aint, \aint[b], \ldots$\\\hline
		$\alfact$ & Actions & $\aact, \aact[b], \ldots$\\\hline
		$\alfint^\infty$ & Interaction words & $\acword,\acword',\ldots$\\\hline
		$\alfact^\infty$ & Words of actions & $\aaword,\aaword',\ldots$\\\hline
		$\alfint^\infty \cup \alfact^\infty$ & Words & $\aword,\aword',\ldots$\\\hline
		subsets of $\alfint^\infty$ & Global languages & $\clang$, $\clang',\ldots$\\\hline
		subsets of $\alfact^\infty$ & Local languages & $\alang$, $\alang',\ldots$\\\hline
	 \end{tabular}
  \end{center}
  \caption{Name conventions for variables\label{tab:nameconventions}}
\end{table}

A \emph{global language} specifies the expected interactions of a
system while a \emph{local language} specifies the communication
behaviour of a participant.

\begin{definition}[Global language]\label{def:chorlang}
  A \emph{global language} (\emph{g-language} for short) is a
  prefix-closed language $\clang$ on $\alfint$ such that
  $\ptpof[\clang]$ is finite.
\end{definition}

\begin{example}\label{ex:glang}
  For any finite subset $\Sigma$ of $\alfint$, the set $\Sigma^\star$
  is a g-language; notice that, if $\mathfrak{P}$ is infinite, then
  $\alfint^\star$ is not a g-language since it encompasses infinitely
  many participants.
  
  The set
  $\clang = \pref[\Set{\gint[][c][w][a] \cat \gint[][a][g][b], \,
	 \gint[][c][w][b] \cat \gint[][a][g][b], \, \gint[][c][w][a] \cat
	 \gint[][c][w][b] }]$ is a g-language (which will be used later on
  in other examples).
  It formally describes the following interaction protocol involving
  $\ptp[A]$lice, $\ptp[B]$ob and $\ptp[C]$arol: Carol can decide to
  ask either Bob or Alice to $\msg[w]$ork and after that Alice
  $\msg[g]$ossips with Bob, unless also Bob is asked to work after
  Alice is.
  \finex
\end{example}

\begin{definition}[Local language]\label{def:actlang}
  A \emph{local language} (\emph{l-language} for short) is a
  prefix-closed language $\alang$ on $\alfact$ such that 
  $\ptpof[\alang]$ is finite.
  An l-language is \emph{$\p$-local} if its words have all actions
  with subject \p.
\end{definition}
\begin{example}\label{ex:llang}
	 The set
	 $\alang =\Set{\emptyword,\, \aout[C][A][][w],\,
		\aout[C][B][][w],\, \aout[C][A][][w] \cat \aout[C][B][][w] }$ is
	 a $\ptp[C]$-local language; in particular, it specifies the local
	 behaviour of $\p[c]$arol with respect to the g-language $\clang$ of
	 \cref{ex:glang}.
  \finex
\end{example}

As discussed in the Introduction,
l-languages give rise to \emph{communicating systems}.

\begin{definition}[Communicating system]\label{def:commSyst}
  Let $\ptpset \subseteq \mathfrak{P}$ be a finite set of
  participants.
  A \emph{(communicating) system over $\ptpset$} is a map
  $\aCS = (\alang_{\p})_{\p \in \ptpset}$ assigning,  to each participant
  $\p \in \ptpset$, an $\p$-local
  language $\alang_{\p}\neq \Set{\epsilon}$ such that
  $\ptpof[\alang_{\p}] \subseteq \ptpset$.
\end{definition}

By projecting a g-language $\clang$ on a participant $\p$ we obtain
the $\p$-local language describing the sequence of actions performed
by $\p$ in the interactions involving $\p$ in the words of $\clang$.
\cref{def:projection} below recasts in our setting the notion of
projection used in several choreographic formalisms, e.g., in~\cite{CarboneHY12,honda16jacm}.
\begin{definition}[Projection]\label{def:projection}
  The \emph{projection on \p} of an interaction $\gint[][B][@][C]$ is
  computed by the function
  $\proj{\_}{\_} : \alfint \times \mathfrak{P} \to \alfact \cup
  \Set{\epsilon}$ defined by:
  \[
	 \proj{(\gint[])}{A} = \aout \qquad\qquad \proj{(\gint[])}{B} =
	 \ain \qquad\qquad \proj{(\gint[])}{c} = \epsilon
  \]
  and extended homomorphically to interaction words and g-languages.
  The \emph{projection of a g-language $\clang$}, written
  $\proj{\clang}{}$, is the communicating system
  ${(\proj{\clang}{\p})_{\p \in \ptpof[\clang]}}$.
\end{definition}

\begin{example}\label{ex:simple}
  Let $\clang$ be the g-language in \cref{ex:glang}.
  By \cref{def:projection}, we have
  $\proj{\clang}{} = (\proj{\clang}{X})_{\ptp[X] \in
	 \Set{\p,\q,\ptp[c]}}$ where
  $\proj \clang a = \Set{\emptyword,\,\ain[C][A][][w]$,
	 $\ain[C][A][][w]\cat\aout[A][B][][g],\, \aout[A][B][][g]}$,
  $\proj \clang B = \Set{\emptyword,\,\ain[A][B][][g],\,
	 \ain[C][B][][w]$, $\ain[C][B][][w]\cat\ain[A][B][][g]}$, and
  $\proj \clang C= \Set{\emptyword,\, \aout[C][A][][w],\,
	 \aout[C][B][][w],\, \aout[C][A][][w] \cat \aout[C][B][][w] }$.
  The latter is the l-language in \cref{ex:llang}.
  \finex
\end{example}

We consider a \emph{synchronous} semantics of communicating systems,
similarly to other choreographic approaches such
as~\cite{BravettiZ07,CarboneHY12,Dezani-Ciancaglini16,SeveriD19}.
Intuitively, a choreographic word is in the semantics of a system
$\aCS$ iff its projection on each participant \p\ yields a word in
the local language assigned by $\aCS$ to \p.

\begin{definition}[Semantics]\label{def:syncSem}
  Given a system $\aCS$ over $\ptpset$, the set
  \[
	 \ssem{\aCS} = \Set{\acword \in \alfint^\infty \sst  \ptpof[\acword] \subseteq \ptpset \text{ and for all } \p \in
		\ptpset, \ \proj{\acword}{\p} \in \aCS(\p)}
  \]
  is the \emph{(synchronous) semantics} of $\aCS$.
\end{definition}
Notice that the above definition coincides with the \emph{join}
operation in~\cite{fbs04}, used in realisability conditions for an
asynchronous setting.

By the finiteness condition on the number of participants in a
g-language (\cref{def:chorlang}) we immediately get the following.
\begin{myfact}\label{fact:finfin}
  Let $\clang$ be a g-language.  Then $\clang\subseteq \alfint^\star$
  implies $\ssem{\proj{\clang}{}}\subseteq \alfint^\star$.
\end{myfact}
\begin{example}\label{ex:simple2}
  For the system $\proj \clang {}$ in \cref{ex:simple}, we have
  \begin{align*}
	 \ssem{\proj{\clang}{}} = \pref[\Set{\gint[][c][w][a] \cat \gint[][a][g][b], \ \gint[][c][w][b] \cat
		\gint[][a][g][b], \ 		\gint[][a][g][b], \
		\gint[][c][w][a] \cat \gint[][c][w][b] \cat
		\gint[][a][g][b]
	 }] \tag*{$\diamond$} 
  \end{align*}
\end{example}

Two interactions $\aint$ and $\aint[b]$ are \emph{independent} (in
symbols $\intpar$) when
$\ptpof[\aint] \cap \ptpof[{\aint[b]}] = \emptyset$.
Informally, a language $\clang$ describing the behaviour of a system,
and containing a word $\acword$, does contain also all the words where
independent interactions in $\acword$ are swapped; we say that
$\clang$ is \emph{concurrency closed}.
The notion of concurrency closure in our setting is a delicate one
because of the possible presence of infinite words.
One in fact has to allow infinitely many swaps of independent interactions while avoiding that
the interactions do disappear by pushing them infinitely far away.
Technically, we consider Mazurkiewicz's traces~\cite{maz86} on
$\alfint$ with independence relation $\intpar$:
\begin{definition}[Concurrency closure]\label{def:cclos}
  Let $\comm$ be the reflexive and transitive closure of the relation
  $\equiv$ on finite interaction words defined by
  $\acword\,\aint\,\aint[b]\,\acword' \equiv
  \acword\,\aint[b]\,\aint\,\acword'$ where $\intpar$.
  Following \emph{\cite[Def. 2.1]{Gastin90}}, $\comm$ extends to
  $\alfint^\omega$ by defining
  \[
	 \text{for all } \acword,\acword'\in\alfint^\omega:
	 \qquad
	 \acword \comm\acword' \quad \text{ iff }\quad \acword \ll\acword' \qand \acword' \ll\acword
  \]
  where
  $\acword \ll \acword'$ iff for each finite prefix $\acword_1$ of
  $\acword$ there are a finite prefix $\acword_1'$ of $\acword'$ and a
  g-word $\hat \acword \in \alfint^\star$ such that
  $\acword_1 \cat \hat \acword \comm \acword_1'$.
  A g-language $\clang$ is \emph{concurrency closed} if it coincides
  with its concurrency closure, namely
  $\clang = \Set{\acword \in \alfint^\infty \sst \text{there is }
	 \acword' \in \clang \text{ such that } \acword \comm \acword'}$.
\end{definition}

As expected from the discussion above, the semantics of systems is
naturally concurrency closed since in a distributed setting
independent events can occur in any order.
Indeed, relation $\comm$ can be characterised as follows.
\begin{lemma}\label{lemma:equalProj}
  Given a g-language $\clang$ and two words
  $\acword_1, \acword_2 \in \clang$, $\acword_1 \comm \acword_2$ iff
  $\proj{\acword_1}{\p} = \proj{\acword_2}{\p}$ for each
  $\p \in \ptpof[\clang]$.
\end{lemma}
\begin{proof}
  This follows directly from~\cite[Proposition 2.2]{Gastin90}.
\end{proof}
Therefore we have
\begin{restatable}{proposition}{ccsem}\label{prop:par}
  Let $\aCS$ be a system. Then $\ssem{\aCS}$ is
  concurrency closed.
\end{restatable}
\begin{proof}
  Trivial, since closure under swap does not change the projection by
  \cref{lemma:equalProj}.
\end{proof}

The intuition that g-languages, equipped with the projection and
semantic functions of \cref{def:projection} and \cref{def:syncSem}, do
correspond to a natural syntax and semantics for the abstract notion
of choreography, can be strengthened by showing that these functions
form a Galois connection.

Let us define
$\glanguages = \Set{\clang \sst \clang \text{ is a g-language}}$ and
$\csystems = \Set{\aCS \sst \aCS \text{ is a system}}$.
Moreover, given $\aCS, \aCS' \in \csystems$, we define
$\aCS\subseteq\aCS'$ if $\aCS(\p)\subseteq \aCS'(\p)$ for each \p.

\begin{restatable}{proposition}{gc}\label{prop:gc}
  The functions $\proj{\_}{}$ and $\ssem{\_}$ form a (monotone)
  Galois connection between the posets $(\glanguages,\subseteq)$ and
  $(\csystems,\subseteq)$, namely, $\proj{\_}{}$ and $\ssem{\_}$ are
  monotone functions such that, given $\clang\in\glanguages$ and
  $\aCS\in\csystems$:
$$\proj{\clang}{}\subseteq\aCS  \iff  \clang\subseteq\ssem{\aCS}$$
\end{restatable}
\begin{proof}
  Note that $\proj{(\_)}{}$ and $\ssem{\_}$ are trivially monotone by their
  definitions.
  \\
  ($\implies$)
  We first observe that
  \[\begin{array}{llll}
	 \ssem{\proj{\clang}{}}
		& =
		&
		\Set{\acword' \in \alfint^\infty \sst  \ptpof[\acword'] \subseteq \ptpof[\clang] \text{ and for all } \p \in \ptpof[\clang], \ \proj{\acword'}{\p} \in \proj \clang \p}
		\\
		& \subseteq
		&
		\Set{\acword' \in \alfint^\infty \sst  \ptpof[\acword'] \subseteq \ptpof[\clang] \text{ and for all } \p \in \ptpof[\clang], \ \proj{\acword'}{\p} \in  \aCS(\p)}
		\\		& = & \ssem{\aCS}
	 \end{array}
  \]
  where the equalities above hold by \cref{def:syncSem} and the inclusion holds by hypothesis
  ($\proj{\clang}{} \subseteq\aCS$).
  Hence, given a word $\acword \in \clang$, we get
  $\acword \in \ssem{\proj{\clang}{}}$ by construction and  $\acword \in \ssem{\aCS}$ by the above.
  \\
  ($\Longleftarrow$) We have to show that
  $\proj{\clang}{A}\subseteq\aCS(\p)$ for each $\p$.  Let hence
  $\acword_{\p} \in \proj{\clang}{\p}$. By definition there is
  $\acword \in \clang$ with $\proj{\acword}{\p}=\acword_{\p}$.
 Then $\acword \in \ssem{\aCS}$ by hypothesis and hence, by
 \cref{def:syncSem}, $\proj{\acword}{\p}(=\acword_{\p})\in\aCS(\p)$.
\end{proof}

Notice that, by \cref{prop:gc}, $\proj{\clang}{}\subseteq\aCS$ can be
understood as \quo{$\clang$ can be realised by $\aCS$} according to the
notion of realisability frequently used in the literature, namely that
all behaviours of the choreography are possible for the system.

  A \emph{closure operator} is a function $\cl$ whose domain and
  codomain are ordered sets and such that $\cl$ is monotone
  ($x \leq y \implies \cl[x] \leq \cl[y]$), extensive
  ($x \leq \cl[x]$), and idempotent ($\cl[x] = \cl[{\cl[x]}]$).
  It is well-known that, given a Galois connection
  $(f_\star,f^\star)$, $\cl = f^\star \circ f_\star$ is a closure
  operator.
  In our setting $\ssem{\proj \_ {}}$ is a closure operator,
hence the above boils down to the following corollary:
\begin{corollary}\label{fac:clgc}
  For all g-languages $\clang, \clang' \in \glanguages$,
  \begin{description}
  \item[monotonicity]
	 $\clang \subseteq \clang' \implies \ssem{\proj \clang {}}
	 \subseteq \ssem{\proj{\clang'} {}}$,
  \item[extensiveness] $\clang \subseteq \ssem{\proj{\clang}{}}$,
  \item[idempotency] $\ssem{{\proj \clang {}}} = \ssem{\proj{\ssem{\proj{\clang}{}}}{}}$.
  \end{description}
\end{corollary}
As we shall see, extensiveness coincides with completeness
(\cref{def:cc}) and, together with monotonicity, implies
\emph{harmonicity} (\cref{def:llive}).


%% file: ui.tex
A g-language specifies the expected communication behaviour of a
system made of several components.
\cref{def:commSyst} formalises such systems in terms of l-languages.
In the present section we deal with properties relating a
communicating system with a specification (i.e., a g-language).  In
particular, we first introduce correctness and completeness of a
communicating system with respect a g-language. The latter property
follows by the Galois connection discussed in the previous section.
Instead in order to prove correctness conditions are introduced in
terms of a closure property, CUI; this requires to handle
continuity. We highlight the expressiveness of CUI g-languages by
showing that there exist non-regular CUI g-languages.

\begin{definition}[Correctness and completeness]\label{def:cc}
  Let $\clang$ be a g-language.
  A system $\aCS$ is \emph{correct with respect to $\clang$} if
  $\ssem{\aCS} \subseteq \clang$ and it is \emph{complete with respect
	 to $\clang$} if $\ssem{\aCS} \supseteq \clang$.
\end{definition}
Correctness and completeness are related to existing notions.
For instance, in the literature on multiparty session types (see,
e.g., the survey~\cite{hlvlcmmprttz16}) correctness is analogous to
\emph{subject reduction} and completeness to \emph{session fidelity}.

Notice that, by \cref{prop:gc}, we can interpret
$\proj{\clang}{} \subseteq \aCS$ as a characterisation for
completeness of $\aCS$ with respect to $\clang$.
Hence, an immediate result of the Galois connection defined in
\cref{sec:c-lang} is that any system projected from a g-language is
complete.
In fact, completeness coincides with the extensiveness property of the
closure operator associated to our Galois connection.
\begin{restatable}{corollary}{completeness}\label{th:completeness}
  For any g-language $\clang$, $\proj{\clang}{}$ is complete with
  respect to $\clang$.
\end{restatable}
It is easy to check that a similar result does not hold for correctness. 
If we consider the g-language $\clang$ of \cref{ex:glang}, we have that, 
as shown in \cref{ex:simple2},  $\gint[][c][w][a] \cat \gint[][c][w][b]\cat\gint[][a][g][b]\in\ssem{\proj{\clang}{}}$ but $\gint[][c][w][a] \cat \gint[][c][w][b]\cat\gint[][a][g][b]\not\in\clang$.
That is $\ssem{\proj{\clang}{}} \not\subseteq \clang$.

\paragraph{Characterising correctness for projected systems.}
Can we identify conditions on g-languages to ensure correctness of
their projections?
The answer to this question is positive since correctness can be
characterised as a closure property.

\begin{definition}[CUI]\label{def:closedness}
  A g-language $\clang$ is \emph{closed under unknown information}, in
  symbols $\cuui$, if for all finite words
  $\acword_1 \cat \aint, \acword_2 \cat \aint \in \clang$ with the
  same final interaction $\aint = \gint \in \alfint$ we have
  $\acword \cat \aint \in \clang$ for all $\acword \in \clang$ such
  that $\gproj[\acword][\p]= \gproj[\acword_1][\p]$ and
  $\gproj[\acword][\q] = \gproj[\acword_2][\q] $.
\end{definition}
Intuitively, participants cannot distinguish words with the same
projection on their role.
Hence, if two participants \p\ and \q\ find words $\acword_1$ and
$\acword_2$ compatible with another word $\acword$, and interaction
$\gint$ can occur after both $\acword_1$ and $\acword_2$,
then it should be enabled also after $\acword$.
Indeed, \p\ cannot know whether the current word is $\acword$ or
$\acword_1$ and, likewise, \q\ cannot tell apart $\acword$ and
$\acword_2$.
Hence, $\clang$ should encompass $\acword$ because, after the
execution of $\acword$, \p\ and \q\ are willing to take $\gint$, which
can thus happen at the system level.
Closure under unknown information (CUI for short) lifts this
requirement at the level of g-language.
\begin{example}\label{ex:notCUI}
  The language $\clang$ in \cref{ex:glang} is not CUI because it
  contains the words
  \[
	 \acword_1 \cat \aint = \gint[][c][w][a] \cat \gint[][a][g][b] \qquad
	 \acword_2 \cat \aint = \gint[][c][w][b] \cat \gint[][a][g][b] \qqand \acword =
	 \gint[][c][w][a] \cat \gint[][c][w][b]
  \]
  and \p\ cannot distinguish between $\acword_1$ and $\acword$ while
  \q\ cannot distinguish between $\acword_2$ and $\acword$;
  nonetheless
  $\acword \cat \gint[][a][g][b] = \gint[][c][w][a] \cat \gint[][c][w][b] \cat
  \gint[][a][g][b]\not \in \clang$.
  Notice that, as shown after \cref{th:completeness}
  $\clang \not\supseteq \ssem{\proj{\clang}{}}$.
  \finex
\end{example}
The language in \cref{ex:notCUI} is not the semantics of any system,
in fact languages obtained as semantics of a communicating system are
always CUI.

\begin{restatable}[Semantics is CUI]{proposition}{semiscui}\label{prop:semiscui}
  For all systems $\aCS$, $\ssem{\aCS}$ is CUI.
\end{restatable}
\begin{proof}
  Let $\aCS$ be a system over $\ptpset$.
  In order to show closure of $\ssem{\aCS}$ under unknown information, let
  us take words
  $\acword_{\p} \cat \gint,\
  \acword_{\q} \cat \gint,\ \acword' \in \ssem{\aCS}$, such that
  \begin{eqnarray*}
	 \proj{\acword_{\p}}{\p}=\proj{\acword'}{\p}
	 \qqand
	 \proj{\acword_{\q}}{\q}=\proj{\acword'}{\q}
  \end{eqnarray*}
  By \cref{def:closedness}, we have to show that
  $\acword'\cat\gint \in \ssem{\aCS}$.
  This in turn, by definition of synchronous semantics, amounts to
  show that $\proj{(\acword'\cat\gint)}{x} \in \aCS(\px)$ for each
  $\p[x] \in \ptpset$.
  If $\p[x] \neq \p,\q$ then
  $\proj{(\acword'\cat\gint)} x = \proj{\acword'} x \in \aCS(\px)$
  since $\acword' \in \ssem \aCS$.
  Otherwise $\p[x] \in \Set{\p,\q}$; hence, by hypothesis and
  definition of projection we have
  \[
	 \proj{(\acword'\cat\gint)}{\p} = \proj{\acword_{\p}}{\p} \cat \aout \in \aCS(\p)
	 \qand
	 \proj{(\acword'\cat\gint)}{\q} =
	 \proj{\acword_{\q}}{\q} \cat \ain \in \aCS(\q)
  \]
  where the last two equalities hold by hypothesis.
\end{proof}

The next property connects finite and infinite words in a language; it
corresponds to the closure under the limit operation used in
$\omega$-languages~\cite{Eilenberg76,Staiger97}.
\begin{definition}[Continuity]\label{def:continuity}
  A language $L$ on an alphabet $\Sigma$ is \emph{continuous} if
  $z \in L$ for all $z \in \Sigma^\omega$ such that $\pref[z] \cap L$
  is infinite.
\end{definition}
This notion of continuity, besides being quite natural, is the most
suitable for our purposes among the possible
ones~\cite{Redziejowski86}.
Intuitively, a language $L$ is continuous if an $\omega$-word is in
$L$ when infinitely many of its approximants (i.e., finite prefixes)
are in $L$.
A g-language $\clang$ is \emph{standard or continuous} (\sclang, for
short) if either $\clang \subseteq \alfint^\star$ or $\clang$ is
continuous.
Notice that, for prefix-closed languages, for all
$z \in \Sigma^\omega$ we have that $\pref[z] \cap L$ is infinite
iff $\pref[z] \subseteq L$.

Closure under unknown information characterises correct
projected systems.
\begin{restatable}[Characterisation of correctness]{theorem}{correctness}
  \label{th:correctness}
  If $\proj{\clang}{}$ is correct with respect to~$\clang$ then $\cuui$ holds.
  If $\clang$ is an \sclang and $\cuui$ then $\proj{\clang}{}$
  is correct with respect to $\clang$.
\end{restatable}
\begin{proof}
  We prove the first implication.
  In order to show closure of $\clang$ under unknown information, let
  us take words
  $\acword_{\p} = \acword'_{\p} \cat \gint,\ \acword_{\q} =
  \acword'_{\q} \cat \gint,\ \acword' \in \clang$, such that
  \begin{eqnarray}\label{eq:wawb}
	 \proj{\acword'_{\p}}{\p}=\proj{\acword'}{\p}
	 \qqand
	 \proj{\acword'_{\q}}{\q}=\proj{\acword'}{\q}
  \end{eqnarray}
  By \cref{def:closedness}, we have to show that
  $\acword'\cat\gint \in \clang$.
  Thanks to correctness, it is enough to show that
  $\acword'\cat\gint \in \ssem{\proj{\clang}{}}$.
  This in turn, by definition of synchronous semantics, amounts to
  show that
  $\aaword_{\p[x]} =\proj{(\acword'\cat\gint)}{x} \in
  \proj{\clang}{x}$ for each $\p[x] \in \ptpset$.
  If $\p[x] \neq \p,\q$ then
  $\aaword_{\p[x]} = \proj{(\acword'\cat\gint)} x = \proj{\acword'} x
  \in \proj{\clang} x$.
  Otherwise $\p[x] \in \Set{\p,\q}$; we consider only the case
  $\p[x] = \p$ since the other case is analogous.
  We have
  \[
	 \begin{array}{lcll}
		\aaword_{\p[x]} &=& \proj{(\acword'\cat\gint)}{\p} \\
							 &=  & \proj{\acword'}{\p}\cat\proj{(\gint)}{\p} & \text{(by def. of projection)}\\
							 & = &\proj{\acword'_{\p}}{\p}\cat\proj{(\gint)}{\p} & \text{(by (\ref{eq:wawb}))}\\
							 & = & \proj{(\acword'_{\p}\cat\gint)}{\p}  \in \proj \clang A & \text{(by def. of projection and  $\acword'_{\p}\cat\gint\in\clang$)}
	 \end{array}
  \]
  as required.

  We now prove the second implication.
  By \cref{def:cc}, we have to show that
  $\ssem{\proj{\clang}{}} \subseteq \clang$; we proceed by
  contradiction.
  Fix a word $\acword \in \ssem{\proj{\clang}{}} \setminus \clang$.
  The only possible cases are:
 \proofcase{ $\clang \subseteq \alfint^\star$} By \cref{fact:finfin},
	 $\ssem {\proj \clang {}}$ does not contain infinite
	 g-words.
	 Hence, $\acword$ is finite and we can take its longest prefix
	 $\acword'$ that belongs to $\clang$.
	 Let $\aint = \gint$ be the interaction immediately following
	 $\acword'$ in $\acword$.
	 We can choose $\acword_{\p}, \acword_{\q}\in\clang$ such that
	 $\proj{\acword_{\p}}{a}= \proj{\acword}{a}$ and
	 $\proj{\acword_{\q}}{b}= \proj{\acword}{b}$.
	 (Recall that by \cref{def:syncSem} for each $\p[x] \in \ptpset$
	 there is a word $\acword_{\p[x]} \in \clang$ such that
	 $\proj{\acword}{x} = \proj{\acword_{\ptp[x]}}{x}$.)
	 Take the shortest prefixes $\acword'_{\p}$ and $\acword'_{\q}$ of
	 $\acword_{\p}$ and $\acword_{\q}$ respectively such that
	 \[
		\proj{\acword'_{\p}}{a} = \proj{(\acword'\cat\gint)}{a}
		\qqand
		\proj{\acword'_{\q}}{b}= \proj{(\acword'\cat\gint)}{b}
	 \]
	 It is then easy to check that $\acword'_{\p}$ and $\acword'_{\q}$
	 have necessarily the following shapes
	 \[
		\acword'_{\p} = \tilde{\acword}_{\p}\cat\gint \qquad\qquad
		\acword'_{\q} = \tilde{\acword}_{\q}\cat\gint
	 \] where
	 \begin{itemize}
	 \item $\tilde{\acword}_{\p},\tilde{\acword}_{\q}\in\clang$, since
		$\clang$ is prefix-closed;
	 \item $\proj{\tilde{\acword}_{\p}}{a}= \proj{\acword'}{a}$ and
		$\proj{\tilde{\acword}_{\q}}{b}= \proj{\acword'}{b}$.
	 \end{itemize}
	 Now, by $\acword'\in\clang$ and the definition of CUI, we infer that $\acword' \cat\gint \in \clang$
	 against the hypothesis that $\acword'$ was the longest such prefix
	 of $\acword$.
 \proofcase{ $\clang \nsubseteq \alfint^\star$} If the set
	 $\widehat \clang = \Set{\acword' \in \alfint^\star \sst \acword'
		\preceq \acword \text{ and } \acword' \in \clang}$ is finite then we
	 take the longest g-word in $\widehat \clang$ and the proof is as in
	 the previous case. 
	 Otherwise we have a contradiction because, by continuity,
	 $\acword \in \clang$.
\end{proof}
Notice that CUI is defined in terms of g-languages only, hence
checking CUI does not require to build the corresponding system.
This allows us to study CUI on specific classes of g-languages.
For instance, we can show that CUI is decidable on a class of languages
accepted by B\"uchi automata (cf. \cref{sec:chor-automata}).
An interesting observation is that strengthening the precondition of
\cref{def:closedness} with the additional requirement
$\acword_1 = \acword_2$ would invalidate \cref{th:correctness}.
Indeed, $\clang \cup \{\gint[][a][g][b]\}$ with $\clang$ the
language in \cref{ex:glang} would become CUI but not correct.
The next example shows that the continuity condition in
\cref{th:correctness} is necessary for languages containing infinite
g-words.
\begin{example}[Continuity matters] The CUI language
  \[
	 \clang = \pref[{\bigcup_{i \geq 0}\Set{
		  \gint[][@][l] \cat \gint[][b][n][c] \cat (\gint[][c][n][d])^i}} \cup
		 \Set{\gint[][@][r] \cat \gint[][b][n][c] \cat (\gint[][c][n][d])^\omega}
		]
  \]
  does contain an infinite word but it is not continuous.
  The projection of $\clang$ is not correct because its semantics
  contains the g-word
  $\gint[][@][l] \cat \gint[][b][n][c] \cat (\gint[][c][n][d])^{\omega}$.
  This word, in fact, does not belong to $ \clang$ since the
  projections of \p[c] and \p[d] can exchange infinitely many messages
  $\msg[n]$ due to the infinite g-word of $\clang$ regardless whether
  \p\ and \q\ exchange $\msg[l]$ or $\msg[r]$.
  \finex
\end{example}
Notice that, since $\clang\subseteq\ssem{\proj{\clang}{}}$ always
holds, \cref {th:correctness} implies that $\cuui$
characterises the languages $\clang$ such that
$\clang=\ssem{\proj{\clang}{}}$.
Besides, the following corollary descends from \cref {th:correctness}.
\begin{restatable}{corollary}{smallest}
  If $\clang$ is an \sclang, $\ssem{\proj \clang {}}$ is the smallest
  CUI \sclang containing $\clang$.
\end{restatable}
\begin{proof}
  Let $\cl[\clang] = {\ssem{\proj \clang {}}}$.
  Given an \sclang $\clang$, $\cuui[{\cl[\clang]}]$ holds by
  \cref{prop:semiscui}.  Moreover, it is not difficult to check that
  if $\clang$ is an \sclang, so is $\cl[\clang]$.
  $\clang\subseteq\cl[\clang]$ holds by extensiveness of $\cl$.  Now,
  in order to show that $\cl[\clang]$ is smaller or equal than any CUI
  \sclang containing $\clang$, let us consider any \sclang
  $\tilde\clang$ such that $\cuui[\tilde\clang]$ and
  $\clang\subseteq\tilde\clang$.  By monotonicity of $\cl$, we have
  that $\cl[\clang]\subseteq\cl[\tilde\clang]$ and, by
  \cref{th:correctness}, $\cl[\tilde\clang]=\widehat\clang$. So
  $\cl[\clang]\subseteq\tilde\clang$.
\end{proof}

CUI ensures that continuous g-languages are concurrency closed.

\begin{restatable}{proposition}{continuousiscc}
  If $\clang$ is an \sclang and $\cuui$, then $\clang$ is
  concurrency closed.
\end{restatable}
\begin{proof}
  Thanks to \cref{th:correctness} $\clang=\ssem{\proj{\clang}{}}$,
  hence it is concurrency closed thanks to \cref{prop:par}.
\end{proof}
Hence, an \sclang cannot be CUI unless it is
  concurrency closed.\\

As recalled before, in many choreographic formalisms (such
as~\cite{BDLT20,hlvlcmmprttz16,cdyp16,bbo12,fbs04}) the correctness
and completeness of a projected system, namely
$\clang = \ssem{\proj{\clang}{}}{}$ (together with some forms of
liveness and deadlock-freedom properties), is guaranteed by
\emph{well-branchedness} conditions.
Most of such conditions guarantee, informally speaking, that
participants reach consensus on which branch to take when choices
arise.
For instance, a well-branchedness condition could be that, at each
choice, there is a unique participant deciding the branch to follow
during a computation and that such participant informs each other
participant.
%
Such a condition is actually not
needed to prove $\clang = \ssem{\proj{\clang}{}}{}$, as shown by the example below.
\begin{example}
  The g-language $\clang$ of \cref{ex:simple2} is CUI, without being
  well-branched in the above sense.
  Indeed, after the interaction $\gint[][c][w][a]$, there is a
  branching in the projected system, since both the interactions
  $\gint[][c][w][b]$ and $\gint[][a][g][b]$ can be performed.
  However, these interactions do not have the same sender.
\end{example}

A key merit of our model is its generality and expressiveness.
We show this with an example of a non-regular CUI g-language whose
projected system is correct and complete by \cref{th:correctness} and
\cref{th:completeness}, respectively.

\begin{example}[A non-regular CUI g-language]\label{ex:parenthesis}
  We consider a task dispatching service where, as soon as a
  $\ptp[S]$erver communicates its $\msg[a]$vailability, a
  $\ptp[d]$ispatcher sends a $\msg[t]$ask to $\p[s]$.
  The server either processes the task directly and sends back the
  resulting $\msg[d]$ata to $\p[d]$ or sends the task to participant
  $\ptp[H]$ for some pre-processing, aiming at resuming it later on.
  Indeed, after communicating a result to $\p[d]$, the server can
  $\msg[r]$esume a previous task (if any) from $\ptp[H]$, process it,
  and send the result to $\p[d]$.
  The server eventually stops by sending $\msg[s]$ to both $\ptp[d]$
  and $\ptp[H]$; this can happen only when all dispatched tasks have
  been processed.
  The task dispatching scenario above is faithfully captured by the
  g-language $\clang$ obtained by prefix-closing the (non-regular)
  language generated by the following context-free grammar.
  \begin{eqnarray*}
		L & \bnfdef & L' \cat  \gint[][S][s][D] \cat \gint[][s][s][H]
		\\
		L' & \bnfdef & \gint[][S][a][D] \cat \gint[][D][t][S] \cat \gint[][S][t][H]
							\cat L' \cat
							\gint[][S][r][H] \cat \gint[][H][r][S] \cat \gint[][S][d][D] \cat L'
							\bnfmid \gint[][S][a][D] \cat \gint[][D][t][S] \cat \gint[][S][d][D]
							\cat L' \bnfmid \emptyword
	 \end{eqnarray*}
  Since $\ptp[S]$ is involved in all the interactions of $\clang$, for
  each pair of words $\acword,\acword' \in \clang$:
  $\proj{\acword}{S}=\proj{\acword'}{S} \textit{iff } \acword =
  \acword'$.
  Now, if $\acword_1 \, \alpha, \acword_2 \, \alpha,\acword\in\clang$
  satisfy the required conditions for CUI then either
  $\proj{\acword_1}{S}=\proj{\acword}{S}$ or
  $\proj{\acword_2}{S}=\proj{\acword}{S}$, since
  $\ptp[S] \in \ptpof[\aint]$.
  Hence $\cuui$ trivially holds.
  \finex
\end{example}
The language in \cref{ex:parenthesis} is non-regular since it has the
same structure of a language of well-balanced parenthesis.
Remarkably, this implies that the g-language cannot be expressed in
most of the other choreographic models in the literature.
There are models going beyond regular languages for both the binary
(e.g., \cite{Dardha_2014,tv16}) and the multi-party case
(e.g.,\cite{jy20}).
These approaches are based on process algebraic methods.
An interesting future research direction would be to study whether the
multi-party models give rise to CUI and branch-aware languages.
The argument used to show $\cuui$ in \cref{ex:parenthesis} proves the
following.
\begin{proposition}
  If there exists a participant involved in all the interactions of a
  g-language $\clang$ then $\cuui$ holds.
\end{proposition}


%% file: properties.tex
Besides correctness and completeness, other properties could be of
interest for message-passing systems.
For instance, one would like to ensure that participants eventually
interact, if they are willing to.
More generally, we are interested in some relations between the
interactions in a system and the communication actions of its
participants.
We consider a number of properties, defined as follows.

\bigskip

\emph{Harmonicity} (HA) requires that each sequence of communications
that a participant is able to perform can be executed in some
computation of the system.
\begin{definition}[Harmonicity]\label{def:llive}
  A system $\aCS$ on $\ptpset$ is \emph{harmonic} if
  $\aCS(\p) \subseteq \proj{\ssem{\aCS}}{\p}$ for each participant
  $\p \in \ptpset$.
\end{definition}

The remaining communication properties rely on the notion of maximal
computation (cf. beginning of \cref{sec:c-lang} on page
\pageref{page:maximal}).
\emph{Lock-freedom} (LF) requires that if a participant has pending
communications to make on an ongoing computation, then there is a
continuation of the computation involving that participant.
\begin{definition}[Lock-freedom]\label{def:wlive}
  A system $\aCS$ on $\ptpset$ is \emph{lock free} if, for each finite
  word $\acword \in \ssem{\aCS}$ and participant $\p \in \ptpset$, if
  $\proj{\acword}{\p}$ is not maximal in $\aCS(\p)$ then there is a
  word $\acword'$ such that $\acword\acword' \in \ssem{\aCS}$ and
  $\proj{\acword'}{\p}\neq\epsilon$.
\end{definition}

\emph{Strong lock-freedom} (SLF) requires that if a participant has
pending communications to make on an ongoing computation, then each
maximal continuation of the computation involves that participant.
\begin{definition}[Strong lock-freedom]\label{def:slf}
  A system $\aCS$ on $\ptpset$ is \emph{strongly lock free} if, for
  each finite word $\acword \in \ssem{\aCS}$ and participant
  $\p \in \ptpset$, if $\proj{\acword}{\p}$ is not maximal in
  $\aCS(\p)$ then for each word $\acword'$ such that $\acword\acword'$
  is maximal in $\ssem{\aCS}$ we have
  $\proj{\acword'}{\p}\neq\epsilon$.
\end{definition}
  
\emph{Starvation-freedom} (SF) requires that if a participant has
pending communications to make on an ongoing computation, then each
infinite continuation of the computation involves that participant.
\begin{definition}[Starvation-freedom]\label{def:lf}
  A system $\aCS$ on $\ptpset$ is \emph{starvation free} if, for each
  finite word $\acword \in \ssem{\aCS}$ and participant
  $\p \in \ptpset$, if $\proj{\acword}{\p}$ is not maximal in
  $\aCS(\p)$ then $\proj{\acword'}{\p}\neq\emptyword$ for each
  infinite word $\acword'$ such that
  $\acword \acword' \in \ssem{\aCS}$.
\end{definition}

\emph{Deadlock-freedom} (DF) requires that in all completed
computations each participant has no pending actions.
\begin{definition}[Deadlock-freedom]\label{def:deadlockfree} 
  A system $\aCS$ on $\ptpset$ is \emph{deadlock free} if, for each
  finite and maximal word $\acword \in \ssem{\aCS}$ and participant
  $\p \in \ptpset$, $\proj{\acword}{A}$ is maximal in $\aCS(\p)$.
\end{definition}

Barred for harmonicity, these properties appear in the literature
under different names in various contexts.
For instance, the notion of lock-freedom in~\cite{BDLT20} corresponds
to ours, which in turn corresponds to the notion of liveness
in~\cite{LangeNTY17,KobayashiS10} in a channel-based synchronous
communication setting.
Likewise, the notion of strong lock-freedom in~\cite{SeveriD19}
corresponds to ours and, under fair scheduling, to the notion of
lock-freedom in~\cite{Kobayashy:InfoComp02}.
As a final example, the definition of deadlock-freedom
in its (equivalent) contrapositive form, coincides with the notion
of progress as defined for synchronous processes
in~\cite{Padovani13,GhilezanJPSY19}.
Harmonicity, introduced in the present paper, assures that no
behaviour of a participant can be taken out from a system without
affecting the overall behaviour of the system itself.
Notice that the inverse of harmonicity,
$\proj{\ssem{\aCS}}{\p} \subseteq \aCS(\p)$, holds by construction.

The next proposition highlights the relations among our
properties.
\begin{restatable}{proposition}{systemproperties}\label{prop:systemproperties}
  The following relations hold among the properties in
  \cref{def:llive,def:wlive,def:slf,def:lf,def:deadlockfree}
  \[
	 \begin{tikzpicture}[node distance=1cm and 1.5cm,scale=1,
		every node/.style={transform shape}]
		\node (lf) {\textbf{LF}};
		\node[left = of lf] (slf) {\textbf{SLF}};
		\node[right = of lf,xshift=2mm] (df) {\textbf{DF}};
		\node[below = of lf,yshift=4mm] (ha) {\textbf{HA}};
		\node[below  = of df,yshift=4mm] (sf) {\textbf{SF}};
		\path (lf) edge[->,double,draw,bend left=10] (df);
		\path (df) edge[->,double,draw,bend left=10] node{/} (lf);
		\path[dashed,draw] (df) -- (sf);
		\path[dashed,draw] (ha) -- (sf);
		\path[dashed,draw] (ha) -- (df);
		\path[dashed,draw] (lf) -- (sf);
		\path[dashed,draw] (lf) -- (ha);
		\path (slf) edge[->,double,draw,bend left=10] (lf);
		\path (lf) edge[->,double,draw,bend left=10] node{/} (slf);
	 \end{tikzpicture}
  \]
  where implication does not hold in any direction between properties
  connected by dashed lines.
  Moreover, $\textsf{DF} \wedge \textsf{SF} \Leftrightarrow \textsf{SLF}$.
\end{restatable}
\begin{proof}
  \proofcase{ \textsf{SLF $\implies$ LF}} Let $\p \in \ptpset$,
  $\acword \in \ssem{\aCS}$ be finite, and $\proj{\acword}{A}$ not to
  be maximal in $\aCS(\p)$.  By \textsf{SLF} we can infer that also
  $\acword$ is not maximal in $\ssem{\aCS}$.  Otherwise, for
  $\acword'=\emptyword$, we would have $\acword\cat\acword'$ to be
  maximal in $\ssem{\aCS}$ and, by \textsf{SLF}, we would get
  $\emptyword=\proj{\acword'}{A}\neq\emptyword$. Contradiction. So, if
  $\acword$ is not maximal in $\ssem{\aCS}$, there exists
  $\acword'\neq\emptyword$ such that $\acword\cat\acword'$ is maximal
  in $\ssem{\aCS}$. We get hence immediately
  $\proj{\acword'}{A}\neq\emptyword$ by \textsf{SLF}.
\proofcase{ \textsf{LF $\centernot\implies$ SLF}}
Let us take the communicating system $\aCS = (\alang_{\ptp[X]})_{\ptp[X] \in \Set{\p,\q,\ptp[C]}}$,
where
  \[
	 \alang_{\p} = \pref[\Set{\aout[A][B]\cat\aout[A][B],\, \aout[A][B]\cat\aout[A][C]}] \quad
	 \alang_{\q} = \pref[\Set{\ain[A][B]\cat\ain[A][B]}] \quad
	 \alang_{\ptp[C]} = \pref[\Set{\ain[A][C]}]
  \]
  and whose synchronous semantics is
  \[
	 \ssem{\aCS} = \pref[\Set{\gint[][A][m][B]\cat\gint[][A][m][B],\, \gint[][A][m][B]\cat\gint[][A][m][C]}]
  \]
  $\aCS$ is lock-free, but not strongly lock-free. In fact, for
  $\gint[][A][m][B]$ and $\q$, we have that
  $\gint[][A][m][B]\cat\gint[][A][m][C]$ is maximal, but
  $\proj{(\gint[][A][m][C])}{B}=\emptyword$.
  \proofcase{ \textsf{LF $\implies$ DF}} Let us assume $\aCS$ to be
  lock-free. By contradiction let us assume $\aCS$ not to be
  deadlock-free.
  Then there is a finite and maximal word $\acword \in \ssem{\aCS}$
  and a participant \p\ such that $\proj{\acword}{A}$ is not maximal
  in $\aCS(\p)$.
  Since $\aCS$ is lock-free, by definition, there is a word $\acword'$
  such that $\acword\acword' \in \ssem{\aCS}$ and
  $\proj{\acword'}{\p}\neq\epsilon$.
  Hence $\acword'\neq\emptyword$ and therefore $\acword$ is not
  maximal in $\aCS$, contrary to our assumption.

  \proofcase{ \textsf{DF $\centernot\implies$ LF}}
  Let us take
  \[
	 \alang_{\p} = \pref[\Set{(\aout[A][B])^\omega}] \quad
	 \alang_{\q} = \pref[\Set{(\ain[A][B])^\omega}] \quad
	 \alang_{\ptp[C]} = \pref[\Set{\ain[A][C]}]
  \]
  and consider the communicating system
  \[
	 \aCS = (\alang_{\ptp[X]})_{\ptp[X] \in \Set{\p,\q,\ptp[C]}}
	 \qand[whose synchronous semantics is]
	 \ssem{\aCS} = \pref[\Set{(\gint[][A][m][B])^\omega}]
  \]
  It is immediate to check that $\aCS$ is vacuously deadlock free,
  since there is no finite maximal word in $\aCS$.
  However, $\aCS$ is not lock-free.  It is enough to take
  $\acword=\emptyword$, which is finite in $\ssem{\aCS}$ and such that
  $\proj{\acword}{C}=\emptyword$ is not maximal in $\alang_{\ptp[C]}$.
  However there is no $\acword'$ such that
  $\acword \cat \acword' \in \ssem \aCS$ and
  $\proj{\acword'}{C} \neq \emptyword$.

   \proofcase{ \textsf{DF $\centernot\implies$ SF}} Let
	\begin{eqnarray*}
	 \alang_{\p} &=& \pref[\Set{  (\aout[A][B])^\omega,\ \aout[A][B] \cat (\aout[A][C])^\omega}]
	 \\
	 \alang_{\q} &=& \pref[\Set{  (\ain[A][B])^\omega}]
	 \\
	 \alang_{\ptp[C]} &=& \pref[\Set{ (\ain[A][C])^\omega }]
  \end{eqnarray*}
  and consider the communicating system
  \[
	 \aCS = (\alang_{\ptp[X]})_{\ptp[X] \in \Set{\p,\q,\ptp[C]}}
  \]
  whose synchronous semantics is
  \[
	 \ssem{\aCS} = \pref[\Set{(\gint[][A][m][B])^\omega, \
		\gint[][A][m][B]\cat(\gint[][A][m][C])^\omega}]
  \]
  It is immediate to check that $\aCS$ is vacuously deadlock-free,
  since there is no finite maximal word in $\aCS$.
  However, $\aCS$ is not starvation-free.
  It is enough to take $\acword=\gint[][A][m][B]$, which is finite in
  $\ssem{\aCS}$ and also $\proj{\acword}{C}=\emptyword$ is not maximal
  in $\alang_{\ptp[C]}$, but for $\acword' =(\gint[][A][m][B])^\omega$
  we have that $\acword\cat\acword'\in\ssem{\aCS}$ and
  $\proj{\acword'}{C}=\emptyword$.

   \proofcase{ \textsf{SF $\centernot\implies$ DF}}
  Let
  \[
	 \alang_{\p} = \Set{\emptyword, \ \aout[A][B]}
	 \quad
	 \alang_{\q} = \Set{\emptyword,\ \ain[A][B],\ \aout[B][C]}
	 \quad
	 \alang_{\ptp[C]} = \Set{\emptyword,\ \ain[B][C]}
  \]
  and consider the communicating system
  \[
	 \aCS = (\alang_{\ptp[X]})_{\ptp[X] \in \Set{\p,\q,\ptp[C]}}
	 \qand[whose synchronous semantics is]
	 \ssem{\aCS} = \Set{\emptyword,\ \gint[][A][m][B], \ \gint[][B][m][C]}
  \]
  It is immediate to check that $\aCS$ is vacuously starvation free,
  since there is no infinite word in $\ssem{\aCS}$.
  However, $\aCS$ is not deadlock free.
  It is enough to take $\acword=\gint[][A][m][B]$, which is finite and
  maximal in $\ssem{\aCS}$ but
  $\proj{(\gint[][A][m][B])}{C} = \emptyword$ is not maximal in
  $\alang_{\ptp[C]}$.
	 
   \proofcase{ \textsf{LF $\centernot\implies$ HA}} Let
  \[
	 \alang_{\p} = \pref[\Set{\ain[B][A] \cat \ain[B][A],\ \ain[B][A] \cat \aout[A][B][][y]}]
	 \qqand
	 \alang_{\q} = \pref[\Set{ \aout[B][A] \cat \aout[B][A]}]
  \]
  and consider the communicating system
  \[
	 \aCS = (\alang_{\ptp[X]})_{\ptp[X] \in \Set{\p,\q}}
	 \qand[whose synchronous semantics is]
	 \ssem{\aCS} = \pref[\Set{\gint[][B][m][A] \cat \gint[][B][m][A]}]
  \]
  It is easy to check that $\aCS$ is lock-free.
  However, $\aCS$ is not harmonic, since
  $\ain[B][A]\, \cat\, \aout[A][B][][y] \in \aCS(\p)$ and there is no
  $\acword \in \ssem{\aCS}$ such that
  $\proj{\acword}{A} = \ain[B][A] \cat \aout[A][B][][y]$.
	 
   \proofcase{ \textsf{HA $\centernot\implies$ LF}} Let us consider the
  system $\aCS = (\alang_{\ptp[X]})_{\ptp[X] \in \Set{\p,\q,\ptp[C]}}$
  where
  \begin{align*}
	 \alang_{\p} & = \pref[\Set{ \ain[C][A] \cat \aout[A][C]}]
	 \\
	 \alang_{\q} & = \pref[\Set{\aout[B][C]}]
	 \\
	 \alang_{\ptp[C]} & = \pref[\Set{\aout[C][A],\ \ain[B][C] \cat \aout[C][A] \cat \ain[A][C]}]
  \end{align*}
  It is easy to check that
  \[
	 \ssem{\aCS} = \pref[\Set{\gint[][C][m][A],\  \gint[][B][m][C] \cat \gint[][C][m][A] \cat \gint[][A][m][C]}]
  \]
  and that $\aCS$ is harmonic.
  However, $\aCS$ is not lock-free.
  In fact, by taking $\acword = \gint[][C][m][A] \in \ssem \aCS$, we
  have that $\proj{\acword}{A}$ is not maximal in $\clang_{\p}$, since
  $\ain[C][A] \cat \aout[A][C] = \proj{\acword}{A} \cat \aout[A][C]
  \in \clang_{\p}$.
  However, there is no word $\acword'$ such that
  $\acword \cat \acword' \in \ssem{\aCS}$ and
  $\proj{\acword'}{A} \neq \epsilon$.
	 
   \proofcase{ \textsf{SF $\centernot\implies$ LF}} Immediate, since
  otherwise, by
  $\mathsf{LF} \implies \mathsf{DF}$ proved above, we would get
  $\mathsf{SF}$ to imply $\mathsf{DF}$, which we showed not to hold.
  
   \proofcase{ \textsf{LF $\centernot\implies$ SF}} Let us consider
  $\aCS = (\clang_{\ptp[X]})_{\ptp[X] \in \Set{\p,\q,\ptp[C]}}$ where
  \begin{align*}
	 \alang_{\p} & = \pref[\Set{(\aout[A][B])^n\cat \aout[A][C]\cat(\aout[A][B])^\omega \sst n\in\Nat}\cup\Set{(\aout[A][B])^\omega}]
	 \\
	 \alang_{\q} & = \pref[\Set{(\ain[A][B])^\omega}]
	 \\
	 \alang_{\ptp[C]} & = \pref[\Set{\ain[A][C]}]
  \end{align*}
  It is easy to check that
  \[
	 \ssem \aCS = \pref[\Set{(\gint[][A][m][B])^n\cat\gint[][A][m][C]\cat(\gint[][A][m][B])^\omega \sst n\in\Nat}\cup\Set{(\gint[][A][m][B])^\omega}]
  \]
  It is not difficult also to check that $\aCS$ is lock-free.
  However, it is not starvation-free.
  In fact, by taking the non maximal word $\acword=\emptyword$, we
  have that for the infinite word
  $\acword' = (\gint[][A][m][B])^\omega$ and for the participant
  $\ptp[C]$ we have that $\proj{\acword}{C}$ is non maximal and
  $\acword\acword'\in\ssem \aCS $.
  However $\proj{\acword'}{C}=\emptyword$.
	 
   \proofcase{ \textsf{SF $\centernot\implies$ HA}} Let
  \[
	 \alang_{\p} = \pref[\Set{\ain[B][A] \cat \ain[B][A],\ \ain[B][A] \cat \aout[A][B][][y]}]
	 \qqand
	 \alang_{\q} = \pref[\Set{ \aout[B][A] \cat \aout[B][A]}]
  \]
  and consider the communicating system
  \[
	 \aCS = (\alang_{\ptp[X]})_{\ptp[X] \in \Set{\p,\q}}
	 \qand[whose synchronous semantics is]
	 \ssem{\aCS} = \pref[\Set{\gint[][B][m][A] \cat \gint[][B][m][A]}]
  \]
  $\aCS$ is trivially starvation-free, since it contains no infinite
  word.
  However, $\aCS$ is not harmonic, since
  $\ain[B][A] \cat \aout[A][B][][y] \in \aCS(\p)$ and there is no
  $\acword \in \ssem{\aCS}$ such that
  $\proj{\acword}{A} = \ain[B][A] \cat \aout[A][B][][y]$.

   \proofcase{ \textsf{HA $\centernot\implies$ SF}} Let us consider
  $\aCS = (\clang_{\ptp[X]})_{\ptp[X] \in \Set{\p,\q,\ptp[C]}}$ where
  \begin{align*}
	 \alang_{\p} & = \pref[\Set{(\aout[A][B])^n\cat \aout[A][C]\cat(\aout[A][B])^\omega \sst n\in\Nat}\cup\Set{(\aout[A][B])^\omega}]
	 \\
	 \alang_{\q} & = \pref[\Set{(\ain[A][B])^\omega}]
	 \\
	 \alang_{\ptp[C]} & = \pref[\Set{\ain[A][C]}]
  \end{align*}
  It is easy to check that
  \[
	 \ssem \aCS = \pref[\Set{(\gint[][A][m][B])^n\cat\gint[][A][m][C]\cat(\gint[][A][m][B])^\omega \sst n\in\Nat}\cup\Set{(\gint[][A][m][B])^\omega}]
  \]
  It is not difficult also to check that $\aCS$ is harmonic.
  However, it is not starvation-free.
  In fact, by taking the non maximal word $\acword=\emptyword$, we
  have that for the infinite word
  $\acword' = (\gint[][A][m][B])^\omega$ and for the participant
  $\ptp[C]$ we have that $\proj{\acword}{C}$ is non maximal and
  $\acword\acword'\in\ssem \aCS $.
  However $\proj{\acword'}{C}=\emptyword$.

   \proofcase{ \textsf{DF $\centernot\implies$ HA}} Let
  \[
	 \alang_{\p} = \Set{\emptyword,\ \ain[B][A],\ \aout[A][B][][y]}
	 \qqand
	 \alang_{\q} = \Set{\emptyword,\ \aout[B][A]}
  \]
  and consider the communicating system
  \[
	 \aCS = (\alang_{\ptp[X]})_{\ptp[X] \in \Set{\p,\q}}
	 \qand[whose synchronous semantics is]
	 \ssem{\aCS} = \Set{\emptyword,\ \gint[][B][m][A]}
  \]
  $\aCS$ is deadlock-free, since the only finite maximal word is
  $\gint[][B][m][A]$ whose projection on $\p$ and $\q$ are both
  maximal.
  However, $\aCS$ is not harmonic, since
  $\aout[A][B][][y] \in \aCS(\p)$ and there is no
  $\acword \in \ssem{\aCS}$ such that
  $\proj{\acword}{A} =\aout[A][B][][y]$.

   \proofcase{ \textsf{HA $\centernot\implies$ DF}} Let
  \[
	 \alang_{\p} = \Set{\emptyword,\ \aout[A][B],\ \aout[A][C]}
	 \qqand
	 \alang_{\q} = \Set{\emptyword,\ \ain[A][B]}
	 \qqand
	 \alang_{\ptp[C]} = \Set{\emptyword,\ \ain[A][C]}
  \]
  and consider the communicating system
  \[
	 \aCS = (\alang_{\ptp[X]})_{\ptp[X] \in \Set{\p,\q,\ptp[C]}}
	 \qand[whose synchronous semantics is]
	 \ssem{\aCS} = \Set{\emptyword,\ \gint[][A][m][B],\ \gint[][A][m][C]}
  \]
  It is easy to check that $\aCS$ is harmonic.
  However, $\aCS$ is not deadlock-free, since for the maximal and
  finite word $\gint[][A][m][B]\in \aCS$ we have that
  $\proj{\acword}{C}$ ($=\emptyword$) is not maximal in
  $\alang_{\ptp[C]}$.
	 
   \proofcase{ \textsf{DF $\wedge$ SF $\implies$ SLF}} In order to show SLF, let
  $\acword \in \ssem{\aCS}$ be finite and let $\proj{\acword}{\p}$ be
  non maximal in $\aCS(\p)$ for $\p \in \ptpset$.
  Besides, let $\acword'$ be such that $\acword\acword'$ is maximal in
  $\ssem{\aCS}$.
  We have to show that $\proj{\acword'}{\p}\neq\emptyword$.
  We distinguish two cases:
  \begin{description}
  \item[$\acword'$ is finite]  Since $\aCS$ is DF we can infer that
	 $\acword'\neq\emptyword$ and $\proj{\acword'}{\p}\neq\epsilon$.
  \item[$\acword'$ is infinite]  We get
	 $\proj{\acword'}{\p}\neq\epsilon$ immediately by SF.
  \end{description}
  
   \proofcase{ \textsf{SLF $\implies$ DF $\wedge$ SF}} To show DF by contradiction, let us
  assume to have $\acword \in \ssem{\aCS}$ finite and maximal and such
  that $\proj{\acword}{\p}$ is non maximal in $\aCS(\p)$ for
  $\p \in \ptpset$.
  By SLF we have that $\proj{\acword'}{\p}\neq\emptyword$ for each
  $\acword'$ such that $\acword\acword'$ is maximal in $\ssem{\aCS}$.
  Since we assumed $\acword$ to be maximal, the only
  possible $\acword'$ is $\emptyword$, contradicting $\proj{\acword'}{\p}\neq\emptyword$.
  To show SF by contradiction, let us assume to have
  $\acword \in \ssem{\aCS}$ finite and such that $\proj{\acword}{\p}$
  is non maximal in $\aCS(\p)$ for $\p \in \ptpset$.
  Besides, let us assume that there exists an infinite $\acword'$ such
  that $\acword\acword'\in\ssem{\aCS}$ and
  $\proj{\acword'}{\p} = \emptyword$.
  Since $\acword'$ is infinite, $\acword\acword'$ is trivially maximal
  in $\ssem{\aCS}$ and we get a contradiction with SLF.
\end{proof}


%% file: proj-properties.tex
Harmonicity (cf. \cref{def:llive}) is the only property 
guaranteed on any system obtained via projection.
This is a simple consequence of \cref{th:completeness}.
\begin{restatable}{corollary}{wfliveb}\label{thm:wfliveb}
  If $\clang$ is a g-language then
  $\proj{\clang}{}$ is harmonic.
\end{restatable}
\begin{proof}
  By \cref{th:completeness},
  $\clang\subseteq\ssem{\proj{\clang}{}}$. Now, by monotonicity of
  projection, we get
  $\proj{\clang}{}\subseteq\proj{\ssem{\proj{\clang}{}}}{}$, that is
  harmonicity of $\proj{\clang}{}$.
\end{proof}
To ensure the other properties on a system ${\proj{\clang}{}}$ we need
to require some conditions on the g-language $\clang$.
Basically, we will strengthen CUI which is too weak.
For instance, $\cuui$ does imply neither deadlock-freedom nor
lock-freedom for $\proj{\clang}{}$.
\begin{example}[CUI $\centernot\implies$ DF, LF]\label{ex:closnodl}
  It is easy to check that $\cuui$ holds for the g-language
  $\clang = \pref[\Set{\acword, \acword'}]$ where
  \[
	 \acword = \gint[][A][l][C] \cat \gint[][A][@][B]
	 \cat \gint[][A][m][C] \qqand \acword' = \gint[][A][r][C] \cat
	 \gint[][A][@][B] \cat \gint[][B][@][C]
  \]
  Informally, $\clang$ is CUI because $\ptp[C]$ can ascertain
  which of its last actions to execute from the first input.
  So, \cref{th:completeness,th:correctness} ensure that
  $\clang = \ssem{\proj{\clang}{}}$.
  However, $\proj{\clang}{}$ is not deadlock-free.
  In particular, $\acword \in \clang = \ssem{\proj \clang {}}$ is a
  deadlock since it is a finite maximal word whose projection on $\q$,
  namely $\proj{\acword}{B} = \ain$, is not maximal in
  $\proj{\clang}{B}$ because
  $\proj{\acword'}{B}= \ain[A][B] \cat \aout[B][C] \in
  \proj{\clang}{B}$.
  $\proj{\clang}{}$ is non lock-free either by
  \cref{prop:systemproperties}.
  \finex
\end{example}

In many models (cf.~\cite{hlvlcmmprttz16}) in order to ensure, besides
other properties, also the correctness of $\proj{\clang}{}$, a
condition called \emph{well-branchedness} is required.
We identify a notion weaker than well-branchedness, which by analogy
we dub \emph{branch-awareness} (BA for short).
\begin{definition}[Branch-awareness]\label{def:ba}
  A participant $\px$ \emph{distinguishes} two g-words
  $\acword_1, \acword_2 \in \alfint^\infty$ if
  \[
	 \proj{{\acword_1}}{X} \neq \proj{{\acword_2}}{X}
	 \quad\text{ and }\quad
	 \proj{{\acword_1}}{X} \not\prec \proj{{\acword_2}}{X}
	 \quad\text{ and }\quad
	 \proj{{\acword_2}}{X} \not\prec \proj{{\acword_1}}{X}.
  \]
  A g-language $\clang$ on $\ptpset$ is \emph{branch-aware} if each
  $\px \in \ptpset$ distinguishes all maximal words in
  $\clang$ whose projections on $\px$ differ.
\end{definition}
\begin{example}
  The language $\clang = \pref[\Set{ \acword, \acword'}]$ with
  $\acword = \gint[][A][l][C] \cat \gint[][A][@][B] \cat
  \gint[][A][m][C]$ and
  $\acword' = \gint[][A][r][C] \cat \gint[][A][@][B] \cat
  \gint[][B][@][C]$ from \cref{ex:closnodl} is not branch-aware, since
  $\proj{\acword}{B}=\ain$ and
  $\proj{\acword'}{B}=\ain[A][B] \cat \aout[B][C]$, hence
  $\proj{\acword}{B} \neq \proj{\acword'}{B}$ but
  $\proj{\acword}{B} \prec \proj{\acword'}{B}$.
  \finex
\end{example}
Condition $\proj{{\acword_1}}{{\px}} \neq \proj{{\acword_2}}{{\px}}$
in \cref{def:ba} is not strictly needed to define BA, but it makes the
notion of \squo{distinguishes} more intuitive.
Equivalently, as shown in \cref{prop:distinguish} below, a participant
\px\ distinguishes two branches if, after a common prefix, \px\ is
actively involved in both branches, performing different interactions.
\begin{restatable}{proposition}{distinguish}\label{prop:distinguish}
  Participant $\px$ distinguishes two g-words
  $\acword_1, \acword_2 \in \alfint^\infty$ iff there are
  $\acword'_1 \cat \aint_1 \preceq \acword_1$ and
  $\acword'_2 \cat \aint_2 \preceq \acword_2$ such that
  $ \proj{{\acword'_1}}{\px} = \proj{{\acword'_2}}{\px}$ and
  $\proj{{\alpha_1}}{\px} \neq \proj{{\alpha_2}}{\px}$.
\end{restatable}
\begin{proof}
  Trivial.
\end{proof}
  
The notions of well-branchedness in the
literature~\cite{hlvlcmmprttz16} additionally impose that
$\proj{{\alpha_1}}{X}$ and $\proj{{\alpha_2}}{X}$ in the above
proposition are input actions, but for a (unique) participant (\aka\
the \emph{selector}) which is required to have different outputs.
  In our case the notion of selector corresponds to the \quo{first}
  participant that distinguishes two words.
  Also in our case a selector must be involved in each branch but, due
  to the perfect symmetry of input and output actions in synchronous
  communications, its involvement can happen through input or output
  actions.
  This is illustrated by the following example.
  \begin{example}[Selector and input actions]
	 Consider the words
	 \[
		\acword = \gint[] \cat \gint[][a][m][c]
		\qqand
		\acword' = \gint[][b][m][c] \cat \gint[][c][m][a]
	 \]
	 The participant \q\ immediately distinguishes $\acword$ and
	 $\acword'$ via its first actions, that is the input from \p\
	 and the output to $\p[c]$.
	 Hence, \q\ is the selector of the branch made of $\acword$ and
	 $\acword'$.
	 Notice that \p\ and $\p[c]$ also distinguish these two words,
	 however this happens \quo{later}.
	 Finally, observe that the same would hold if we replace
	 $\gint[][b][m][c]$ with $\gint[][c][m][b]$ in $\acword'$.
	 \finex
  \end{example}

In our theory, BA is not needed for correctness, but it is
nevertheless useful to prove the communication properties presented in
\cref{sec:prop}.
\begin{restatable}[Consequences of BA]{theorem}{baconseq}\label{thm:baconseq}
  Let $\clang$ be a branch-aware and CUI \sclang. Then
  ${\proj{\clang}{}}$ satisfies all the properties in
  \cref{def:llive,def:wlive,def:slf,def:lf,def:deadlockfree}.
  \end{restatable}
\begin{proof}
  Let $\clang$ be a branch-aware \sclang
  such that $\cuui$ holds.
  We prove the properties separately.
  \proofcase{Harmonicity} Immediate by \cref{thm:wfliveb}.
  \proofcase{Lock-freedom} By contradiction, let us assume
  $\proj{\clang}{}$ not to be lock-free.
  By \cref{def:wlive} and $\cuui$, it follows that there exist a
  participant $\p\in\ptpset$ and a finite g-word $\acword\in\clang$
  such that
  \begin{itemize}
  \item $\proj{\acword}{A}$ is not maximal in $\proj{\clang}{A}$;
  \item for all $\acword'$ if $\acword \cat \acword' \in \clang$ then
	 $\proj{\acword'}{A}=\emptyword$.
  \end{itemize}
  By \cref{thm:wfliveb}, $\proj{\clang}{}$ is harmonic.  Hence, by the
  above and $\cuui$, there exists $\acword''\in\clang$ such that
  \begin{itemize}
  \item $\acword''\neq\acword$;
  \item $\proj{\acword''}{A} = \proj{\acword}{A}$;
  \item there is $\hat \acword$ such that
	 $\acword'' \cat \hat \acword \in \clang$ and
	 $\proj{\hat\acword}{A}\neq\emptyword$.
  \end{itemize}
  This means that, by taking a maximal extension of $\acword$ and a
  maximal extension of $\acword''\hat \acword$, we would get two non
  branch-aware words in $\clang$, contradicting our hypothesis of
  $\clang$ being branch-aware.
  \proofcase{Deadlock-freedom} Immediate by
  \cref{prop:systemproperties}.
  \proofcase{Starvation-freedom} By contradiction, let us assume
  $\proj{\clang}{}$ not to be starvation-free.
  By \cref{def:lf} and $\cuui$, it follows that there exist a
  participant $\p\in\ptpset$ and a finite g-word $\acword \in \clang$
  such that
  \begin{itemize}
  \item $\proj{\acword}{A}$ is not maximal in $\proj{\clang}{A}$;
  \item there is an infinite word $\acword'$ such that
	 $\acword \cat \acword' \in \clang$ and
	 $\proj{\acword'}{A} = \emptyword$.
  \end{itemize}
  Now, by harmonicity of $\proj{\clang}{}$ (\cref{thm:wfliveb}), non
  maximality of $\proj{\acword}{A}$ and by $\cuui$, it follows that
  there exist $\acword''\in\clang$ and a finite word $\hat\acword$
  such that
  \begin{itemize}
  \item $\acword''\hat \acword\in\clang$;
  \item $\proj{\acword''}{A} = \proj{\acword}{A}$;
  \item $\proj{\hat\acword}{A}\neq\emptyword$.
  \end{itemize}
  The above means that by taking $\acword\acword'$ and any maximal
  extension of $\acword''\hat\acword$ we would get two maximal words
  in $\clang$ which $\p$ cannot distinguish, so contradicting our
  hypothesis of $\clang$ being branch-aware.
  \proofcase{Strong lock-freedom} Immediate since
  $\mathsf{SLF} = \mathsf{SF} \wedge \mathsf{DF}$.
\end{proof}

\begin{example}[Task dispatching and branch-awareness]\label{ex:parenthesis2}
  In order to show that the g-language $\clang$ in
  \cref{ex:parenthesis} is branch-aware, we first notice that each
  maximal word in $\clang$ ends with the interactions
  $\gint[][s][s][d] \cat \gint[][s][s][h]$.  If $\clang$ were not
  branch-aware, there should be two maximal words
  $\acword \cat \gint[][s][s][d] \cat \gint[][s][s][h]$ and
  $\acword' \cat \gint[][s][s][d] \cat \gint[][s][s][h]$ and a
  participant $\ptp[X]\in\ptpof[\clang]$ such that
  $\proj{(\acword \cat \gint[][s][s][d] \cat \gint[][s][s][h])}{X}
  \prec \proj{(\acword' \cat \gint[][s][s][d] \cat
	 \gint[][s][s][h])}{X}$. This is impossible, since $\acword$ and
  $\acword'$ are both generated by the non terminal symbol $L'$ and
  hence cannot contain the message $\msg[s]$.
  \finex
\end{example}

\cref{prop:systemproperties} refines as follows when restricting to
projections of g-languages.

\begin{restatable}{proposition}{systempropertiesbis}\label{prop:systemproperties2}
  When considering only systems which are projections of g-languages
  the following relations hold among the properties in
  \cref{def:llive,def:wlive,def:slf,def:lf,def:deadlockfree}
  \[
	 \begin{tikzpicture}[node distance=1cm and 1.5cm,scale=1,
		every node/.style={transform shape}]
		\node (lf) {\textsf{LF}};
		\node[left = of lf] (slf) {\textbf{SLF}};
		\node[right = of lf,xshift=2mm] (df) {\textsf{DF}};
		\node[below = of lf,yshift=4mm] (ha) {\textsf{HA}};
		\node[below  = of df,yshift=4mm] (sf) {\textsf{SF}};
		\path (lf) edge[->,double,draw,bend left=10] (df);
		\path (df) edge[->,double,draw,bend left=10] node{/} (lf);
		\path[dashed,draw] (df) -- (sf);
		\path (sf) edge[->,double,draw,bend left=15] (ha);
		\path (ha) edge[->,double,draw,bend left=10] node{/} (sf);
		\path (df) edge[->,double,draw,bend left=10] (ha);
		\path (ha) edge[->,double,draw,bend left=10] node{/} (df);
		\path[dashed,draw] (lf) -- (sf);
		\path (lf) edge[->,double,draw,bend left=15] (ha);
		\path (ha) edge[->,double,draw,bend left=15] node{/} (lf);
		\path (slf) edge[->,double,draw,bend left=10] (lf);
		\path (lf) edge[->,double,draw,bend left=10] node{/} (slf);
	 \end{tikzpicture}
\]
  where implication does not hold in any direction between properties
  connected by dashed lines.
  Moreover, $\textsf{DF} \wedge \textsf{SF} \Leftrightarrow \textsf{SLF}$.
\end{restatable}
\begin{proof}
\proofcase{ SLF $\implies$ LF} The proof of \cref{prop:systemproperties} applies.
\proofcase{ LF $\centernot\implies$ SLF} The proof of \cref{prop:systemproperties} applies.
  \proofcase{ LF $\implies$ DF} The proof of \cref{prop:systemproperties} applies.
  \proofcase{DF $\centernot\implies$ LF} Let us take the
  language
  \[
	 \clang = \pref[\Set{ \gint[][a][n][c] \cat \gint[][a][m][b] \cat \gint[][a][m][c] \cat (\gint[][a][m][c])^\omega,\
		\gint[][a][m][c] \cat \gint[][a][m][b] \cat \gint[][b][m][c] \cat (\gint[][a][m][c])^\omega}]
  \]
  The system $\proj{\clang}{}$ is trivially DF since its maximal words
  are all infinite, while it is not LF since the non maximal word
  $\gint[][A][n][C] \cat \gint[][A][m][B] \in \ssem{\proj{\clang}{}}$
  is such that its projection on $\q$ is not maximal in
  $\proj \clang \q$ and no continuation in $\clang$ of such word
  contains interactions involving \q; in fact for the (only) g-word
  $\hat \acword = \gint[][A][m][C] \cat (\gint[][a][m][c])^\omega$
  such that
  $\gint[][A][n][C] \cat \gint[][A][m][B] \cat \hat\acword \in
  \ssem{\proj{\clang}{}}$, we have $\proj{\hat\acword}{B}=\emptyword$.
  \proofcase{DF $\centernot\implies$ SF} The counterexample of
  \cref{prop:systemproperties} applies.
  \proofcase{SF $\centernot\implies$ DF} The counterexample of
  \cref{prop:systemproperties} applies.
  \proofcase{LF $\implies$ HA} Trivial since HA always holds
  thanks to \cref{thm:wfliveb}.
  \proofcase{HA $\centernot\implies$ LF} Trivial since HA
  always holds thanks to \cref{thm:wfliveb}, while this is not the
  case for LF thanks to \cref{ex:closnodl}.
  \proofcase{SF $\centernot\implies$ LF} The proof of
  \cref{prop:systemproperties} applies.
  \proofcase{LF $\centernot\implies$ SF} The counterexample of
  \cref{prop:systemproperties} applies.
  \proofcase{SF $\implies$ HA} Trivial since HA always holds
  thanks to \cref{thm:wfliveb}.
  \proofcase{HA $\centernot\implies$ SF} The counterexample of
  \cref{prop:systemproperties} applies.
  \proofcase{DF $\implies$ HA} Trivial since HA always holds
  thanks to \cref{thm:wfliveb}.
  \proofcase{HA $\centernot\implies$ DF} Trivial since HA
  always holds thanks to \cref{thm:wfliveb}, while this is not the
  case for DF thanks to \cref{ex:closnodl}.
  \proofcase{DF $\wedge$ SF $\implies$ SLF} The proof of
  \cref{prop:systemproperties} applies.
  \proofcase{SLF $\implies$ DF $\wedge$ SF} The proof of
  \cref{prop:systemproperties} applies.
 \end{proof}

 It is not difficult to show that branch-awareness actually
 characterises {\sf SLF} for systems obtained by projecting CUI
 languages.
 
\begin{restatable}[Branch-awareness characterises {\sf
	 SLF}]{proposition}{bachar}
 A CUI g-language $\clang$ is branch-aware iff $\proj{\clang}{}$ is strongly lock-free.
\end{restatable}
\begin{proof}
  Necessity follows from \cref{thm:baconseq} while
  for sufficiency we reason by contraposition.
  Assume $\clang$ not to be branch-aware.
  Then, by \cref{def:ba} there exist a participant $\px$ and two
  maximal words in $\clang$ such that their projections on $\px$
  differ and
  either $\proj{{\acword_1}}{X} \prec \proj{{\acword_2}}{X}$
  or $\proj{{\acword_2}}{X} \prec \proj{{\acword_1}}{X}$.
  Let us consider the first case (the second is analogous).
  Let now $\acword'_1$ and $\acword'_2$ be the maximal prefixes (hence
  finite) of, respectively, $\acword_1$ and $\acword_2$ such that
  $\proj{{\acword'_1}}{X} = \proj{{\acword'_2}}{X}$.
  We distinguish the following two cases.

  \proofcase{$\acword'_1 = \acword_1$} In such a case we get that
  $\proj{\clang}{}$ is not deadlock-free, since, for the finite
  maximal $\acword$ we have that $\proj{{\acword_1}}{X}$ is not
  maximal.  So, by \cref{prop:systemproperties2}, $\proj{\clang}{}$ is
  not strongly lock-free as well.
  \proofcase{$\acword'_1 \prec \acword_1$}
  In such a case we get that $\acword'_1$ is a non-maximal finite word
  such that $\acword_1 = \acword'_1\cat\acword''_1$ is maximal,
  $\proj{{\acword_1}}{X}$ is non maximal and
  $\proj{{\acword''_1}}{X}=\emptyword$, that is $\proj{\clang}{}$ is
  not strongly lock-free.
\end{proof}


%% file: cfsm.tex
\emph{Communicating finite-state machines} (CFSMs) have been
introduced in~\cite{bz83} as a convenient model to analyse
message-passing protocols. 
Basically, a CFSM is a finite-state automaton (FSA), defined below,
whose transitions are actions.

\begin{definition}[Finite state automaton (FSA)]
A \emph{finite state automaton} (FSA) is a tuple
$A = \conf{\sset, q_0, \lset, \tset}$ where
\begin{itemize}
\item $\sset$ is a finite set of states (ranged over by $q,s,\ldots$) and  $q_0 \in \sset$ is the \emph{initial state};
\item $\lset$ is a finite set of labels (ranged over by
  $\al,\lambda,\ldots$);
\item
  $\tset \subseteq \sset \times \lset \times
  \sset$ is a set of transitions.
\end{itemize}
\end{definition}
We use the usual notation $q_1 \arro{\lambda} q_2$ for the transition $(q_1,\lambda,q_2)\in\arro{}$, and
$q_1 \arro{} q_2$ when there exists $\lambda$ such that
$q_1 \arro{\lambda} q_2$, as well as $\arro{}^*$ for the reflexive and transitive
closure of $\arro{}$.

The set of \emph{reachable states} of $A$ is
$\RS[A] = \Set{q \sst q_0\arro{}^\star q}$.
\begin{remark}\label{rmk:fsa}
  Our definition of FSA omits the set of \emph{accepting} states since
  we consider only FSAs where each state is accepting.
  \rmkend
\end{remark}

Following the assumption in \cref{rmk:fsa}, we define the language
$\mlang(A)$ as the union of the finite words accepted by $A$ in the
classical sense and the infinite words accepted by $A$ considered as
B\"uchi automaton\footnote{A B\"uchi automaton accepts an infinite
	 word $\aword$ when it traverses infinitely often elements of the
	 set of accepting states while consuming the symbols of $\aword$.}
where all states are accepting.

A CFSM is an FSA labelled in $\alfact \cup \Set{\emptyword}$, where
$\emptyword$ is not a symbol in $\alfact$ and it overloads the
notation for the empty string to represent internal transitions of
CFSMs as usual in automata theory.
\begin{remark}
  FSAs, and consequently CFSMs, can be deterministic or not.
  Deterministic FSAs have no label $\emptyword$, and transitions from
  the same state are pairwise different.
  Given a non-deterministic FSA one can build a deterministic FSA
  generating the same language.
  We will assume that each CFSM is deterministic since we are
  interested in languages and non-deterministic CFSMs can be
  determinised while preserving their language.
  
  Notice that it is necessary to adopt finer notions of equivalence
  such as bisimilarity~\cite{par80} in order to tell apart
  deterministic FSAs from non-deterministic ones.
  It could be interesting to investigate the possibility of including
  the notion of nondeterminism in FCLs in the future.
  \rmkend
\end{remark}
A CFSM is \emph{local} to a participant \p\ (\p-local for short) if
all its transitions have subject \p; we will consider communicating
systems where the behaviour of each participant \p\ is specified by an
\p-local CFSM.
\renewcommand{\aCS}{\mathbf{S}}
Formally, given a finite set $\ptpset \subseteq \mathfrak{P}$ of
participants, a \emph{system of CFSMs} is a map
$(\aCM_{\p})_{\p \in \ptpset}$ assigning an \p-local CFSM $\aCM_{\p}$
to each participant $\p \in \ptpset$ such that any participant
occurring in a label of a transition of $\aCM_{\p}$ is in $\ptpset$.

The synchronous behaviour of a system of CFSMs
$(\aCM_{\p})_{\p \in \ptpset}$ has been defined in~\cite{blt20} as any
FSA where states are maps assigning a state in $\aCM_{\p}$ to each
$\p \in \ptpset$ and transitions are labelled by interactions.
Intuitively, given a configuration $\aConf$, if $\aCM_{\p}$ and
$\aCM_{\q}$ have respectively transitions
$\aConf(\p) \arro \aout q_{\p}'$ and $\aConf(\q) \arro \ain q_{\q}'$
then
$\aConf \arro \gint \aConf[\p \mapsto q_{\p}',\q \mapsto q_{\q}']$,
where $f[x \mapsto y]$ denotes the update of $f$ on $x$ with $y$.
The next definition is a slightly different version of the one
in~\cite{blt20}, where the semantics of a system of CFSMs $\aCS$ is
the FSA $\cssem \aCS$ defined below rather than its language as in our
case.
\begin{definition}[Synchronous semantics of systems of CFSMs]\label{def:syncSemaut}
  Let $\aCS = (\aCM_{\p})_{\p \in \ptpset}$ be a system of CFSMs where
  $\aCM_{\p} = \conf{\sset_{\p}, q_{0\p}, \alfact, \tset_{\p}}$ for
  each participant $\p \in \ptpset$.
  A \emph{synchronous configuration} for $\aCS$ is a map
  $\aConf = (q_{\p})_{\p \in \ptpset}$ assigning a \emph{local state}
  $q_{\p}\in \sset_{\p}$ to each $\p \in \ptpset$.

  The \emph{synchronous semantics} of $\aCS$ is the g-language
  $\mlang(\cssem \aCS)$ where
  $\cssem \aCS = \conf{\sset, \aConf_0, \alfint, \tset}$ is defined as
  follows:
  \begin{itemize}
  \item $\sset$ is the set of configurations of $\aCS$, as defined
	 above, and $\aConf_0: \p \mapsto q_{0\p}$ for each $\p \in
	 \ptpset$ is the \emph{initial} configuration of $\sset$
  \item $\tset$ is the set of transitions
	 $\aConf_1 \arro{\gint} \aConf_2$ such that
	 \begin{itemize}
	 \item $\aConf_1(\p) \arro{\aout} \aConf_2(\p)$ in $\aCM_{\p}$ and
		$\aConf_1(\q) \arro{\ain} \aConf_2(\q)$ in $\aCM_{\q}$, and
	 \item for all $\ptp[x] \in \ptpset \setminus \Set{\p,\q}$,
		$\aConf_1(\ptp[x]) = \aConf_2(\ptp[x])$.
	 \end{itemize}
  \end{itemize}
\end{definition}
  As we will see in \cref{sec:chor-automata}, the synchronous
  semantics in \cref{def:syncSemaut} is a \emph{choreography automaton}~\cite{blt20}.
  Note that this is not the case for the asynchronous semantics of
  communicating systems which is in general a transition system with
  infinitely many configurations.

An immediate relation between systems of CFSMs and our communicating
systems is that, given a system of CFSMs
$\aCS = (\aCM_{\ptp[X]})_{\ptp[X] \in \ptpset}$, we can define the
abstract system corresponding to $\aCS$ as
$\hat \aCS = (\mlang(\aCM_{\ptp[X]}))_{\ptp[X] \in \ptpset}$.
Unsurprisingly, the semantics of $\aCS$ and $\hat \aCS$ do coincide.
\begin{proposition}\label{prop:eqsems}
  For all systems $\aCS = (\aCM_{\ptp[X]})_{\ptp[X] \in \ptpset}$ of
  CFSMs, $\mlang(\cssem \aCS) = \ssem{\hat \aCS}$.
\end{proposition}
\begin{proof}
	 The proof shows the commutativity of the following
	 diagram
	 \[\begin{tikzpicture}[node distance = 1cm and 3cm]
		\node (S) {$\csystems$};
		\node (L) [below = of S] {$\mathfrak{L}$};
		\node (A) [right = of S] {$\mathfrak{C}$};
		\node (G) [below = of A] {$\glanguages$};
		\path[->] (S) edge[above] node{$\cssem \_$} (A)
		(S) edge[left] node{$\hat \_$} (L)
		(A) edge[right] node{$\mlang(\_)$} (G)
		(L) edge[below] node{$\ssem{\_}$} (G);
	 \end{tikzpicture}\]
  where $\csystems$ is the set of systems of CFSMs, $\mathfrak{L}$ is
  the set of communicating systems (cf. \cref{def:commSyst}),
  $\mathfrak{C}$ is the set of c-automata, and $\glanguages$ is the
  set of global languages.

  Let
  $M_{\ptp[X]}={\conf{\sset_{\ptp[X]},
		{q_0}_{\ptp[X]},\alfact,\tset_{\ptp[X]}}}$ and let
  $\cssem{\aCS} = \conf{\sset,\aConf_0,\alfint,\tset}$ where
  $\conf{\sset,\aConf_0,\alfint,\tset}$ is as in
  \cref{def:syncSemaut}.
  We prove the equality by proving separately the following two
  inclusions, for all $\aCS$.

  \proofcase{$\mlang(\cssem \aCS) \subseteq \ssem{\hat \aCS}$}
  By definition of $\hat \aCS$, it is enough to show that, for any
  $\acword\in\alfint^\infty$, if $\acword \in \mlang(\cssem \aCS)$,
  then for any $\ptp[X] \in \ptpset$, we have that
  $\proj \acword X \in \mlang(M_{\ptp[X]})$.
  Let $\acword \in \mlang(\cssem \aCS)$ and proceed by coinduction on
  the paths of machine $\cssem \aCS$.

  If $\acword = \emptyword$ the thesis follows immediately.
  Otherwise, let $\acword = (\gint[]) \cat \acword'$.
  By definition of recognised word, there is $\aConf_1\in \sset$ such
  that $\aConf_0\arro{\gint[]} \aConf_1$ and $\acword'$ is recognised
  by the automaton $\conf{\sset,\aConf_1,\alfint,\tset}$.\\
  Necessarily, by \cref{def:syncSemaut},
  \begin{enumerate}
  \item $\aConf_0(\p) \trans[\p] \aout \aConf_1(\p)$ and
	 $\aConf_0(\q) \trans[\q] \ain \aConf_1(\q)$, and
  \item for all $\ptp[x] \neq \p, \q$,
	 $\aConf_1(\ptp[x]) = \aConf_0(\ptp[x])$.
  \end{enumerate}
  It follows that $\acword' \in \mlang(\aCS')$ where
  $\aCS' = ({\conf{\sset_{\ptp[X]}, \aConf_1(\ptp[X]), \alfact,
		\tset_{\ptp[X]}}})_{\ptp[X] \in {\Set{\p,\q}}} \cup
  (M_{\ptp[X]})_{\ptp[X] \in {\ptpset\setminus\Set{\p,\q}}}$.
  The thesis hence follows by coinduction, since
  $\proj \acword \p = (\aout) \cat \proj {\acword'} \p$,
  $\proj \acword \q = (\ain) \cat \proj {\acword'} \q$ and, for each
  ${\ptp[X] \in {\ptpset \setminus \Set{\p,\q}}}$,
  $\proj \acword x = \proj {\acword'} x$.

  \proofcase{$\ssem{\hat \aCS} \subseteq \mlang(\cssem \aCS$)}
  For this case we have to prove that for all
  $\acword \in \alfint^\infty$, if $\acword\in \ssem{\hat \aCS}$ then
  $\acword\in \mlang(\cssem \aCS)$.
  Let $\acword \in \ssem{\hat \aCS}$ and proceed by coinduction.
  If $\acword = \emptyword$ the thesis follows immediately.
  Otherwise, let $\acword = (\gint[]) \cat \acword'$.
  Since $\proj \acword \p = (\aout) \cat \proj{\acword'}\p$,
  $\proj \acword \q = (\ain) \cat \proj{\acword'} \q$ and, for each
  ${\ptp[X] \in {\ptpset\setminus\Set{\p,\q}}}$,
  $\proj \acword x = \proj{\acword'} x$, and since, for each
  $\ptp[X] \in \ptpset$, we have
  $\proj \acword x \in \mlang(\aCM_{\ptp[X]})$, it follows, by
  definition of recognised word and by \cref{def:syncSem}, that
  $\acword' \in \ssem{\hat \aCS'}$, with
  $\hat\aCS' = (\mlang({\conf{\sset_{\ptp[X]}, \aConf_1(\ptp[X]),
		\alfact, \tset_{\ptp[X]}}}))_{\ptp[X] \in {\Set{\p,\q}}} \cup
  (\mlang(M_{\ptp[X]}))_{\ptp[X] \in {\ptpset\setminus\Set{\p,\q}}}$,
  where
  \begin{enumerate}
  \item $\aConf_0(\p) \trans[\p]{\aout} \aConf_1(\p)$ and
	 $\aConf_0(\q) \trans[\q]{\ain} \aConf_1(\q)$, and
  \item for all $\ptp[x] \neq \p,\q$,
	 $\aConf_1(\ptp[x]) = \aConf_0(\ptp[x])$.
  \end{enumerate}
  The thesis hence follows by coinduction.
\end{proof}
Notice that the proof of \cref{prop:eqsems} does not require CFSMs to
be deterministic; indeed, the result above holds for non-deterministic
CFSMs too.

The communication properties of a system of CFSMs $\aCS$ on $\ptpset$
considered in~\cite{blt20} are liveness, lock-freedom, and
deadlock-freedom.
Intuitively
\begin{itemize}
\item $\aCS$ is \emph{live} when each reachable configuration where a
  participant $\p \in \ptpset$ can execute a communication has a
  continuation where \p\ is involved;
\item $\aCS$ is \emph{lock-free} when in all computations starting
  from a reachable configuration where a participant $\p \in \ptpset$
  can execute, \p\ is involved;
\item $\aCS$ is \emph{deadlock-free} if in none of its reachable
  configurations without outgoing transitions there exists
  $\p \in \ptpset$ willing to communicate.
\end{itemize}

The next definition formalises properties of systems of CFSMs.
\begin{definition}[Communication properties of systems of CFSMs~\cite{blt20}]\label{def:commpropaut}
  Let $\aCS = (\aCM_{\px})_{\px \in \ptpset}$ be a system of CFSMs.
  \begin{itemize}
  \item {\bf Liveness}: $\aCS$ is \emph{live} if for each
	 configuration $\aConf \in \RS[\cssem \aCS]$ and each
	 $\p \in \ptpset$ such that $\aConf(\p)$ has some outgoing
	 transition in $\aCM_{\p}$, there exists a run of $\cssem \aCS$
	 from $\aConf$ including a transition involving \p.
  \item {\bf Lock freedom}: a configuration
	 $\aConf \in \RS[\cssem \aCS]$ is a \emph{lock} if there is
	 $\p \in \ptpset$ with an outgoing transition from $\aConf(\p)$ in
	 $\aCM_{\p}$ and there exists a run of $\cssem \aCS$ starting
	 from $\aConf$, maximal with respect to prefix order and containing
	 no transition involving \p.
	 System $\aCS$ is \emph{lock-free} if for each
	 $\aConf \in \RS[\cssem \aCS]$, $\aConf$ is not a lock.
  \item {\bf Deadlock freedom}: a configuration
	 $\aConf \in \RS[\cssem \aCS]$ is a \emph{deadlock} if
	 $\aConf$ has no outgoing transitions in $\cssem \aCS$, yet there exists
	 $\p \in \ptpset$ such that $\aConf(\p)$ has an outgoing transition
	 in $\aCM_{\p}$.
	 System $\aCS$ is \emph{deadlock-free} if for each
	 $\aConf \in \RS[\cssem \aCS]$, $\aConf$ is not a
	 deadlock.
  \end{itemize}
\end{definition}
  
It is the case that lock-freedom, strong lock-freedom, and
deadlock-freedom of $\hat \aCS$ (in the sense of
\cref{def:lf,def:slf,def:deadlockfree}) respectively imply liveness,
lock-freedom, and deadlock-freedom of $\aCS$ as stated in
\cref{prop:livelive} below, which relies on the following lemma.

\begin{lemma}\label{lem:auxlivelive}
  Let $\aCS = (\aCM_{\ptp[X]})_{\ptp[X] \in \ptpset}$ be a system of
  CFSMs.
  If $\acword \in \mlang(\cssem \aCS)$ is recognised by a run of
  $\cssem \aCS$ from $\aConf_0$ to $\aConf$ then, for each
  $\p\in\ptpset$, $\proj \acword \p \in \mlang(\aCM_{\p})$ and
  $\proj \acword \p$ is recognised by a run of $\aCM_{\p}$ from
  $\aConf_0(\p)$ to $\aConf(\p)$.
\end{lemma}
\begin{proof}
  By induction on the length of $\acword$ using \cref{def:syncSemaut}.
\end{proof}
  
\begin{proposition}\label{prop:livelive}
  For all systems of CFSMs
  $\aCS = (\aCM_{\ptp[X]})_{\ptp[X] \in \ptpset}$
  \begin{enumerate}[(i)]
  \item\label{prop:livelivei} $\hat \aCS$ lock-free iff $\aCS$ live;
  \item\label{prop:liveliveii} ${\hat \aCS}$ strong lock-free iff $\aCS$ lock-free;
  \item\label{prop:liveliveiii} ${\hat \aCS}$ deadlock-free iff $\aCS$
	 deadlock-free.
  \end{enumerate}
\end{proposition}
\begin{proof}
  \ref{prop:livelivei}\proofcase{$\Rightarrow$} Let ${\hat \aCS}$ be
  lock-free.
  Following \cref{def:commpropaut}, in order to show the liveness of
  $\aCS$, let us consider $\aConf \in \RS[\cssem \aCS]$, and let
  $\aConf(\p) \arro{\aact} {}$ be a transition in $\aCM_{\p}$ (for a
  certain $\p\in\ptpset$).
  Let now $\acword$ be the trace of a run of $\cssem \aCS$
  from $\aConf_0$ to $\aConf$.
  By \cref{prop:eqsems} $\acword \in \ssem{\hat \aCS}$ and, by
  \cref{lem:auxlivelive}, $\proj \acword \p \in \mlang(\aCM_{\p})$ and
  $\proj \acword \p$ is the trace of a run from $\aConf_0(\p)$ to
  $\aConf(\p)$ in $\aCM_{\p}$.
  Now, to prove that ${\aCS}$ is live, we have to show that there
  exists a run of $\cssem \aCS$ from $\aConf$ such that one of its
  transitions has a component transition from $\aConf(\p)$.
  We have that $\proj \acword \p \cat \aact \in \mlang(\aCM_{\ptp[A]})$.
  Hence, by lock-freedom of ${\hat \aCS}$, there is $\acword'$ such
  that $\acword \cat \acword' \in \ssem{\hat \aCS}$ and
  $\proj {\acword'} \p \neq \emptyword$.
  By \cref{prop:eqsems} $\acword'$ is also a run of $\cssem \aCS$ from
  $\aConf$ and, by \cref{lem:auxlivelive} (considering the CFSMs
  $\aCS' = (\aCM'_{\ptp[X]})_{\ptp[X] \in \ptpset}$ where each
  $\aCM'_{\ptp[X]}$ is like $\aCM'_{\ptp[X]}$ but with
  $\aConf(\ptp[X])$ as initial state), such a run has a transition
  from $\aConf(\p)$ as component transition.  \proofcase{$\Leftarrow$}
  By contraposition, let us assume ${\hat \aCS}$ not to be lock-free.
  Then there exists a participant $\p$ and a word
  $\acword\in\ssem{\hat \aCS}$ such that
  \begin{itemize}
  \item
  $\proj{\acword}{A}$ is not maximal in $\hat\aCS(\p)$; and
  \item
  for all $\acword'$, $\acword\cat\acword'\in\ssem{\hat \aCS}$ implies $\proj{\acword'}{A}=\emptyword$
  \end{itemize}
  Let $\aConf$ be the configuration in $\cssem \aCS$ reached by
  recognising $\acword$.  Since $\proj{\acword}{A}$ is not maximal, we
  get, by \cref{lem:auxlivelive} and determinism of the automata in
  $\aCS$, that $\aConf(\p)$ has at least an outgoing transition. If
  $\aCS$ were live, then there would be a run of $\cssem \aCS$ from
  $\aConf$ -- and hence a word of the form
  $\acword\cat\acword'\in \mlang{(\cssem \aCS)}=\ssem{\hat \aCS}$ --
  including a transition involving $\p$, so contradicting the liveness
  of $\aCS$. So $\aCS$ is not live and we are done.
  
The proof of \ref{prop:liveliveii} and \ref{prop:liveliveiii} are
similar to \ref{prop:livelivei}.
\end{proof}

\begin{remark}
  \cref{prop:livelive} does not hold in case we consider
  non-determistic CFSMs.
  For instance, let us consider the following system of CFSMs.
  $\aCS = (\aCM_{\px})_{\px \in \Set{\p,\q}}$ where
  \[
	 \begin{tikzpicture}[node distance=1.3cm]
      \tikzstyle{every state}=[cnode]
      \tikzstyle{every edge}=[carrow]
      \node[state, initial, initial text={$\aCM_{\p}$}] (0) {$q_0$};
      \node[state,  right of=0] (1) {$q_1$};
      \node[state,  right of=1] (2) {$q_2$};
      \path
      (0) edge node[above] {$\aout[A][B][][m]$} (1)
      (1) edge node[above] {$\aout[A][B][][n]$} (2)
      ;
    \end{tikzpicture}
	 \qquad
	 \begin{tikzpicture}[node distance=1.3cm]
      \tikzstyle{every state}=[cnode]
      \tikzstyle{every edge}=[carrow]
      \node[state, initial, initial text={$\aCM_{\q}$}] (0) {$q_0$};
      \node[state, above right of=0, yshift=-0.5cm] (1) {$q_1$};
      \node[state,  right of=1] (2) {$q_2$};
      \node[state,  right of=2] (3) {$q_3$};
      \node[state, below right of=0, yshift=0.5cm] (4) {$q_4$};
		\node[state, right of=4] (5) {$q_5$};
      \path
      (0) edge node[above] {$\emptyword$} (1)
      (0) edge node[below] {$\emptyword$} (4)
      (1) edge node[above] {$\ain[A][B][][m]$} (2)
      (4) edge node[above] {$\ain[A][B][][m]$} (5)
      (2) edge node[above] {$\ain[A][B][][n]$} (3)
      ;
    \end{tikzpicture}
  \]
The corresponding communicating systems is
$\hat\aCS = (\alang_{\px})_{\px \in \Set{\p,\q}}$ 
where 
\[\alang_{\p} =\Set{\emptyword,\, \aout[A][B][][m],\, \aout[A][B][][m]\cat\aout[A][B][][n]}
\qquad 
\alang_{\q} =\Set{\emptyword,\, \ain[A][B][][m],\, \ain[A][B][][m]\cat\ain[A][B][][n]}
\]
It is easy to check that $\hat\aCS$ is deadlock-free.
However $\aCS$ is not so, since the system can reach the stuck
configuration $(q_1,q_5)$ and there is an outgoing transition from
$q_1$ in $\aCM_{\p}$.
Likewise, the communicating system made of $\aCM_{\p}$
  and either of the following machines
  \[
	 \begin{tikzpicture}[node distance=1.3cm, initial text={}]
      \tikzstyle{every state}=[cnode]
      \tikzstyle{every edge}=[carrow]
      \node[state, initial] (0) {\phantom{$q_0$}};
      \node[state, above right of=0, yshift=-0.5cm] (1) {\phantom{$q_0$}};
      \node[state, right of=1] (2) {\phantom{$q_0$}};
      \node[state, below right of=0, yshift=.5cm ] (3) {\phantom{$q_0$}};
      \path
      (0) edge node[above] {$\ain[A][B][][m]$} (1)
      (0) edge node[below] {$\ain[A][B][][m]$} (3)
      (1) edge node[above] {$\ain[A][B][][n]$} (2)
      ;
    \end{tikzpicture}
	 \qquad
	 \begin{tikzpicture}[node distance=1.3cm, initial text={}]
      \tikzstyle{every state}=[cnode]
      \tikzstyle{every edge}=[carrow]
      \node[state, initial] (0) {\phantom{$q_0$}};
      \node[state, above right of=0, yshift=-0.5cm] (1) {\phantom{$q_0$}};
      \node[state, right of=1] (2) {\phantom{$q_0$}};
      \node[state, right of=2] (3) {\phantom{$q_0$}};
      \node[state, below right of=0, yshift=.5cm ] (4) {\phantom{$q_0$}};
      \path
      (0) edge node[above] {$\emptyword$} (1)
      (0) edge node[below] {$\ain[A][B][][m]$} (4)
      (1) edge node[above] {$\ain[A][B][][m]$} (2)
      (2) edge node[above] {$\ain[A][B][][n]$} (3)
      ;
    \end{tikzpicture}
  \]
  also reaches a deadlock configuration.

The FCL formalism, by abstracting from the notion of state, is in fact
intrinsically deterministic.
\end{remark}

\renewcommand{\aCS}{S}


%% file: choreography-automata.tex
We advocated \emph{choreography automata}
(c-automata)~\cite{blt20} as an expressive and flexible model of
global specifications.
Essentially c-automata are finite-state automata (FSAs) whose
transitions are labelled with interactions.
This yields an immediate connection between g-languages and
c-automata in terms of the languages the latter accept.
This section explores such connection, as well as the connection
between the projection operation on c-automata in~\cite{blt20} and the
projection on g-languages.

\subsection{The g-languages of c-automata}
\emph{Choreography automata} (c-automata for short, ranged over by
$\chora$, $\chora[B]$, etc.) are defined in~\cite{blt20} as
deterministic FSAs with labels in the set $\alfint$ of interactions.
%
(Observe that the set of participants occurring in a c-automaton is
necessarily finite.)
We can see c-automata as a tool for specifying g-languages
as shown by the next proposition.
\begin{proposition}\label{fac:pcc}
  Given a c-automaton $\chora$, $\mlang(\chora)$ is a continuous
  g-language.
\end{proposition}
\begin{proof}
  To show that $\mlang(\chora)$ is a g-language we need to show prefix
  closure. It follows from the fact that all the states are accepting:
  if a word $w$ is in $\mlang(\chora)$ then any prefix of $w$ can be
  generated by taking the corresponding prefix of the computation
  generating $w$.

  We now need to show that the language is continuous.
  We need to show that $\mlang(\chora)$ contains an
  infinite word if it contains all its finite prefixes. Note that,
  thanks to determinism, words which are prefixes one of the other are
  generated by computations which are prefixes one of the other as
  well. Let $w$ be an infinite word whose prefixes are in
  $\mlang(\chora)$. The infinite run obtained as the limit of the runs
  generating the prefixes of $w$ generates $w$. This concludes the
  proof.
\end{proof}

Interestingly, c-automata have an immediate relation with
CFSMs~\cite{bz83}, introduced in the previous section, which can be
adopted to model the local behaviour of distributed components.
Indeed, the local behaviour of a participant of a c-automaton $\chora$
can be algorithmically obtained directly from $\chora$.
Formally\footnote{%
  Overloading the projection operator of \cref{def:projection} does
  not introduce confusion and avoids the introduction of further
  notation.
}
\begin{definition}[Projection of c-automata~\cite{blt20}]\label{def:projca}
  Let $\chora = \conf{\sset, q_0, \alfint, \tset}$ be a c-automaton
  and \p\ be a participant.
  The \emph{projection of $\chora$ on \p} is the CFSM
  $\proj \chora \p$ obtained by determinising up-to-language
  equivalence the \emph{intermediate} automaton
  \[
	 \intaut{}{\p} = \conf{ \sset, q_0, \alfact \cup \Set{\emptyword}, \Set{q
		  \arro{\proj{\aint[l]} \p} q' \sst q \arro{\aint[l]} q'} }
  \]
  The \emph{projection of $\chora$}, written $\proj \chora {}$, is the
  system of CFSMs $(\proj \chora \p)_{\p \in \ptpset}$.
\end{definition}

The l-language of a projection of a c-automaton $\chora$ coincides
with the projection of the language of $\chora$:
\begin{proposition}\label{prop:semcom}
  If $\chora$ is a c-automaton on $\ptpset$ then for all
  $\p[x] \in \ptpset$,
  $\proj{\mlang(\chora)} x = \mlang(\proj \chora x)$.
\end{proposition}
\begin{proof}
  By definition of projection and since determinisation preserves the
  language, we have
  $\mlang(\proj \chora x)=\mlang(\intaut{}{\p[x]})$, where
  $\intaut{}{\p[x]}$ is the intermediate automaton
  (cf.~\cref{def:projca}).
  Since $\intaut{}{\p[x]}$ is obtained by taking the projection of
  every transition in $\chora$,
  $\mlang(\intaut{}{\p[x]})=\proj{\mlang(\chora)} x$.
  The thesis follows by transitivity.
\end{proof}

Observing that $\proj \chora {}$ is $\emptyword$-free, the LTS of
$\proj \chora {}$ is a c-automaton and its language coincides with the
g-language of the system
$(\mlang(\proj \chora x))_{\ptp[x] \in \ptpset}$ thanks to
\cref{prop:eqsems}.

The notion of well-formedness we provided in~\cite{blt20} was aimed to
guarantee correctness and completeness as well as the communication
properties of projected systems (cf.~\cref{def:commpropaut}).
We discovered later on that our notion of well-formedness
in~\cite{blt20} was flawed.
In fact it does not correctly handle the interplay between concurrent
transitions and choices.
This is shown by the following example.
\begin{example}\label{ex:bad}
  The c-automaton $\chora$ below
  \[
	 \begin{tikzpicture}[node distance=2.5cm]
      \tikzstyle{every state}=[cnode]
      \tikzstyle{every edge}=[carrow]
      \node[state, initial, initial text={$\chora$}] (0) {$q_0$};
      \node[state, above right of=0, yshift=-1cm] (1) {$q_1$};
      \node[state, below right of=0, yshift=1cm] (2) {$q_2$};
      \node[state, below right of=1, yshift=1cm] (3) {$q_3$};
      \node[state, right of=1] (4) {$q_4$};
		\node[state, right of=2] (5) {$q_5$};
		\node[state, right of = 3] (6) {$q_6$};
      \path
      (0) edge node[above] {$\gint[]$} (1)
      (1) edge node[below, near start] {$\gint[][c][n][d]$} (3)
      (1) edge node[above] {$\gint[][c][r]$} (4)
      (0) edge node[below] {$\gint[][c][n][d]$} (2)
      (2) edge node[above, near start] {$\gint[]$} (3)
      (2) edge node[below] {$\gint[][c][r]$} (5)
      (3) edge node[below, near start] {$\gint[][c][r]$} (6)
      (4) edge node[above] {$\gint[][c][n][d]$} (6)
      (5) edge node[below] {$\gint[]$} (6)
      ;
    \end{tikzpicture}
  \]
  is well-formed according to~\cite[Def. 4.12]{blt20}.
  However, the system $\proj{\chora}{}$ admits the run
  $\gint[][c][r]\ \gint[][c][n][d]$ which is not a word accepted by
  $\chora$.
  \finex
\end{example}
In our setting, the problem of the above $\chora$ is that
$\mlang(\chora)$ is not CUI.
In fact, by setting $\acword_1= \gint[][a][m][b]$ and
$\acword_2= \gint[][c][n][d]$, we have that
$\acword_1\cat\gint[][c][r][b], \acword_2\cat\gint[][c][r][b]\in
\mlang(\chora)$.
Now, by taking $\acword = \emptyword$, we have that
$\proj{\acword_1}{c} = \proj{\acword}{c}$ and
$\proj{\acword_2}{b} = \proj{\acword}{b}$, but
$\acword\cat\gint[][c][r][b]=\gint[][c][r][b]\not\in\mlang(\chora)$.


%% file: decidability.tex
In this section we show that CUI and BA are decidable when
restricting to c-languages associated to c-automata.
The approach of the present paper can hence be used to prove the
communication properties considered in~\cite{blt20} (namely those in
\cref{def:commpropaut}) by showing the corresponding ones in the FCL
setting.

We begin by unveiling an interesting interplay between c-automata
and FCLs.
Languages accepted by c-automata are continuous closures of regular
languages, so making conditions like CUI and branch-awareness
decidable.

\begin{theorem}
  CUI is decidable on languages accepted by c-automata.
\end{theorem}
  \newcommand{\apath}[4]{
	 {#2} \xRightarrow{#3}_{#1} {#4}
  }
\begin{proof}
  Let $\clang$ be a language accepted by a c-automaton $\chora$.
  We show that $\neg\cuui$ is decidable by reducing the problem to a
  search in the FSAs $\chora$, $\proj \chora a$, and $\proj \chora b$.

  By definition of CUI we have $\neg\cuui$ iff there are an
  interaction $\aint = \gint$ and finite words
  $\acword_1, \acword_2, \acword \in \clang$ such that
  \[
	 \acword\cat\aint\not\in\clang
	 \qquad
	 \proj \acword a = \proj{\acword_1} a
	 \qquad
	 \proj \acword b = \proj{\acword_1} b
	 \qqand[while]
	 \acword_1 \cat \aint, \acword_2 \cat \aint \in \clang
  \]
  This amounts to find a state $q$ of $\chora$ and two states\footnote{
	 Recall that by definition of projection the states of
	 $\proj \chora a$ and $\proj \chora b$ are subsets of the states of
	 $\chora$ due to the determinisation of the intermediate automata.
  } $Q_{\p}$ and $Q_{\q}$ respectively in
  $\proj \chora a$ and $\proj \chora b$ such that
  \[
	 q \not \arro \aint
	 \quad
	 q \in Q_{\p} \cap Q_{\q}
	 \qand[while there are]
	 q_{\p} \in Q_{\p}, q_{\q} \in Q_{\q} \text{ s.t. }
	 q_{\p} \arro \aint \text{ and } q_{\q} \arro \aint
  \]
  Since all the involved FSAs are finite an exhaustive search can determine
  if such states exist for each interaction.
\end{proof}

\begin{theorem}
  BA is decidable on languages accepted by c-automata.
\end{theorem}
\begin{proof}
  Let $\clang$ be a language accepted by a c-automaton $\chora$ and
  let $\intaut{}\px$ be the intermediate automaton obtained as in
  \cref{def:projca} on a participant $\p[X]$ of $\chora$.
  Given a state $p$ of an FSA $A$, we let $[A]_p$ denote the FSA
  obtained by replacing the initial state of $A$ with $p$ and let
  $[A]^p$ denote the FSA obtained by making $p$ the only final state
  in $A$ (i.e., setting the final states of $A$ to the singleton
  $\Set p$).
  
  In order to get a decision procedure for $\ba$, we observe that the
  following equivalences do hold.
  \begin{itemize}
  \item[]
	 $\neg\ba$
  \item[]
    iff there are $\ptp[X] \in \ptpset$ and two maximal words
	 $\acword_1 \neq \acword_2 \in \clang$ such that
	 $\proj{\acword_1}{X}\prec\proj{\acword_2}{X}$
  \item[]
    iff there are $\px \in \ptpset$, two states $p \neq q$ of
	 $\chora$ such that
	 \begin{enumerate}[a)]
	 \item\label{labaa}
		$\mlang(\intaut{p}{\px}) \cap \mlang(\intaut{q}{\px}) \neq
		\emptyset$;
	 \item  
		\label{laba}
		there are two maximal words $\acword_1' \in \mlang([\chora]_p)$
		and $\acword_2' \in \mlang([\chora]_q)$ such that
		$\proj{\acword_1'}{X}=\emptyword$ and
		$\proj{\acword_2'}{X} \neq \emptyword$.
	 \end{enumerate}
  \end{itemize}
  Intuitively, the first equivalence above states
  that for a participant $\px$ it is not possible to distinguish when
  a computation halts with word $\acword_1$ or continues as
  $\acword_2$.
  Transferring this condition on automata yields the last equivalence,
  which basically requires that
  $\mlang(\intaut{p}{\px}) \cap \mlang(\intaut{q}{\px}) \neq
  \emptyset$.
  Then from its initial state $\chora$ has two paths from the initial
  state respectively reaching $p$ and $q$ with two words whose projection
  on $\px$ coincide.
  Therefore there are two runs from $p$ and $q$ respectively such that
  $\px$ does not occur in the run from $p$ while it occurs in the run
  from $q$, as required by \cref{laba}.
   
  We now notice that \cref{labaa} is decidable because the
  intersection of $\omega$-regular languages is computatable.
  In order to decide \cref{laba}, instead, we can proceed as follows.
  We perform a breath-first search on $\intaut{p}{\px}$ stopping
  either when a state does not have outgoing transitions or when it
  has been already visited.
  \Cref{laba} holds iff the resulting tree has a path from
  $p$ to a leaf made of $\emptyword$-transitions only.
  Finiteness of automata guarantees that this procedure terminates.
  Likewise we can check the existence of a word
  $\acword'_2\in \mlang([\chora]_q)$ such that
  $\proj{\acword'_2}{X} \neq \emptyword$ (as required by \cref{laba}).
\end{proof}


%% file: DS-globalTypes.tex
\newcommand{\Gvti}[5]{\ensuremath{#1\to#2:\{#3_i. #5_i \}_{1 \leq i \leq n}}}
\newcommand{\pp}{\p}
\newcommand{\G}{\ensuremath{{\sf G}}}
\newcommand{\tend}{\mathtt{end}}
\newcommand{\mkt}[3]{\ensuremath{#1\to#2:#3  }}
\newcommand{\setl}[3]{\{#1_{i}. #2_{i}\}_{1 \leq i \leq #3}}
\newcommand{\Labgt}{\Gamma}
\newcommand{\pq}{\q}
\newcommand{\pr}{\p[r]}
\newcommand{\ps}{\p[s]}
\newcommand{\inact}{\ensuremath{\mathbf{0}}}
\newcommand{\project}{\upharpoonright}
\newcommand{\projp}[2]{ (#1) \!\project_{#2}}
\newcommand{\projds}[2]{ #1 \!\project_{#2}}
\newcommand{\pout}[2]{#1 ! #2}
\newcommand{\setlp}[4]{\{#1_{i}. \projds{#2_{i}}{#4} \}_{1 \leq i \leq #3}}
\newcommand{\pin}[2]{#1 ? #2}
\newcommand{\mklab}[2]{#1 \uplus #2}
\newcommand{\Lab}{\Lambda}
\newcommand{\langc}[1]{\clang(#1)}
\newcommand{\langa}[1]{\alang(#1)}
\newcommand{\langap}[1]{\widehat\alang(#1)}
\newcommand{\langp}[1]{\widehat\clang(#1)}
\newcommand{\PP}{\mathit{P}}
\newcommand{\MM}{\mathcal{M}}
\renewcommand{\ell}{\msg}

The global types of~\cite{SeveriD19} are our last case study.
Informally, a global type $\Gvti \pp \pq \ell \S {\G} $ specifies a
protocol where participant $\pp$ must send to $\pq$ a message $\ell_i$
for some $1 \leq i \leq n$ and then, depending on which $\ell_i$ was
chosen by $\pp$, the protocol continues as $\G_i$.
Global types and multiparty sessions are defined in~\cite{SeveriD19}
in terms of the following grammars for, respectively, \emph{pre-global
  types}, \emph{pre-processes}, and \emph{pre-multiparty sessions} (we
adapted some of the notation to our setting):
\[
  \begin{array}{ccc}
    \begin{array}{lcl}
      \G & ::=^{\textrm{co}} & \tend
      \\
			& \bnfmid & \mkt \pp \pq  \setl {\ell}{\G} {n}
		  \\\ 
    \end{array}
		&
		  \begin{array}{lcl}
		    \PP & ::=^{\textrm{co}} & \mathbf{0}
		    \\
			& \bnfmid & \pp?\{\ell_i.\PP_i\}_{1\leq i\leq n}
		    \\[1mm]
			& \bnfmid &  \pp!\{\ell_i.\PP_i\}_{1\leq i\leq n}
		  \end{array}
		&
	   \begin{array}{lcl}
	     \MM  & ::= & \pp\,\triangleright\PP
	     \\
		  & \bnfmid & \MM \!\!\bnfmid\!\! \MM
			   \\\ 
	   \end{array}
  \end{array}
\]
where we assume messages $\ell_i$ to be pairwise distinct in sets
$\setl {\ell}{\G} {n}$ and $\{\ell_i.\PP_i\}_{1\leq i\leq n}$, to
whose elements we call \emph{branches}.
The first two grammars are interpreted coinductively, that is their
solutions are both minimal and maximal fixpoints (the latter
corresponding to infinite trees).
A pre-global type $\G$ (resp.~pre-process $\PP$)
is a \emph{global type} (resp.~\emph{process}) if its tree
representation is \emph{regular}, namely it has finitely many distinct
sub-trees.
A \emph{multiparty session} (MPS for short) is a pre-multiparty
session such that (a) in ${\pp}\,\triangleright{\PP}$, participant \p\
does not occur in process $\PP$ and (b) in
${\pp_1\,\triangleright}{\PP_1} \mid \ldots \mid
{\pp_n}\,\triangleright{\PP_n}$, participants $\pp_i$ are pairwise
different.

The semantics of global types is the LTS induced by
the following two rules:
\[\begin{array}{c}
  \\
  \prooftree
  \justifies
  \mkt \pp \pq \setl {\ell}{\G} {n}
  \xrightarrow{\pp\rightarrow\pq:\ell_i} \G_i
  \using \text{(GT1)}
  \endprooftree
  \\[2em]
  \prooftree
  \text{for each } 1\leq i\leq n \quad
  {\G_i} \xrightarrow{\pp\rightarrow\pq:\ell} \G'_i
  \qquad
  \Set{\pp,\pq}\cap\Set{\pr,\ps}=\emptyset
  \justifies
  \mkt \pr \ps \setl {\ell}{\G} {n}
  \xrightarrow{\pp\rightarrow\pq:\ell} \mkt \pr \ps \setl {\ell}{\G'}
  n
  \using \text{(GT2)}
	 \endprooftree
	 \\[2em]
  \end{array}
\]
Rule \text{(GT2)} allows out-of-order execution, namely interaction
$\gint[]$ between participants \p\ and \q\ can happen even if an
interaction between other participants is syntactically occurring
before, provided that $\gint[]$ occurs in all branches.
As usual, we let
$\G \xrightarrow{\aint_0\cat \cdots \cat \aint_n} \G'$ to denote
$\G \xrightarrow{\aint_0} \G_1 \cdots \G_{n}
\xrightarrow{\aint_n} \G'$ and always assume that
$\G \xrightarrow \emptyword \G$ holds.

The semantics for MPSs is the LTS defined by the following rule
\begin{align}\label{eq:mpssem}
	 \pp\triangleright\pq!(\{\ell.\PP\}\uplus\Lambda)
	 \bnfmid \pq\triangleright\pp?(\{\ell.\PP'\}\uplus\Lambda') \bnfmid \MM
	 \quad \xrightarrow{\pp\rightarrow\pq:\ell} \quad
	 \pp\triangleright\PP \bnfmid \pq\triangleright\PP' \bnfmid \MM
\end{align}
where
$\{\ell_i: \PP_i\}_{i \in I} \uplus \{\ell_j': \PP'_j\}_{j \in J} =
\{\ell_i: \PP_i\}_{i \in I} \cup \{\ell_j': \PP'_j\}_{j \in J}$ if
$\{\ell_i\}_{i\in I} \cap \{\ell'_j\}_{j \in J} = \emptyset$,
otherwise
$\{\ell_i: \PP_i\}_{i \in I} \uplus \{\ell_j': \PP'_j\}_{j \in J}$ is undefined.

Rule~\eqref{eq:mpssem} applies only if the messages in $\Lambda'$
include those in $\Lambda$, which is the case for MPSs obtained by
projection, defined below.
\begin{definition}[Projection~{\cite[Definition 3.4]{SeveriD19}}]\label{definition:projection}
  The \emph{projection} of $\G$ on a participant \p[x] such that the
  depths of its occurrences in $\G$ are bounded is the partial
  function $\projds{\G}{\p[x]}$ coinductively defined by
  $ \projds \tend {\p[x]} = \inact$ and, for a global type
  $\G = \Gvti \pp\pq \ell \S {\G}$, by:
	 \[
		  \projds \G {\p[x]} = \begin{cases} \inact & \text{if $\ptp[x]$ is
		  not a participant of } \G
		  \\
		  \pout \pq \setlp{\ell} {\G}{n}{\p[x]} & \text{if } \p[x] = \p
		  \\
		  \pin\p \setlp{\ell}{\G}{n}{\p[x]} & \text{if } \p[x] = \q
		  \\
		  \projds{\G_1}{\p[x]} & \text{if } \p[x] \not\in \{\p, \q \}
		  \text{ and } n=1
		  \\
		  \pin \ps (\mklab {\Lab_1} {\mklab{\ldots}{\Lab_n}} ) &
		  \text{if } \p[x] \not\in \{\p, \q \}, \ n>1, \text{ and }
		  \text{for all } 1 \leq i \leq n, \projds {\G_i}{\p[x]} = \pin \ps
		  {\Lab_i}
		\end{cases}
  \]
  The global type $\G$ is \emph{projectable}\footnote{%
	 In~\cite{SeveriD19}, projectability embeds well-branchedness.
  } if $\projds \G{\p[x]}$ is defined for all participants $\p[x]$ of
  $\G$, in which case $\projds \G {}$ denotes the corresponding MPS.
\end{definition}
The projection on $\p[x]$ is partial because if we have a choice where
$\p[x]$ is not involved in the first interaction then only the last
clause in \cref{definition:projection} could apply, but this clause
requires $\p[x]$ to start each branch with an input action from the
same sender $\p[s]$.

Let $\langp{\G}$ be the language coinductively defined as follows:
\[
	 \langp{\tend} = \Set{\emptyword}
	 \qqand
	 \langp{\mkt \pp \pq  \setl {\ell}{\G} {n}} =
	 \bigcup_{1\leq i\leq n}\Set{\gint[][\pp][\ell_\mathit{\color{black}i}][\pq] \cat \acword \sst \acword \in \langp{\G_i}}
\]
The g-language $\clang(\G)$ associated to a global type $\G$ is the concurrency and
prefix closure of $\langp \G$, that is
\[
  \langc \G = \pref[\Set{\acword \in \alfint^\infty \sst \text{ there is }
	 \acword' \in \langp \G \text{ such that } \acword \comm \acword'}]
\]

We define the l-language $\langa{\pq \triangleright \PP}$ associated
to a named process $\pq \triangleright \PP$ as the prefix closure of
$\langap{\pq \triangleright \PP}$ which, letting
$\star \in \Set{?,!}$, is defined by
\[\begin{array}{l}
  \langap{\pq \triangleright\mathbf{0}} = \Set{\emptyword} \qand
  \langap{\pq \triangleright\pp\!\star\{\ell_i.\PP_i\}_{1\leq i\leq n}} =
  \bigcup_{1\leq i\leq n}\Set{
  \text{\small $\pp\pq\!\star\ell_\mathit{\color{black}i}$} \cat \acword \sst \acword \in \langap{\PP_i}
  }
\end{array}
\]

The system associated to an MPS is defined as the following map:
\[
  S({\pp_1\,\triangleright}{\PP_1} \mid \ldots\mid {\pp_n}\,\triangleright{\PP_n}) =
  \Set{\pp_i \mapsto \langa{\pp_i\,\triangleright\PP_i} \sst 1 \leq i \leq n}
\]

Our constructions capture relevant properties of the
global types in~\cite{SeveriD19}.

We establish correspondences between the two frameworks in
\cref{prop:sda,prop:sdb}.

\begin{restatable}{proposition}{severia}\label{prop:sda}
  Given a projectable global type $\G$,  
  $\langc \G = \Set{\acword \sst \G \xrightarrow{\acword}}$.
\end{restatable}
\begin{proof}
  We show the two inclusions.

 \proofcase{$\subseteq$}
	 We begin by showing
	 $\langp \G \subseteq \Set{\acword \sst \G \xrightarrow{\acword}}$.
	 If $\G = \tend$ then the thesis trivially follows since
	 $\G \xrightarrow{\emptyword}$.
	 If $\G = \mkt \pp \pq \setl {\ell}{\G} {n}$ then each word in
	 $\langp \G$ begins with an interaction $\gint[][\pp][\ell_\mathit{\color{black}i}][\pq]$
	 with $1 \leq i \leq n$.
	 By the reduction semantics of global types,
	 $\G \xrightarrow{\gint[][\pp][\ell_\mathit{\color{black}i}][\pq]} \G_i$.
	 By coinduction~\cite{KozenS17} we obtain the thesis.
	 Each word in $\langc \G$ is a prefix of a concurrency equivalent
	 word $\acword \in \langp \G$.
	 Hence we obtain the thesis by observing that
	 $\G \xrightarrow{\acword}$ from the previous argument and by
	 (possibly repeated) application of rule~(GT2) of the semantics of
	 global types.
  \proofcase{$\supseteq$}
	 If $\G \xrightarrow{w}$ then, for some $\G'$ we have
	 $\G \xrightarrow{\aint} \G' \xrightarrow{\acword'}$ where
	 $\acword = \aint \cat \acword'$.
	 If $\G \xrightarrow{\aint} \G'$ is an instance of the first rule of
	 the operational semantics of global types then the thesis follows by
	 coinduction.
	 Otherwise, assuming $\G = \mkt \pr \ps \setl {\ell}{\G} {n}$, for
	 each $1\leq i\leq n$ we have (i) $\G_i \xrightarrow{\aint} \G'_i$,
	 (ii) $\ptpof[\aint] \cap \Set{\pr,\ps}=\emptyset$, and (iii)
	 $\G' = \mkt \pr \ps \setl \ell {\G'} n \xrightarrow{\acword'}$.
	 By coinduction, $\acword' \in \langc{\G'}$ hence
	 $\acword = \aint \cat \acword' \in \langc \G$.
\end{proof}

We now relate projectability (cf.~\cref{definition:projection}) and
our properties.
\begin{restatable}{proposition}{wfba}\label{prop:cuiba}
  If $\G$ is a projectable global type then $\langc{\G}$ is a CUI and
  branch-aware \sclang.
\end{restatable}
\begin{proof}
  Let $\G = \Gvti \pp \pq \ell \S \G$ be a global type projectable on
  $\p[x]$.
  By construction, $\langc \G \subseteq \alfint^\star$ or $\langc \G$
  is continuous.

  To show that $\langc \G$ is branch-aware we proceed by
  contradiction: assume that $\langc{\G}$ is not branch-aware.
  Then there exist two distinct maximal words in $\langc \G$, say
  $\acword = \aint_0 \cat \aint_1 \cdots$ and
  $\acword' = \aint'_0 \cat \aint'_1 \cdots$, such that
  $\acword \not\sim \acword'$ and, for a participant $\p[x]$, we have
  $\proj \acword x \prec \proj{\acword'} x$, that is
  $\proj \acword x$ is a proper prefix of $\proj{\acword'} x$.
  Note that $\p[x]$ must occur in $\G$ otherwise,
  $\proj \G x = \inact$.
  By definition of $\langc \G$, there are
  $\widehat \acword = \aint_0 \cat \aint_1 \cdots$ and
  $\widehat{\acword'} = \aint'_0 \cat \aint'_1 \cdots$ in $\langp \G$,
  such that $\acword \sim \widehat \acword$ and
  $\acword' \sim \widehat{\acword'}$; we have
  $\proj{\widehat \acword} x \prec \proj{\widehat \acword'} x$.
  
  By \cref{prop:sda}, there are two runs
  \[
	 \G \xrightarrow{\aint_0} \G_1 \xrightarrow{\aint_2} \cdots
	 \qqand
	 \G \xrightarrow{\aint'_0} \G'_1 \xrightarrow{\aint'_2} \cdots
  \]
  Let $h$ be the index such that $\aint_h \neq \aint'_h$ and
  $\aint_i = \aint'_i$ for all $0 \leq i < h$; note that $h$ must
  exist otherwise $\widehat \acword = \widehat \acword'$,
  contradicting $\acword \not\sim \acword'$.
  Necessarily, $\p[x]$ does not occur in $\aint_h$ otherwise
  $\proj{(\aint_0\cat \cdots \cat \aint_h)} x \not\prec
  \proj{(\aint'_0\cat \cdots \cat \aint'_h)} x$ contradicting
  $\proj{\widehat \acword} x \prec \proj{\widehat \acword'} x$.
  
  Take the last index, say $n$, of the interaction in
  $\widehat \acword$ involving $\p[x]$, which necessarily exists
  otherwise $\proj{\widehat \acword} x$ cannot be a proper prefix of
  $\proj{\widehat \acword'} x$.
  This contradicts the projectability of $\G$ since its subtree $\G_h$
  has two branches initiating with $\aint_h$ and $\aint'_h$ on which
  the projections for $\p[x]$ differ.
  Hence, none of the cases of \cref{definition:projection} can be
  applied.

  To show CUI we have to prove that if there are words
  $\acword_1 \cat \gint, \acword_2 \cat \gint,\acword \in \langc{\G}$
  such that $\gproj[\acword_1][\p] = \gproj[\acword][\p]$ and
  $\gproj[\acword_2][\q] = \gproj[\acword][\q]$ then
  $\acword\cat\gint\in \langc{\G}$.
  
  By definition of $\langc{\G}$, from
  $\acword_1 \cat \gint, \acword_2 \cat \gint,\acword\in \langc{\G}$
  we can infer that there exist
  $\hat \acword_1, \hat\acword_2,\hat\acword\in \langp{\G}$ and
  $\acword'_1, \acword'_2,\acword'\in \langc{\G}$ such that
  \begin{itemize}
  \item
  $\hat \acword_1 \sim \acword'_1$\qquad  $\hat\acword_2\sim \acword'_2$ \qquad
  $\hat\acword\sim\acword'$;
  \item
  $\acword'_i = \acword_i\cat\gint\cat\acword''_i$ \quad with $\ptpof[\acword''_i]\cap\Set{\p,\q}=\emptyset$\qquad ($i\in\Set{1,2}$);
   \item
  $\acword' = \acword\cat\acword''$ \quad with $\ptpof[\acword'']\cap\Set{\p,\q}=\emptyset$.
  \end{itemize}
  Let us now consider the following cases.
    \proofcase{\mbox{$\ptpof[\acword']\cap\Set{\p,\q}=\emptyset$}}\mbox{}
    From the hypothesis $\acword'_1$ contains exactly one interaction
	 involving $\p$ and such an interaction is $\gint$.
	 We proceed now by coinduction on the projection of
	 $\G=\Gvti \pr\ps \ell \S {\G}$ on $\p$
	 (cf.~\cref{definition:projection}).
	 Since the first action of $\p$ is an output we necessarily have
	 a sequence of co-inductive cases where we apply the 4-th case
	 of \cref{definition:projection} followed by a case where
	 the 2-nd case 	 of \cref{definition:projection} is applied.
	 Hence, $\G$ has no branch before the interaction $\gint$ which
	 must occur in all the maximal words, as desired.
  \proofcase{\mbox{$\ptpof[\acword']\cap\Set{\p,\q}\neq\emptyset$}}\mbox{}
    Let $\gint[][c][x][d]$ be the first interaction in $\acword'$
	 (hence also in $\hat\acword$)
	 such that $\Set{\p[c],\p[d]} \cap \Set{\p,\q} \neq \emptyset$.
	 With no loss of generality, assume $\p = \p[c]$; the other cases
	 are analogous and therefore omitted.
	  
	 By $\gproj[\acword_1][\p] = \gproj[\acword][\p]$, the interaction
	 $\gint[][a][x][d]$ is present also in $\acword'_1$ and it is the
	 first one having $\p$ as participant.  Moreover, the same holds
	 for $\hat\acword_1$.  Let us proceed by coinduction on the
	 projection of
	 $\G=\ptp[S]\rightarrow\ptp[R]:\Set{\msg[m]_i.\G_i}_{1\leq i\leq
		n}$ on $\p$.  The case $\pr=\p$ follows by coinduction since
	 both $\hat \acword$ and $\hat \acword_1$ are generated from the
	 same branch since $\proj{\acword}{A}=\proj{\acword_1}{A}$. The
	 case $\p \not\in \{\pr, \ps\}$ and $n=1$ follows by coinduction.

  Let us consider the case
  $\Set{\ptp[S],\ptp[R]}\cap\Set{\ptp[A],\ptp[B]}=\emptyset$, and
  assume that $\hat\acword_1$ and $\hat\acword$ begin, respectively with
  $\ptp[S]\rightarrow\ptp[R]:\msg[m]_u$ and 
  $\ptp[S]\rightarrow\ptp[R]:\msg[m]_v$ with $u\neq v$ (the case $u=v$ is simpler).
  Let now $\projds {\G_u}{\p[A]} = \ptp[D] ? {\Lab_u}$,
  and $\projds {\G_v}{\p[A]} =  \ptp[D] ? {\Lab_v}$ (which are defined by projectability
  of $\G$).
  By what said previously, we have that
  ${\Lab_u} = \Set{\msg[x].P' \uplus \ldots}$ and ${\Lab_u} = \Set{\msg[x].P'' \uplus \ldots}$.
  Now, since 
  $\projds \G {\p[A]} = \ptp[D] ?\Set{{\Lab_u} \uplus {\Lab_v}\uplus\ldots }$,
  we can infer that $P'=P''$. 
 This implies that $\acword\cat\gint\in \langc{\G}$.
This concludes the proof.
\end{proof}

The proof of \cref{prop:sdb} below uses the following auxiliary lemma.
\begin{restatable}{lemma}{sd}\label{lem:sd}
  If $\projds \G \pp$ is defined then $\acword \in \langp \G$ implies
  $\proj \acword \pp \in \langap{\projds \G \pp}$.
\end{restatable}
\begin{proof}
  By coinduction on $\G$.
  If $\G = \tend$ the thesis trivially hold.
  If $\G = \mkt \pp \pq \setl {\ell}{\G} {n}$ or
  $\G = \mkt \pq \pp \setl {\ell}{\G} {n}$ then the thesis immediately
  follows by coinduction observing that
  $\acword = \gint[][a][\ell_\mathit{\color{black}i}][b] \cat
  \acword'$ or
  $\acword = \gint[][b][\ell_\mathit{\color{black}i}][a] \cat
  \acword'$ for some $1 \leq i \leq n$ and
  $\acword' \in \langp {\G_i}$.
  If $\G = \mkt \pr \ps \setl {\ell}{\G} {n}$ with $\pp \neq \pr$ and
  $\pp \neq \ps$ then
  $\acword = \gint[][r][\ell_\mathit{\color{black}i}][s] \cat
  \acword'$ for some $1 \leq i \leq n$ and
  $\acword' \in \langp {\G_i}$; hence
  $\projds \acword \pp = \projds{\acword'} \pp$.
  By \cref{definition:projection}, either (a) $n=1$ and
  $\projds \G \pp = \projds {\G_1} \pp$ or (b) $n>1$ and
  $\projds \G \pp = \pin \ps (\mklab {\Lab_1} {\mklab{\ldots}{\Lab_n}}
  )$ where, for all $1 \leq i \leq n$,
  $\projds {\G_i} \pp = \pin \ps {\Lab_i}$ and $\Lambda_{i}$ of the
  form $\{\ell^i_1: \G^i_1,\ldots,\ell^i_{n_i}: \G^i_{n_i}\}$.
  In case (a) the thesis immediately follows by coinduction observing
  that $i = n = 1$.
  In case (b),
  $\projds {\acword'} \pp = \pin \ps {\ell_{ij}} \cat \aaword$ with
  $1 \leq j \leq n_i$ and, coinductively,
  $\aaword \in \langap{\pp \triangleright \projds {\G_i} \pp}$ the
  projection on $\pp$ of some word in $\langp{\G_i}$ otherwise we
  would violate the projectability of $\G$ on $\pp$.
\end{proof}

\begin{restatable}{proposition}{severib}\label{prop:sdb}
  Given a projectable global type $\G$,
  $\ssem{S(\projds \G {})} = \Set{\acword \sst \projds \G {}
	 \xrightarrow{\acword}}$.
\end{restatable}
\begin{proof}
  We get
  $\ssem{S(\projds \G {})} = \Set{\acword \sst \projds \G {}
	 \xrightarrow{\acword}}$ by showing that
  \begin{align}\label{eq:prj}
	 S(\projds{\G}{}) = & \proj{\langc{\G}}{}
  \end{align}
  Indeed, $\ssem{S(\projds{\G}{})} = \ssem{\proj{\langc{\G}}{}}$ by
  \cref{prop:sda}.
  \cref{prop:cuiba} ensures $S(\projds{\G}{})$ is correct and
  complete, hence $\ssem{S(\projds{\G}{})} = \langc{\G}$ by
  \cref{th:correctness,th:completeness}.
  So, assume
  $\projds \G {} = {\pp_1\,\triangleright}{\PP_1} \mid \ldots \mid
  {\pp_n}\,\triangleright{\PP_n}$, then
  \begin{align*}
	 \proj{\langc \G} {} =
 	 & \proj{\Set{\acword \sst \G \xrightarrow{\acword}}}{}
 	 & & \text{by \cref{prop:sda}}
 	 \\
	 = & \Set{ \pp_i \mapsto \Set{\proj \acword {\pp_i} \sst \G \xrightarrow{\acword}} \ \sst\ 1 \leq i \leq n}
	 & & \text{by \cref{def:projection}}
 	 \\
	 = & \Set{ \pp_i \mapsto \Set{\proj \acword {\pp_i} \sst \acword \in \pref[\Set{\acword' \in  \langp \G \sst \acword \comm \acword'}] } \! \sst \! 1 \leq i \leq n} 
	 & & \text{by def. of $\langc \G$}
 	 \\
	 = &  \Set{\pp_i \mapsto \langa{\pp_i\,\triangleright\PP_i} \sst 1\leq i\leq n }
	 & & \text{by \cref{lem:sd}}
  \end{align*}
  Finally $\Set{\pp_i \mapsto \langa{\pp_i\,\triangleright\PP_i} \sst 1\leq i\leq n }
  = S(\projds \G {})$
  by definition of $S(\_)$. 
\end{proof}

Projectable global types are proved strongly lock-free
in~\cite{SeveriD19}.
The following result corresponds to \cite[Theorem 4.7]{SeveriD19}.
\begin{corollary}\label{cor:WFtoLF}
  $S(\projds{\G}{})$ is strongly lock-free for any projectable $\G$.
\end{corollary}

The symmetry between senders and receivers in CUI and branch-awareness
allows for an immediate generalisation of the projection in
\cref{definition:projection} by extending the last case with the
clause:
\[
  \pout \ps (\mklab {\Lab_1}   {\mklab{\ldots}{\Lab_n}} )
  \qquad \text{if } \p[x] \not\in \{\p, \q \}, \ n>1, \text{ and for all } 1 \leq i \leq n, \projds {\G_i}{\p[x]} = \pout \ps  {\Lab_i}
\]
\cref{cor:WFtoLF} still holds for this generalised definition of
projection.


%% file: related.tex
The use of formal language theories for the modelling of concurrent
systems dates back to the theory of traces~\cite{maz86}.
A trace is an equivalence class of words that differ only for swaps
of independent symbols.
Closure under concurrency corresponds on finite words to form traces,
as we noted after \cref{def:cclos}.
An extensive literature has explored a notion of realisability whereby
a language of traces is realisable if it is accepted by some class of
finite-state automata.
Relevant results in this respect are the characterisations
in~\cite{zie87,dub86} (and the optimisation in~\cite{gm06}) for finite
words and the ones in~\cite{Gastin90,gas91,gpz91} for infinite ones.
A key difference of our framework with respect to this line of work is that we
aim to stricter notions of realisability: in our context it is not
enough that the runs of the language may be faithfully executed by a
certain class of finite-state automata.
Rather we are interested in identifying conditions on the g-languages
that guarantee well-behaved executions in \quo{natural} realisations.

Other abstract models of choreographies, such as Conversation
protocols (CP)~\cite{fbs04} and c-automata~\cite{blt20}, have some
relation with our FCL.  We discussed in depth the relation with
c-automata in \cref{sec:chor-automata}.

CP, probably the first
automata-based model of choreographies, are non-deterministic B\"uchi
automata whose alphabet resembles a constrained variant of our
$\alfint$.
A comparison with the g-languages accepted by CPs is not immediate as
CPs are based on asynchronous communications (although some
connections are evident as noted below \cref{def:syncSem}).

Other proposals ascribable to choreographic settings
(cf.~\cite{hlvlcmmprttz16}) define global views that can be seen as
g-languages.
We focus on synchronous approaches because our current theory needs to
be extended to cope with asynchrony.

In~\cite{BravettiZ07,lgmz08} the correctness of implementations of
choreographies (called \emph{choreography conformance}) is studied in
a process algebraic setting.
The other communication properties we consider here are not discussed
there.

The notion of choreography implementation in~\cite{BravettiZ07}
corresponds to our correctness plus a form of existential termination.
It is shown that one can decide whether a system is an implementation
of a given choreography, since both languages are generated by
finite-state automata, hence language inclusion and existential
termination are decidable.

In~\cite{lgmz08} three syntactic conditions (connectedness, unique
points of choice and causality safety) ensure bisimilarity (hence
trace equivalence) between a choreography and its
projection.
Connectedness rules out systems which are not concurrency closed,
while we conjecture that unique points of choice and connectedness
together imply our CUI and BA.
Causality safety, immaterial in our case, is needed in~\cite{lgmz08}
due to the existence of an explicit parallel composition operator in
their process algebra.

Many multiparty session type systems~\cite{hlvlcmmprttz16} have two
levels of types (global and local) and one implementation level (local
processes).
This is the case also for synchronous session type systems such
as~\cite{KouzapasY14,Dezani-Ciancaglini16}. The approach
in~\cite{Stolze2022} instead merges the two levels of types into a
single one to allow for composition, while having a level of local
processes for implementation.
Our approach, like the session type systems
in~\cite{SeveriD19,BDLT20}, considers only (two) abstract
descriptions, g-languages and l-languages.
The literature offers several behavioural types featuring
correctness-by-construction principles through conditions (known as
projectability or well-branchedness) more demanding than ours.
For instance, relations similar to those in \cref{sec:mpst} can be
devised for close formalisms, such as~\cite{BDLT20} whose notion of
projection is more general than the one in~\cite{SeveriD19}, yet its
notion of projectability still implies CUI and BA.

There is a connection between CUI and the closure property CC2~\cite{aey03} on message-sequence charts (MSCs)~\cite{msc}.
On finite words CC2 and CUI coincide.
Actually, CUI can be regarded as a step-by-step way to ensure CC2 on
finite words.
The relations between our properties and CC3, also used in MSCs, are
still under scrutiny.


%% file: conc.tex
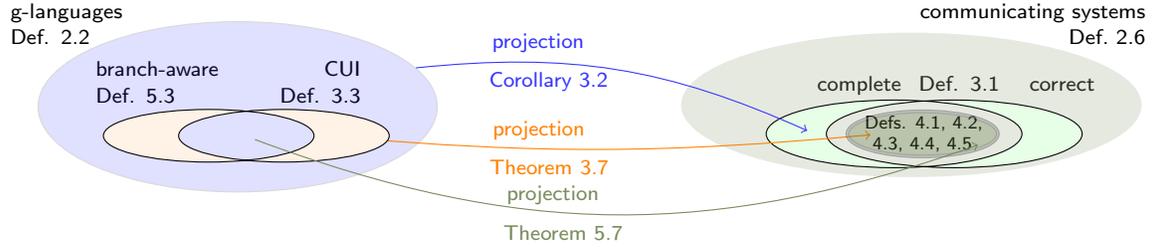
\begin{figure}[t]
  \centering
  \begin{tikzpicture}[node distance=.5cm and .5cm, scale = 1, every label/.style = {text width=1.5cm}, every node/.style={transform shape}]
	 \node[label={[above left, xshift=-1cm]{g-languages Def. \ref{def:chorlang}}}, text width=2.5cm] (g-languages) {};
	 \node[right = of g-languages] (g') {};
	 \node[right = 5.8cm of g-languages, label={[above right, text width=3cm, align=right, xshift=2cm]communicating systems \\ Def.~\ref{def:commSyst}}] (cs) {}; 
	 \node[below right = of cs, text height = .3cm, text width=1.5cm, label={[above,yshift=.1cm]complete}] (complete) {};
	 \node[right = of complete, text height = .2cm, text width=1.5cm, label={[above, yshift=.2cm, text width=4cm, align=center]Def. \ref{def:cc} \quad correct}, xshift = -.5cm] (correct) {};
	 \node[below = of g-languages, xshift=1cm, label={[above,xshift=-1cm, text width=3cm, align=left]branch-aware \\ Def. \ref{def:ba}}, text width=1.5cm, text height=.35cm] (ba) {};
	 \node[right = of ba, xshift=-2.5cm, label={[above,align=right,xshift=.5cm]CUI \\ Def. \ref{def:closedness}}, text width=1.5cm, text height=.35cm] (cui) {};
	 \node[draw,fill=blue!80,opacity=.15,inner sep=3pt,ellipse,fit=(g-languages) (g') (cui) (ba)] (eg){};
	 \node[fill=DarkOliveGreen!60,opacity=.25,inner sep=.25pt,ellipse,fit=(cs) (correct) (complete)] {};
	 \node[right = .02cm of complete] (x) {};
	 \node[left = 1.5cm of correct] (x) {};
	 \node[right =-.5cm of ba] (ecui) {};
	 \node[right = -2.7cm of correct] (z) {};
	 \node[left = -.5cm of ba] (u) {};
	 \node[left = .1cm of correct] (v) {};
	 \node[right = .6cm of v,yshift=-.1cm] (v') {};
	 \node at (v) {
		\begin{minipage}{.2\linewidth}\centering\tiny
		  Defs. \ref{def:llive}, \ref{def:wlive},
		  \\
		  \ref{def:slf}, \ref{def:lf}, \ref{def:deadlockfree}
		\end{minipage}
	 };
	 \fill[draw,fill=green!10, even odd rule] (correct) circle[x radius = 1.7cm, y radius =.45cm,xshift=-1.5cm] (correct) circle[x radius = 1.7cm, y radius =.45cm,xshift=-.7cm];
	 \path[draw, ->, blue!80] (eg) edge[bend left=15] node[below,xshift=-0.9cm,yshift=.1cm,label={[above,] projection}]{\cref{th:completeness}} (x);
	 \path[draw, ->, orange] (ecui) edge[bend right = 5] node[below,xshift=-.7cm,label={[above]projection}]{\cref{th:correctness}} (z);
	 \fill[draw,fill=orange!10, even odd rule] (ba) circle[x radius = 1.4cm, y radius =.35cm,xshift=-1cm] (ba) circle[x radius = 1.4cm, y radius =.35cm];
	 \fill[draw,double,fill=DarkOliveGreen,opacity=.3] (v) circle[x radius = 1cm, y radius =.3cm];
	 \path[draw, ->, DarkOliveGreen!80] (u) edge[bend right = 20] node[below,xshift=-.7cm,label={[above]projection}]{\cref{thm:baconseq}} (v');
  \end{tikzpicture}
  \caption{Main contributions}\label{fig:contributions}
\end{figure}

We introduced formal choreography languages as a general and abstract
theory of choreographies based on formal languages.
In this theory we recasted known properties and constructions such as
projections from global to local specifications.

One of our contributions is the characterisation of systems'
correctness in terms of closure under unknown information.
Other communication properties can be ensured by additionally
requiring branch awareness.
A synopsis of our main contributions is in \cref{fig:contributions}.

We demonstrated the versatility of FCL by considering three existing
models.
We showed some relations between FCL and two automata models,
communicating finite-state machines~\cite{bz83} and c-automata~\cite{blt20}.
These models are close to FCL, given that the relations retrace the
well-known connection between automata and formal language theories.
The last model captured with FCL is the variant of MPSTs presented
in~\cite{SeveriD19}.
Being based on behavioural types, this model radically differs from FCL.

\medskip

\noindent\textbf{Future work.}\ 
Our investigation proposes a new point of view for choreography
formalisms and the related constructions.
As such, a number of extensions and improvements need to be analysed,
to check how they may fit in our setting. We list below the most
relevant.

  Our framework considers only point-to-point synchronous
  communications.
  Possible generalisations could contemplate other interaction models
  such as those e.g., in~\cite{cmsy17,wad14} and those based on
  \emph{asynchronous} communications.
  Another possible generalisation is to consider \emph{nominal} FCLs
  in the line of nominal languages~\cite{kst12,bbkl12,kmps15}.
While the general approach should apply, it is not immediate how to
extend CUI in order to characterise correctness for an asynchronous
semantics.
This is somehow confirmed by the results in~\cite{aey05,aey01} on the
realisability of MSCs showing that in the asynchronous setting this is
a challenging problem.

A second direction is analysing how to drop prefix-closure, so
allowing for specifications where the system (and single participants)
may stop their execution at some points but not at others;
a word would hence represent a complete computation, not only a
partial one.

A further direction would unveil the correspondence between closure
properties and subtyping relations used in many multiparty session
types approaches.

Additionally, it would be good to understand whether the composition
and decomposition operators or the partial multiparty sessions defined
in, respectively, \cite{BDLT20} and \cite{Stolze2022}, for multiparty
session types can be rephrased in the more general setting of FCLs.

Another intriguing research direction is the application of FCLs to capture
properties that are not usually considered in behavioural type
frameworks.
Specifically, the notions of \emph{receptiveness} and
\emph{responsiveness} have been defined
in~\cite{10.1007/978-3-319-59746-1_14,10.1007/978-3-030-90870-6_26,10.1007/978-3-030-64276-1_11}
to formalise the properties of communicating systems where some
components can always succeed to send or receive messages.
An interesting initial question to address is whether receptiveness
and responsiveness can be characterised in terms of FCLs.
A starting point to address it could be to give FCL models for the
dynamic logic used in~\cite{tchp23} to characterise receptiveness.

Finally, our approach strives for generality neglecting efficiency and
practical applicability.
An interesting research direction is to identify classes of languages
ensuring that our properties, such as CUI, could be checked
efficiently.
